%% file: paper_biostat_revision.tex
\documentclass[oupdraft]{bio_arxiv}
\input{macros_biostat}


\begin{document}
\title{Quantifying uncertainty in spikes estimated from calcium imaging data}
\author{
YIQUN T. CHEN$^\ast$\\[4pt]
\textit{Department of Biostatistics, University of Washington, Seattle, WA 98195, USA}
\\[2pt]
{yiqunc@uw.edu}\\
SEAN W. JEWELL\\[4pt]
\textit{Department of Statistics, University of Washington, Seattle, WA 98195, USA}
\\[2pt]
DANIELA M. WITTEN\\[4pt]
\textit{Departments of Statistics \& Biostatistics, University of Washington, Seattle, WA 98195, USA} 
\\[2pt]
}
\markboth%
{Y. T. CHEN AND OTHERS}
{Quantifying uncertainty in calcium imaging data}
\maketitle
\footnotetext{To whom correspondence should be addressed.}
\begin{abstract}
{In recent years, a number of methods have been proposed to estimate the times at which a neuron spikes on the basis of calcium imaging data. However,  quantifying the uncertainty associated with these  estimated spikes remains an open problem. We consider a simple and well-studied model for calcium imaging data, which states that calcium decays exponentially in the absence of a spike, and instantaneously increases when a spike occurs. 
We wish to test the null hypothesis that the neuron did \emph{not} spike --- i.e., that there was no increase in calcium --- at a particular timepoint at which a spike was estimated. In this setting, classical hypothesis tests lead to inflated Type I error, because the spike was estimated on the same data used for testing. To overcome this problem, we propose a selective inference approach.  We describe an efficient algorithm to compute finite-sample $p$-values that control selective Type I error, and confidence intervals with correct selective coverage, for spikes estimated using a recent proposal from the literature. We apply our proposal in simulation and on calcium imaging data from the \texttt{spikefinder} challenge.}
 {Calcium imaging; Changepoint detection; Neuroscience; Hypothesis testing; Selective inference }
\end{abstract}
\section{Introduction}
\label{section:intro}
\input{./intro}

\section{Selective inference for spike detection}
\label{section:hypothesis}
\input{./hypothesis}

\section{Computation of the selective $p$-value}
\label{section:method}
\input{./method}

\section{Confidence intervals with correct selective coverage}
\label{section:extension}
\input{./extension}

\section{Simulation study}
\label{section:sim}
\input{./sim}

\section{Application to calcium imaging data}
\label{section:real_data}
\input{./real_data}

\section{Discussion}
\label{section:discussion}
\input{./discussion}

\linespread{1.5}\selectfont
\setlength{\bibsep}{0.0pt}
\bibliographystyle{biorefs}
\bibliography{ref.bib}


\clearpage 
\markboth%
{Y. T. CHEN AND OTHERS}
{Figures}

\input{./figures}

\clearpage

{\centering
\title{\textbf{Quantifying uncertainty in spikes estimated from calcium imaging data \\Supplementary Materials}}

}
\markboth%
{Y. T. CHEN AND OTHERS}
{Supplementary Materials}
\maketitle
\pagenumbering{arabic}

\linespread{1.5}\selectfont

\appendix
\input{./appendix_a_1_2}

\input{./appendix_a_4}

\input{./appendix_a_5}

\input{./appendix_a_6_7}
\input{./appendix_a_8_9}

\input{./appendix_a_10}

\input{./appendix_a_11}

\input{./appendix_a_12}

\input{./appendix_figure}

\end{document}

%% file: macros_biostat.tex
\usepackage{amsmath, amsthm, amssymb, calrsfs, verbatim, bbm, color, graphics,fancybox,setspace}
\usepackage{color,physics}
\usepackage[norelsize,ruled]{algorithm2e}
\usepackage[export]{adjustbox}
\usepackage{bbm}
\usepackage{bm}

\usepackage[htt]{hyphenat}

\usepackage{tabularx}

\usepackage[labelformat=simple]{subcaption}

\usepackage{mathtools}
\makeatletter
\newcommand{\vast}{\bBigg@{4}}
\newcommand{\Vast}{\bBigg@{5}}

\makeatother

\DeclareMathAlphabet{\mathcal}{OMS}{cmsy}{m}{n}

\newcommand{\argmin}[2]{\underset{#1}{\mathrm{argmin}}\left\{ #2 \right\}}
\def\minimize#1#2{  {\underset{#1}{\mathrm{minimize}}}\left\{#2\right\}}

\newcommand{\thj}{\hat\tau_{j}}
\newcommand{\thi}[1]{\hat\tau_{#1}}

\newcommand{\pval}{\ensuremath{p}_{\text{}}}

\newcommand{\cost}{\text{Cost}}

%% file: intro.tex

In the field of neuroscience, recent advances in calcium imaging have enabled recording from large populations of neurons \textit{in vivo}~\citep{Prevedel2014-pq,Ahrens2013-ac,Chen2013-ha}. 
When a neuron spikes,  calcium floods the cell;  the presence of fluorescent calcium indicator molecules  causes it to fluoresce. Thus, for each neuron, calcium imaging results in a time series of fluorescence intensities that can be seen as a noisy approximation to its unobserved spike times. 
Typically, the neuron's observed fluorescence trace is  not of scientific interest; instead, the interest lies in the unobserved spike times.  

A number of methods have been developed to estimate spike times from  the fluorescence trace of a neuron~\citep{Theis2016-wd,Berens2018-wq,vogelstein2010fast,Jewell2018-de,Pachitariu2018-ew,Stringer2019-bi,Jewell2019-jw}. 
One line of work makes use of a simple model that relates the unobserved calcium $c_t$ and the observed fluorescence $Y_t$ at the $t$th time step \citep{vogelstein2010fast,friedrich2016fast,Jewell2018-de,Jewell2019-jw},
{
  \setlength{\belowdisplayskip}{1pt}
\setlength{\abovedisplayskip}{1pt}
\begin{align}
Y_{t} &= c_{t} + \epsilon_{t}, \quad \epsilon_{t} \overset{\text{i.i.d.}}{\sim} \mathcal{N}(0, \sigma^{2}), \quad t=1,\ldots,T, \nonumber \\
c_{t} &= \gamma c_{t-1} + z_{t}, \quad t =2, \ldots, T,
\label{eq:obs-model}
\end{align}
}
where $z_t\geq 0$ for all $t$, and $z_t>0$ indicates the presence of a spike at the $t$th time step. At most time steps, $z_t = 0$, corresponding to no spike. Between spikes, calcium decays exponentially at a rate $\gamma \in (0,1)$; $\gamma$ can be viewed as a property of the calcium indicator, and is taken to be known.  
Model \eqref{eq:obs-model} suggests estimating the underlying calcium $c_t$ by solving the optimization problem
\begin{align}
\minimize{c_1,\ldots, c_T\geq 0; \;\; z_1,\ldots, z_T}{\frac12 \sum_{t = 1}^{T} (y_{t} - c_{t})^{2} + \lambda \sum_{t=2}^{T} 1_{(z_{t}\neq 0)}} \text{ subject to $z_t= c_t-\gamma c_{t-1}\geq 0$ }, 
\label{eq:l0-opt-complete}
\end{align} where $\lambda\geq 0$ is a  tuning parameter that trades off the number of estimated spikes and the fit to the observed fluorescence \citep{Jewell2018-de}. The $\ell_0$ penalty
  $\sum_{t=2}^{T} 1_{(z_{t}\neq 0)}$ is non-convex, which has motivated  a number of authors to consider  a convex relaxation to \eqref{eq:l0-opt-complete} using an $\ell_1$ penalty~\citep{friedrich2016fast,vogelstein2010fast,Friedrich2017-vn}. An efficient dynamic programming algorithm that yields the global optimum to \eqref{eq:l0-opt-complete} has also been proposed \citep{Jewell2018-de,Jewell2019-jw}.

Despite the extensive literature on estimating a neuron's spike times from its fluorescence intensity~\citep{Theis2016-wd,vogelstein2010fast,Jewell2018-de,Pachitariu2018-ew,Jewell2019-jw},  quantifying the uncertainty associated with these estimated spikes remains in large part an open problem. 
More precisely, suppose we observe a $T$-vector of fluorescence intensities under model \eqref{eq:obs-model}, and estimate the $J$ spike times $\hat\tau_1,\ldots,\hat\tau_J$. For fixed $j \in \{1,\ldots,J\}$, consider testing whether there is a spike at $\thj$, i.e.,
{
  \setlength{\belowdisplayskip}{1pt}
\setlength{\abovedisplayskip}{1pt}
\begin{align}
H_{0}: c_{\thj + 1} - \gamma c_{\thj} = 0 \quad \text{versus} \quad  H_{1}: c_{\thj + 1} - \gamma c_{\thj} > 0,
\label{eq:basic-test}
\end{align}
} where the one-sided alternative reflects the fact that a spike  leads to an \emph{increase} (rather than a decrease) in calcium. 
Despite the apparent simplicity of \eqref{eq:basic-test},  obtaining a test with correct size requires care. For instance, motivated by a Wald test, we can consider the $p$-value 
{
  \setlength{\belowdisplayskip}{1pt}
\setlength{\abovedisplayskip}{1pt}
\begin{align}
\label{eq:wald_test}
\mathbb{P}_{H_0}\left( Y_{\thj + 1} - \gamma Y_{\thj}  \geq  y_{\thj + 1} - \gamma y_{\thj}  \right),
\end{align}
} where $y_1,\ldots,y_T$ is the observed fluorescence,
 and \eqref{eq:obs-model} implies that  $Y_{\thj + 1} - \gamma Y_{\thj} \sim \mathcal{N}\left(0, (1+\gamma^2)\sigma^2 \right)$ under $H_0$. 
 But this naive approach ignores the fact that estimation \eqref{eq:l0-opt-complete} and inference \eqref{eq:basic-test} for $\thj$ were performed \emph{on the same data} \citep{Button2019-hz,Fithian2014-ow}. Thus, even in the absence of a true spike, we will observe a large value of $y_{\thj + 1}-\gamma y_{\thj}$; see Figure~\ref{fig:motivation}(a). Figure~\ref{fig:motivation}(b) demonstrates that \eqref{eq:wald_test} does not control the \emph{selective Type I error}:  the probability of a false rejection conditional on the fact that this null hypothesis was tested \citep{Fithian2014-ow}.

In this paper, we leverage the \emph{selective inference} framework, which enables us to test a null hypothesis that was selected using the data, to develop a valid test for  \eqref{eq:basic-test}. Related approaches have been developed for a number of problems, including penalized and stepwise regression \citep{Lee2016-te,Tibshirani2016-bx,Fithian2014-ow}, changepoint detection \citep{Hyun2018-pe,Jewell2019-vv}, and aggregate testing \citep{Heller2018-rn}. In a nutshell, to obtain a test that controls the selective Type I error, we condition  on the aspect of the data that led us to test this particular null hypothesis. In particular, since we have chosen to test the null hypothesis $H_0: c_{\thj + 1} - \gamma c_{\thj} = 0$ in \eqref{eq:basic-test}  because $\thj$ is an estimated changepoint, our $p$-value should be computed \emph{conditional on the event that $\thj$ is an estimated changepoint}. As seen in Figure~\ref{fig:motivation}(c), this results in a test that controls the selective Type I error. 

Some authors have considered quantifying the uncertainty in the location of an estimated spike $\thj$ \citep{Pnevmatikakis2016-xe,Merel2016-dw}. Others have applied a Bayesian lens to the uncertainty associated with the magnitude of the change in calcium associated with an estimated spike $\thj$ \citep{Pnevmatikakis2016-xe,Soltanian-Zadeh2018-ce, Merel2016-dw,Theis2016-wd,Vogelstein2009-mo,Deneux2016-ve}. Despite the flexibility and robustness of Bayesian methods, they do not provide a straightforward way to test \eqref{eq:basic-test}. First, they provide an uncertainty estimate for the change of calcium at \emph{every timepoint}. As a result, we still need to account for selection if we only choose to test the null hypothesis for the estimated spikes \citep{Yekutieli2012-wk}. Second, even with appropriate adjustments, Bayesian hypothesis testing typically will not control Type I error  \citep{Ghosh2011-kb}. 

The current paper is closely related to the literature on changepoint detection. \citet{Jewell2018-de} showed that \eqref{eq:l0-opt-complete} is equivalent to a changepoint detection problem, which allows us to tap into the toolbox of inferential procedures for changepoint detection \citep{Yao1989-iv,Yao1988-vu,Harchaoui2010-cl,Zou2020-dl,Song2016-kh,Fryzlewicz2014-ui}. Despite the abundant literature on this topic, a few gaps remain to be filled, as reviewed in  \citet{Niu2016-es}: (i) much of the prior work has focused on quantifying the uncertainty associated with either the number or locations of the estimated changepoints; and (ii) most existing inferential procedures are asymptotic and approximate. Two recent exceptions include \citet{Hyun2018-pe} and \citet{Jewell2019-vv}, which  took a selective inference approach and computed finite-sample $p$-values for testing the changes in mean around changepoints estimated using an $\ell_1$  and  an $\ell_0$ penalty, respectively. Our work is closest to \citet{Jewell2019-vv}, and extends  their proposal to the model \eqref{eq:obs-model}.

In this paper, we propose a general framework to quantify the uncertainty associated with the set of spikes estimated from calcium imaging data, using \emph{any} spike detection algorithm.  Our testing framework controls the selective Type I error associated with the null hypothesis \eqref{eq:basic-test}. However, in practice it might be very hard to carry out this framework for an arbitrary spike detection algorithm. Thus, in the special case of spikes estimated by solving a variant of the $\ell_0$ optimization problem in \eqref{eq:l0-opt-complete}, we provide an algorithm that can be used to efficiently compute $p$-values and confidence intervals associated with these estimated spikes. 

The rest of this paper is organized as follows. In Section~\ref{section:hypothesis}, we detail the null hypothesis  of interest, and develop a framework to test it for spikes estimated using  \emph{any} spike estimation procedure, under model \eqref{eq:obs-model}. We develop an efficient algorithm to compute the $p$-values for spikes estimated via a variant of \eqref{eq:l0-opt-complete} in Section~\ref{section:method}, and  develop confidence intervals in Section~\ref{section:extension}.  We apply our proposal in a simulation study in Section~\ref{section:sim}, and to calcium imaging data  in Section~\ref{section:real_data}. The discussion is in Section~\ref{section:discussion}. Proofs and other technical details are relegated to the Appendix. 

Throughout this paper, upper case $Y$ denotes a random variable, and lower case $y$ denotes a realization of $Y$. For a vector $\nu \in \mathbb{R}^T$, $\Vert\nu\Vert_2$ denotes its $\ell_2$ norm, $\nu^\top$  its transpose, and $\Pi_\nu^\perp$  the projection matrix onto its orthogonal complement, i.e., $\Pi_\nu^\perp = I-\frac{\nu\nu^{\top}}{\Vert\nu\Vert_2^2}$.  We use $\mathbb{N}$ to denote the natural numbers and $\mathbb{R}$ to denote the real numbers. The notation $1(\cdot)$ and $\overset{d}{=}$ denote an indicator function and equality in distribution, respectively.

%% file: hypothesis.tex

\subsection{Defining the null hypothesis}
\label{subsection:choose_nu}

We wish to test for an increase in calcium at $\hat\tau_j$, an estimated spike time. We re-write
 \eqref{eq:basic-test} as 
 \begin{equation}
  H_0: \nu^\top c=0 \mbox{ versus } H_1: \nu^\top c > 0, 
  \label{eq:null-nu}
  \end{equation}
  where $\nu \in \mathbb{R}^T$ is a contrast vector defined as 
 \begin{align}
\label{eq:nu_example}
\nu_{t} = &
\begin{cases}
 -\gamma, & t = \thj,\\
 1, & t =\thj+1,\\
 0, & \mbox{otherwise}. 
 \end{cases}
\end{align}
However, \eqref{eq:nu_example} only considers the two timepoints immediately before and after $\hat\tau_j$, leaving most data unused. In order to take advantage of a larger data window, we will generalize the contrast vector $\nu$ under a simple assumption.

\noindent \begin{center} \emph{Assumption 1: There are no spikes within a window of $\pm h$ of $\hat\tau_j$. In other words, $ \gamma^{h}c_{\thj-h+1} =  \gamma^{h-1}c_{\thj-h+2}  = \ldots= \gamma c_{\thj}$
 and $c_{\thj+1}=  c_{\thj+2} \gamma^{-1} =\ldots= \gamma^{-h+1}c_{\thj+h}$. } \end{center} 
 Under Assumption 1, and treating  $\thj$ as fixed, the log likelihood of  $Y_{\thj-h+1},\ldots,Y_{\thj}$ is proportional to 
 $\sum_{t = \thj-h+1}^{\thj} \left( Y_t - c_{\thj} \gamma^{t-\thj} \right)^2$.  
  Thus, the maximum likelihood estimator  for $c_{\thj}$ is
 $\hat{c}_{\thj} = \frac{\gamma^2-1}{\gamma^2-\gamma^{-2h+2}}\sum_{t = \thj-h+1}^{\thj}  Y_t \gamma^{t-\thj}.$
 Similarly, using the $h$ observations $Y_{\thj+1},\ldots,Y_{\thj+h}$, the maximum likelihood estimator for ${c}_{\thj+1}$ is
  $\hat{c}_{\thj+1} = \frac{\gamma^2 -1 }{\gamma^{2h}-1} \sum_{t = \thj+1}^{\thj+h}  Y_t \gamma^{t-(\thj+1)} .$
  This suggests that we can test for an increase in calcium at $\thj$ using \eqref{eq:null-nu} with  $\nu$ defined  as
  \begin{align}
\label{eq:nu_def}
\nu_{t} &= 
\begin{cases}
- \frac{\gamma (\gamma^2-1)}{\gamma^2-\gamma^{-2h+2}}    \gamma^{t-\thj}    , & \thj-h+1 \leq  t \leq \thj, \\
\frac{\gamma^2 -1 }{\gamma^{2h}-1} \gamma^{t-(\thj+1)}   , & \thj + 1\leq t \leq \thj+h,\\
 0, & \mbox{otherwise}. \\
\end{cases}
\end{align}
  Details of the form of $\nu$ if $\hat\tau_j + h >T$ or $\hat\tau_j -h+1 < 1$, as well as a visualization of $\nu$ in \eqref{eq:nu_def}, are provided in Appendix~\ref{appendix:general_nu}.

\subsection{A selective test for $H_0: \nu^{\top}c = 0$ versus $H_1: \nu^{\top}c>0$}
\label{sec:selective_test}

   Suppose that we test for an increase in calcium only at timepoints at which (i) we estimate a spike; and (ii) there is an increase in fluorescence associated with this estimated spike. This motivates the 
  following $p$-value to test \eqref{eq:null-nu}:
\begin{align}
 \mathbb{P}_{H_0}\left(\nu^{\top} Y \geq \nu^{\top} y \;\middle\vert\; \thj(y) \in \mathcal{M}(Y) , \nu^{\top}Y > 0 \right),
\label{eq:pval_ideal}
\end{align} where $\mathcal{M}(Y)$ is the set of spikes estimated from $Y$.  
Roughly speaking, this $p$-value answers the question: {\emph{Assuming that there is no true spike at $\thj$, what's the probability of observing such a large increase in fluorescence at $\thj$, given that we decided to test for a spike at $\hat\tau_j$?}}

The $p$-value in \eqref{eq:pval_ideal} controls the \emph{selective Type I error} \citep{Fithian2014-ow}: the probability of \emph{falsely rejecting the null hypothesis, given that the we decided to conduct the test}. However, computing \eqref{eq:pval_ideal} is hard because the conditional distribution of $\nu^{\top}Y $ given $\thj(y) \in \mathcal{M}(Y)$ and $\nu^{\top} Y>0$  depends on the nuisance parameter $\Pi_{\nu}^{\perp} c$. Therefore, we further condition on $\left\{\Pi_{\nu}^{\perp} Y= \Pi_{\nu}^{\perp} y\right\}$ to eliminate the dependence on the nuisance parameter, arriving at the  $p$-value:
\begin{align}
\pval = \mathbb{P}_{H_0}\left(\nu^{\top} Y \geq \nu^{\top} y \;\middle\vert\; \thj(y) \in \mathcal{M}(Y), \nu^{\top}Y > 0, \Pi_{\nu}^{\perp} Y= \Pi_{\nu}^{\perp} y \right).
\label{eq:pval}
\end{align}
Following arguments in Section 5 of \citet{Lee2016-te},  \eqref{eq:pval} controls the selective Type I error. 
This $p$-value is the focus of this paper. 
\begin{Proposition}
\label{prop:pval}
Suppose that $Y \sim \mathcal{N}(c,\sigma^2 I)$. Then,
{
\begin{equation}
\begin{aligned}
\label{eq:single_param_general}
&\mathbb{P}\left(\nu^{\top} Y \geq \nu^{\top} y \;\middle\vert\; \thj(y) \in \mathcal{M}(Y), \nu^{\top}Y > 0, \Pi_{\nu}^{\perp} Y= \Pi_{\nu}^{\perp} y \right)\\
&=\mathbb{P} \left( \phi \geq \nu^{\top}y \;\middle\vert\; \thj(y) \in \mathcal{M}(y'(\phi)) , \phi>0 \right),
\end{aligned}
\end{equation}
}
for $\phi\sim \mathcal{N}\qty(\nu^\top c,\sigma^2||\nu||_2^2)$, where
{
\begin{align}
\label{eq:phi}
y'(\phi) = \Pi_\nu^\perp y+\phi\cdot \frac{\nu}{||\nu||_2^2} = y+\left(\frac{\phi-\nu^{\top}y}{||\nu||_2^2}\right)\nu. 
\end{align}
}
Furthermore, for $\pval$ defined in \eqref{eq:pval}, and $\phi_0 \sim  \mathcal{N}\qty(0,\sigma^2||\nu||_2^2)$,
{
\begin{align}
\label{eq:single_param}
\pval = \mathbb{P} \left( \phi_0 \geq \nu^{\top}y \;\middle\vert\; \thj(y) \in \mathcal{M}(y'(\phi_0)) , \phi_0>0 \right).
\end{align}}
\end{Proposition} 
It follows that to compute the $p$-value in \eqref{eq:pval}, we must characterize the set 
{
\setlength{\belowdisplayskip}{1pt}
\setlength{\abovedisplayskip}{1pt}
\begin{align}
\label{eq:S_set}
\mathcal{S} = \left\{ \phi: \thj \in \mathcal{M}(y'(\phi)) \right\}.
\end{align} 
}
Of course, the practical details of computing the  set \eqref{eq:S_set} will depend on the function  $\mathcal{M}(\cdot)$ that yields the estimated spikes. The task of characterizing the set \eqref{eq:S_set} is the focus of Section~\ref{section:method}.

In \eqref{eq:phi}, $y'(\phi)$ results from perturbing  $y$ by a function of $\phi$ along the direction defined by $\nu$. Elements of $y$ that fall outside of the support of $\nu$ are not perturbed. Then,  $\mathcal{S}$ in \eqref{eq:S_set}  is the set of $\phi$ such that applying  $\mathcal{M}(\cdot)$ to the perturbed data $y'(\phi)$ results in an estimated spike at $\thj$.

As an example, we generate data from  \eqref{eq:obs-model} with $T=80$, $\sigma = 0.1$, and $\gamma=0.98$ with a true spike at $t=40$, and $c_{41} - \gamma c_{40} = 1$. This results in $\phi=\nu^{\top} y = 1.02$.  
  Solving the optimization problem in \eqref{eq:l0-opt-complete} with $\lambda=0.75$ results in a single estimated spike at $t=40$, which means that $\mathcal{S}=\{ \phi: 40 \in \mathcal{M}(y'(\phi)) \}$. The set-up is displayed in Figure \ref{fig:perturbation}(a). In panel (b), we perturb the observed data with $\phi=0$. Now a spike is no longer estimated at $t=40$, so $0 \notin \mathcal{S}$. In panel (c), we perturb the observed data  with $\phi=2$ to exaggerate the increase in fluorescence; now a spike is estimated at  $t=40$, so $2 \in \mathcal{S}$.  In panel (d), we display the set $\mathcal{S} \cap (0, +\infty)=  (0.29, +\infty)$.

%% file: method.tex

Proposition~\ref{prop:pval} indicates that we can compute the $p$-value defined in \eqref{eq:pval} provided that we are able to compute 
the set $\mathcal{S}$ defined in \eqref{eq:S_set}. 
In this section, we will show that $\mathcal{S}$ can be efficiently computed for  spikes estimated by solving a variant of the $\ell_0$ optimization problem in \eqref{eq:l0-opt-complete} that omits the positivity constraint $c_t -  \gamma  c_{t-1} \geq 0$: namely,
\begin{align}
\minimize{c_1,\ldots, c_T\geq 0}{\frac12 \sum_{t = 1}^{T} (y_{t} - c_{t})^{2} + \lambda \sum_{t=2}^{T} 1_{(c_{t}\neq \gamma c_{t-1})}}.
\label{eq:l0-opt}
\end{align}
In Section~\ref{sec:review_dp}, we briefly review the work of \citet{Jewell2018-de} and \citet{Jewell2019-jw}, who showed that the solution to \eqref{eq:l0-opt} can be characterized through a recursion involving piecewise quadratic functions. 
The rest of this section is quite technical. An overview  is as follows: 
\begin{itemize}
	  \setlength\itemsep{1pt}
	\item
   In Section~\ref{sec:p_val_calc}, we introduce functions $C(\phi)$ and $C'(\phi)$ such that $\mathcal{S} = \{ \phi: C(\phi) \leq C'(\phi) \}$.
   \item
   Then, in Section~\ref{sec:comp}, we show that $C(\phi)$ and $C'(\phi)$ are piecewise quadratic in $\phi$.
   \item
   We can therefore apply approaches from \citet{Rigaill2015-pm} and \citet{Maidstone2017-vc} for efficient manipulation of piecewise quadratic functions, to efficiently compute $\{ \phi: C(\phi) \leq C'(\phi) \}$, and in turn, $\mathcal{S}$ in \eqref{eq:S_set}.
\end{itemize}

\subsection{An algorithm to solve \eqref{eq:l0-opt}}
\label{sec:review_dp}
\citet{Jewell2018-de} noted that  \eqref{eq:l0-opt} 
 is equivalent to a changepoint detection problem,
\begin{align}
\minimize{0=\tau_0<\tau_1<\ldots<\tau_k<\tau_{k+1}=T, k}{\sum_{j=0}^k \min_{\alpha\geq 0} \left\{\frac{1}{2} \sum_{t=\tau_j+1}^{\tau_{j+1}} (y_t-\alpha \gamma^{t-\tau_{j+1}})^2 \right\} +\lambda k},
\label{eq:changepoint_form}
\end{align} 
in the sense that $\{ t: \hat{c}_{t+1}-\gamma \hat{c}_t \neq 0\} = \{\hat{\tau}_1,\ldots , \hat{\tau}_{J} \}$, where $\hat{c}_1,\ldots, \hat{c}_T$ and $\hat\tau_1,\ldots, \hat\tau_J, J$ are solutions to \eqref{eq:l0-opt} and \eqref{eq:changepoint_form}, respectively. 
Furthermore, let $F(s)$ denote the optimal objective of \eqref{eq:changepoint_form} for the first $s$ data points $y_{1:s} = (y_1,\ldots, y_s)$, and define
\begin{align}
\label{eq:def_cost}
\cost\qty(y_{1:s},\alpha;\gamma) = \min_{0\leq\tau <s} \left \{ F(\tau) + 
 \left\{ \frac{1}{2} \sum_{t=\tau+1}^s \qty(y_t - \alpha \gamma^{t-s})^2 \right\} + \lambda \right\}.
\end{align}
In words, $\cost\qty(y_{1:s},\alpha;\gamma)$ is the optimal cost of partitioning the data $y_{1:s}$ into exponentially decaying regions with decay parameter $\gamma$, given that the calcium at the $s$th timepoint equals $\alpha$. It turns out that $\cost\qty(y_{1:s},\alpha;\gamma)$  admits a recursion that can be solved efficiently, which provides intuition for characterizing the set $\mathcal{S}$ in the next section.

\begin{Proposition}[Proposition 1 and Section 2.2.3 in \citet{Jewell2019-jw}]
\label{prop:cost_recursion}
For $\cost\qty(y_{1:s},\alpha;\gamma)$ defined in \eqref{eq:def_cost}, the following recursion holds:
\begin{equation}
\label{eq:cost_recursion}
\small {\cost}\qty(y_{1:s},\alpha;\gamma) = \min \left\{ {\cost}\qty(y_{1:(s-1)},\alpha/\gamma;\gamma), \min_{\alpha'\geq 0}{\cost}\qty(y_{1:(s-1)},\alpha';\gamma)+\lambda \right\} + \frac{1}{2}\qty(y_s-\alpha)^2,
\end{equation}
\normalsize with ${\cost}\qty(y_{1},\alpha;\gamma) = \frac{1}{2}\qty(y_1 - \alpha)^2$. Also, ${\cost}\qty(y_{1:s},\alpha;\gamma)$ is a piecewise quadratic function of $\alpha$.
\end{Proposition}

In words, the recursion in \eqref{eq:cost_recursion} considers the following two possibilities: (i) there is no spike at the $(s-1)$th time point, in which case the calcium decays exponentially, and the cost equals $\cost\qty(y_{1:(s-1)},\alpha/\gamma;\gamma)$; (ii) there is a spike at the $(s-1)$th time point, and the cost equals the optimal cost up to $s-1$, $\min_{\alpha'\geq 0}\cost\qty(y_{1:(s-1)},\alpha';\gamma)$, plus the cost of placing a changepoint, $\lambda$.

Building on Proposition~\ref{prop:cost_recursion}, \citet{Jewell2019-jw} made use of the recent literature on \emph{functional pruning}  \citep{Maidstone2017-vc,Rigaill2015-pm} to efficiently compute the cost functions $\cost(y_{1:s},\alpha;\gamma)$, as a function of $\alpha$, using clever manipulations of the piecewise quadratic functions involved in the recursion \eqref{eq:cost_recursion}. This approach has a worst-case complexity of
 $O(s^2)$, and is often much faster in practice. Once the cost functions have been computed, it is straightforward to identify the changepoints in \eqref{eq:changepoint_form}, and, in turn, the spikes in \eqref{eq:l0-opt}. Details are provided in Section 2.2 of \citet{Jewell2019-jw}.

\subsection{Characterizing $\mathcal{S}$ for spikes estimated using \eqref{eq:l0-opt}}
\label{sec:p_val_calc}

In what follows, we leverage ideas from \citet{Jewell2019-vv} to develop an efficient algorithm to analytically characterize \eqref{eq:S_set}, i.e., the set of values $\phi$ such that solving \eqref{eq:l0-opt} on perturbed data $y'(\phi)$  yields an estimated spike $\thj$. Throughout this section, we define  $y_{1:s} = \qty(y_1,\ldots,y_s)$, $y_{T:s} = \qty(y_T,\ldots,y_s)$,  $y'_{1:s}(\phi) = \qty([y'(\phi)]_1,\ldots,[y'(\phi)]_s)$, and $y'_{T:s}(\phi) = \qty([y'(\phi)]_T,\ldots,[y'(\phi)]_s)$.

Let $\mathcal{M}(y)$ denote the spikes estimated by applying \eqref{eq:l0-opt} to the data $y$. To begin, we characterize the set $\mathcal{S}$ using the $\cost\qty(y_{1:s},\alpha;\gamma)$ function defined in \eqref{eq:def_cost}. 

\begin{Proposition}
\label{prop:characterization_S}
Let $\{\hat\tau_1,\ldots,\hat\tau_J \} = \left\{ t: \hat{c}_{t+1}-\gamma\hat{c}_{t}\neq 0\right\}$ be the timesteps of the estimated spikes from \eqref{eq:l0-opt}. 
For ${\cost}\qty(y_{1:s},\alpha;\gamma)$ in \eqref{eq:def_cost}, we have that
 {
  \setlength{\belowdisplayskip}{1pt}
\setlength{\abovedisplayskip}{1pt}
\begin{align}
\label{eq:cost_spike_thj}
C(\phi) &= \min_{\alpha\geq 0} \left\{ {\cost}\qty(y'_{1:\thj}(\phi),\alpha;\gamma)\right\} + \min_{\alpha'\geq 0} \left\{ {\cost}\qty(y'_{T:(\thj+1)}(\phi),\alpha';1/\gamma)\right\} + \lambda \, 
\end{align} 
}
equals the objective of \eqref{eq:changepoint_form} applied to data $y'(\phi)$, subject to the constraint that $\hat\tau_j$ \emph{is} an estimated spike. Furthermore,
{
  \setlength{\belowdisplayskip}{1pt}
\setlength{\abovedisplayskip}{1pt} 
 \begin{align}
\label{eq:cost_no_spike_thj}
C'(\phi) &= \min_{\alpha\geq 0} \left\{ {\cost}\qty(y'_{1:\thj}(\phi),\alpha;\gamma) +  {\cost}\qty(y'_{T:(\thj+1)}(\phi),\gamma \alpha; 1/\gamma )\right\}
\end{align} 
}
 equals the objective of \eqref{eq:changepoint_form} applied to data $y'(\phi)$, subject to the constraint that $\hat\tau_j$ \emph{is not} an estimated spike.
Moreover, for $\mathcal{S}$ defined in \eqref{eq:S_set}, 
{\setlength{\belowdisplayskip}{1pt}
\setlength{\abovedisplayskip}{1pt}
\begin{equation}
\label{eq:S_in_phi}
\mathcal{S}  = \left\{ \phi:C(\phi) \leq C'(\phi) \right\}.
\end{equation} 
}
\end{Proposition} 
Therefore, to characterize $\mathcal{S}$ in \eqref{eq:S_set}, it suffices to characterize $C(\phi)$ in \eqref{eq:cost_spike_thj} and $C'(\phi)$ in \eqref{eq:cost_no_spike_thj}. To do this, we will leverage the toolkit from \citet{Jewell2019-vv} to analytically characterize $\cost(y_{1:s}'(\phi),\alpha;\gamma)$ as a function of both $\phi$ and $\alpha$. While this is related to the task of efficiently characterizing $\cost(y_{1:s},\alpha;\gamma)$ in terms of $\alpha$ in Section \ref{sec:review_dp}, it is substantially more challenging, due to the presence of the additional parameter $\phi$.

\subsection{Efficient computation of $\mathcal{S}$ via $\emph{\cost}\qty(y_{1:s}'(\phi),\alpha;\gamma)$}\label{sec:comp}

While Proposition~\ref{prop:cost_recursion} cannot be directly applied to $\cost\qty(y_{1:s}'(\phi),\alpha;\gamma)$, we can arrive at a very similar result by adapting Theorem 2 from \citet{Jewell2019-vv}. 

\begin{Proposition}
\label{prop:bivariate_cost_recursion}
For $\thj-h+1 \leq s \leq \thj$ and $y'(\phi)$ defined in \eqref{eq:phi},  
{
\setlength{\belowdisplayskip}{1pt}
\setlength{\abovedisplayskip}{1pt}
\begin{align}
{\cost}\qty(y_{1:s}'(\phi),\alpha;\gamma) = \min_{f\in\mathcal{C}_s} f(\alpha,\phi),
\label{eq:forward_fpop}
\end{align}
}
where $\mathcal{C}_s$ is a collection of $s-\thj+h+1$ piecewise quadratic functions of $\alpha$ and $\phi$ constructed with the initialization  
{
\setlength{\belowdisplayskip}{1pt}
\setlength{\abovedisplayskip}{1pt}
\begin{align}
\mathcal{C}_{\thj-h} = \left\{ {\cost}(y_{1:(\thj-h)}'(\phi),\alpha;\gamma) \right\},
\label{eq:recursion_init}
\end{align} 
}
and the recursion
{
\setlength{\belowdisplayskip}{1pt}
\setlength{\abovedisplayskip}{1pt}
\begin{align}
\label{eq:union_update_bivariate}
\mathcal{C}_s = \left( \bigcup_{f \in \mathcal{C}_{s-1}} \left\{ f(\alpha/\gamma,\phi)+\frac{1}{2}(y'_s(\phi)-\alpha)^2  \right\} \right) \bigcup \left\{ g_s(\phi)+\frac{1}{2}(y'_s(\phi)-\alpha)^2  \right\} \, ,
\end{align} 
}
where  
{
\setlength{\belowdisplayskip}{1pt}
\setlength{\abovedisplayskip}{1pt}
\begin{align}
g_s(\phi) = \min_{f\in\mathcal{C}_{s-1}}\min_{\alpha\geq 0} f(\alpha,\phi) +\lambda.
\label{eq:g_update}
\end{align}
}
\end{Proposition}
Proposition~\ref{prop:bivariate_cost_recursion} applies when $\thj - h \geq 1$; Appendix \ref{appendix:general_tau_L_tau_R} details the extension for $\thj - h < 1$. Proposition~\ref{prop:bivariate_cost_recursion} indicates that $\cost\qty(y_{1:s}'(\phi),\alpha;\gamma)$ is in fact a bivariate piecewise quadratic function of both $\phi$ and $\alpha$ (in contrast to a univariate piecewise quadratic function of $\alpha$, as in $\cost\qty(y_{1:s},\alpha;\gamma)$). Moreover, $\cost\qty(y_{1:s}'(\phi),\alpha;\gamma)$ can be efficiently computed with the recursion in \eqref{eq:union_update_bivariate}. 

To compute $C(\phi)$ in \eqref{eq:cost_spike_thj}, we first use Proposition~\ref{prop:bivariate_cost_recursion}  to compute the collection $\mathcal{C}_{\thj}$ such that $\cost\qty(y'_{1:\thj}(\phi),\alpha;\gamma) = \min_{f \in \mathcal{C}_{\thj}} f(\alpha,\phi)$. Using a slight modification of Proposition~\ref{prop:bivariate_cost_recursion}  (see Proposition \ref{prop:bivariate_cost_recursion_reverse}  in Appendix \ref{appendix:recursion_reverse}),  we also compute the collection $\tilde{\mathcal{C}}_{\thj+1}$ such that  $\cost\qty(y'_{T:(\thj+1)}(\phi),\alpha';1/\gamma) = \min_{f \in \tilde{\mathcal{C}}_{\thj+1}} f(\alpha',\phi)$.
Then, we have that 
{
\setlength{\belowdisplayskip}{1pt}
\setlength{\abovedisplayskip}{1pt}
\begin{equation}
\begin{aligned}
C(\phi) &\overset{a.}{=} \min_{\alpha\geq 0} \left\{ \min_{f \in \mathcal{C}_{\thj}} f(\alpha,\phi) \right\} +  \min_{\alpha' \geq 0} \left\{ \min_{f \in \tilde{\mathcal{C}}_{\thj+1}} f(\alpha' ,\phi) \right\} + \lambda \,  \\
 &\overset{b.}{=}  \min_{f \in \mathcal{C}_{\thj}} \left\{ \min_{\alpha\geq 0}f(\alpha,\phi) \right\} +  \min_{f \in \tilde{\mathcal{C}}_{\thj+1}}  \left\{\min_{\alpha' \geq 0}  f(\alpha' ,\phi) \right\} + \lambda.
 \end{aligned}
\label{eq:Cphi}
\end{equation}
}
Here, $a.$ follows from combining the definition of $C(\phi)$ in \eqref{eq:cost_spike_thj} with the expression for ${\cost}\qty(y_{1:s}'(\phi),\alpha;\gamma)$ in \eqref{eq:forward_fpop} and the expression for ${\cost}\qty(y_{T:s}'(\phi),\alpha;1/\gamma)$ in Appendix~\ref{appendix:recursion_reverse}; $b.$ follows from changing the order of minimizations. Furthermore, since Proposition \ref{prop:bivariate_cost_recursion} states that the functions in $\mathcal{C}_{\thj}$ are piecewise quadratic in $\phi$ and $\alpha$, it follows that $\min_{\alpha \geq 0} f(\alpha ,\phi)$ is a piecewise quadratic function of $\phi$ only. A similar result in Appendix~\ref{appendix:recursion_reverse} guarantees that the functions in $\tilde{\mathcal{C}}_{\thj+1}$ are piecewise quadratic in $\phi$ and $\alpha$; therefore, for each $f\in\tilde{\mathcal{C}}_{\thj+1}$, we have that $\min_{\alpha' \geq 0}  f(\alpha' ,\phi)$ is piecewise quadratic in $\phi$. Because minimization and summation over piecewise quadratic functions yields a piecewise quadratic function, it follows that $C(\phi)$ is piecewise quadratic in $\phi$.

We now consider computing $C'(\phi)$ in \eqref{eq:cost_no_spike_thj}. Plugging in the expressions for ${\cost}\qty(y_{1:s}'(\phi),\alpha;\gamma)$ in \eqref{eq:forward_fpop} and ${\cost}\qty(y_{T:s}'(\phi),\alpha;1/\gamma)$ in Appendix~\ref{appendix:recursion_reverse} into \eqref{eq:cost_no_spike_thj}, we have 
{
\setlength{\belowdisplayskip}{1pt}
\setlength{\abovedisplayskip}{1pt}
\begin{equation}
\begin{aligned}
C'(\phi) &= \min_{\alpha\geq 0} \qty{  \min_{f \in \mathcal{C}_{\thj}} f(\alpha,\phi) +  \min_{f \in \tilde{\mathcal{C}}_{\thj+1}} f(\gamma\alpha,\phi)  } \\
&= \min_{\alpha\geq 0} \qty{ \min_{ f\in  \mathcal{C}_{\thj}, \tilde{f}\in\tilde{\mathcal{C}}_{\thj+1}}\qty{  f(\alpha,\phi)+\tilde{f}(\gamma\alpha,\phi) } }  \\
&= \min_{ f\in  \mathcal{C}_{\thj}, \tilde{f}\in\tilde{\mathcal{C}}_{\thj+1}} \qty{ \min_{\alpha\geq 0} \qty{  f(\alpha,\phi)+\tilde{f}(\gamma\alpha,\phi) } } .
\end{aligned}
\label{eq:Cphi_prime}
\end{equation} 
}
By Proposition \ref{prop:bivariate_cost_recursion} and Appendix \ref{appendix:recursion_reverse}, both $ f(\alpha,\phi)$ and $\tilde{f}(\gamma\alpha,\phi)$ are piecewise quadratic in $\alpha$ and $\phi$, which implies that $\min_{\alpha\geq 0}  \qty{  f(\alpha,\phi)+\tilde{f}(\gamma\alpha,\phi) } $ is a piecewise quadratic function of $\phi$. Therefore, $C'(\phi)$ is the minimum over a set of piecewise quadratic functions of $\phi$, and thus is itself piecewise quadratic in $\phi$.

Finally, since both $C(\phi)$ and $C'(\phi)$ are piecewise quadratic in $\phi$, we can apply ideas from the functional pruning literature to compute the set $\mathcal{S}=\{\phi: C(\phi)\leq C'(\phi)  \}$ efficiently \citep{Maidstone2017-vc,Rigaill2015-pm}. The procedure and computation time are summarized in Algorithm~\ref{alg:S_computation} (see Appendix~\ref{appendix:algorithm_s}) and Proposition~\ref{prop:S_set_timing}, respectively.

\begin{Proposition}
\label{prop:S_set_timing}
Once ${\cost}\qty(y_{1:(\thj-h)},\alpha;\gamma)$ and ${\cost}\qty(y_{T:(\thj+h+1)},\alpha;1/\gamma)$ have been computed, Algorithm~\ref{alg:S_computation} can be performed in $O(h^2)$ operations.
\end{Proposition}
The worst-case complexity of computing $\cost\qty(y_{1:(\thj-h)},\alpha;\gamma)$ and $\cost\qty(y_{T:(\thj+h+1)},\alpha;1/\gamma)$ is $O(T^2)$, but it is often much faster in practice \citep{Jewell2019-jw}. Furthermore, ${\cost}(y_{1:(\thj-h)},\alpha;\gamma)$  was already  computed to solve \eqref{eq:l0-opt}. Therefore, estimating $J$ changepoints via  \eqref{eq:l0-opt} and then computing their corresponding $p$-values has a worst-case computation time of $O(T^2 + J h^2)$, and is often much faster in practice. An empirical analysis of the timing complexity of Algorithm~\ref{alg:S_computation} can be found in Appendix~\ref{appendix:empirical_computation}. We walk through Algorithm~\ref{alg:S_computation} on a small example in Appendix~\ref{appendix:example}.

%% file: extension.tex

\label{sec:conf_int}
 We now  construct a $(1-\alpha)$ confidence interval for $\nu^{\top} c$, the change in calcium associated with an estimated spike $\thj$. 
\begin{Proposition}
\label{prop:ci}
Suppose that \eqref{eq:obs-model} holds, and let $\thj$ denote a spike estimated by solving  \eqref{eq:l0-opt}.  
 For a given value of $\alpha\in(0,1)$, define functions $\theta_L(t)$ and $\theta_U(t)$ such that
 {
 \setlength{\belowdisplayskip}{1pt}
\setlength{\abovedisplayskip}{1pt}
\begin{align}
F_{\theta_L(t),\sigma^2||\nu||_2^2}^{\mathcal{S}\cap(0,\infty)}(t) = 1-\frac{\alpha}{2}, \quad  F_{\theta_U(t),\sigma^2||\nu||_2^2}^{\mathcal{S}\cap(0,\infty)}(t) = \frac{\alpha}{2},
\label{eq:LCB_UCB}
\end{align}
 }
  where $F_{\mu,\sigma^2}^{\mathcal{S}\cap(0,\infty)}(t)$ is the cumulative distribution function of a normal distribution with mean $\mu$ and variance $\sigma^2$, truncated to the set $\mathcal{S}\cap(0,\infty)$. 
Then  $[\theta_L(\nu^{\top} Y),\theta_U(\nu^{\top} Y)]$ is a $(1-\alpha)$ confidence interval for $\nu^\top c$, in the sense that 
 {
 \setlength{\belowdisplayskip}{1pt}
\setlength{\abovedisplayskip}{1pt}
 \begin{align}
 \mathbb{P}\Big(\nu^{\top}c \in \left[\theta_L(\nu^{\top} Y),\theta_U(\nu^{\top} Y)\right]  \;\Big\vert\; \thj(y) \in \mathcal{M}(Y), \nu^\top Y  > 0, \Pi_{\nu}^{\perp} Y= \Pi_{\nu}^{\perp} y  \Big) = 1-\alpha.
 \end{align}
  }
\end{Proposition}
Thus, the confidence interval guarantees coverage \emph{conditional on the selection procedure} \citep{Lee2016-te,Fithian2014-ow,Tibshirani2016-bx}. Computing $\theta_L$ (and $\theta_U$) in \eqref{eq:LCB_UCB} amounts to a root-finding problem, which can be solved, e.g., using bisection.

%% file: sim.tex
Recall that our selective inference framework involves testing the null hypothesis of no increase in calcium at timepoints for which the following two conditions hold: (i) 
this timepoint was an estimated spike in the solution to \eqref{eq:l0-opt}; and 
(ii) $\nu^\top y > 0$ for this particular timepoint.  
We let $\left\{ \tilde\tau_1,\ldots,\tilde\tau_M \right\}$
denote the set of timepoints satisfying these two conditions, i.e., the set of timepoints to be tested using our selective inference approach.
That is, 
 {
 \setlength{\belowdisplayskip}{1pt}
\setlength{\abovedisplayskip}{1pt}
\begin{align}
 \left\{ \tilde\tau_1,\ldots,\tilde\tau_M \right\}= \left\{ \hat\tau_1,\ldots,\hat\tau_J: \nu^\top y>0 \right\},
 \label{eq:tautilde}
 \end{align} 
 }
where $\left\{ \hat\tau_1,\ldots,\hat\tau_J \right\}$ denotes the set of spikes estimated from \eqref{eq:l0-opt}. \eqref{eq:tautilde} slightly abuses notation, since $\nu$  in \eqref{eq:nu_def} is a function of $\thj$. 
Therefore, the right-hand side of \eqref{eq:tautilde} should be interpreted as the estimated spike times associated with an increase in fluorescence in a window of $\pm h$.  
\subsection{Selective Type I error control under the global null}
\label{sec:type_I_error}
We simulated $y_1,\hdots, y_{10,000}$ according to \eqref{eq:obs-model} with $\gamma = 0.98$, $\sigma = 0.2$, and $z_t = 0$ for all $t=2,\ldots,10,000$. Thus, the null hypothesis $H_0:\nu^{\top} c = 0$ holds for all contrast vectors $\nu$ defined in \eqref{eq:nu_def}, regardless of the timepoint being tested, and the value of $h$ in \eqref{eq:nu_def}. 

We solved \eqref{eq:l0-opt} with the tuning parameter $\lambda$ selected to yield $J=100$ estimated spikes; thus, $J=100$ in \eqref{eq:tautilde}.
 Then, for each $\hat\tau_j$, 
 we constructed four contrast vectors $\nu$, defined in \eqref{eq:nu_def}, corresponding to $h\in\{1,2,10,20\}$. Then, provided that $\nu^\top y > 0$,  we computed the selective $p$-values in \eqref{eq:pval} and the naive (Wald) $p$-values defined as
 {
\setlength{\belowdisplayskip}{1pt}
\setlength{\abovedisplayskip}{1pt}
\begin{align}
\label{eq:wald_p_val}
\mathbb{P}\qty(\nu^{\top}Y \geq \nu^{\top}y).
\end{align}
}
The results, aggregated over 1,000 simulations, are displayed in Figure \ref{fig:Type_I_power}. Panels (a) and (b) display quantile-quantile plots of the naive and selective $p$-value quantiles versus the Uniform(0,1) quantiles, respectively; we see that for all values of $h$,  (i) the naive procedure in \eqref{eq:wald_p_val} is anti-conservative; and (ii) the proposed selective test in \eqref{eq:pval} controls the selective Type I error. 

\subsection{Power and detection probability}
\label{sec:power_result}

Recall that we test $H_0: \nu^\top c=0$ only for timepoints in the set $\left\{ \tilde\tau_1,\ldots,\tilde\tau_M \right\}$ defined in \eqref{eq:tautilde}. Therefore, we separately consider the \emph{conditional power} of the proposed test \citep{Jewell2019-vv,Hyun2018-pe} and the \emph{detection probability} of the spike estimation procedure. 

Given a dataset $y = (y_1,\ldots,y_T)$ with $K$ true spikes $\tau_1,\ldots,\tau_K$, and recalling the definition in \eqref{eq:tautilde}, we define the \emph{conditional power} to be the ratio between (i) the number of true spikes for which the nearest  null hypothesis among those tested (i.e., the set $\{\tilde\tau_1,\ldots,\tilde\tau_M\}$ in \eqref{eq:tautilde}) is within $b$ timepoints of the true spike \emph{and} has a $p$-value less than $\alpha$; and (ii) the number of true spikes for which the nearest tested hypothesis falls within $b$ timepoints. That is, 
{
\setlength{\belowdisplayskip}{1.5pt}
\setlength{\abovedisplayskip}{1.5pt}
\begin{align}
\label{eq:cond}
\text{Conditional power} = \frac{\sum_{i=1}^K 1\left( p_{m(i)}\leq \alpha,| \tau_i - \tilde{\tau}_{m(i)} | \leq b  \right) }{ \sum_{i=1}^K 1\left( | \tau_i - \tilde{\tau}_{m(i)} | \leq b \right) },
\end{align} 
}
where  $m(i) = \argmin{m}{|\tau_i - \tilde{\tau}_m|}$ indexes the 
 timepoint to be tested that is closest to the $i$th true spike time,  and $p_{m(i)}$ is the corresponding $p$-value. 
Since \eqref{eq:cond} conditions on the event that the closest tested timepoint $\tilde{\tau}_{m(i)}$ is within $b$ timepoints of the true spike time $\tau_i$, we also consider the \emph{detection probability}, which tells us how often this event occurs:
{
\setlength{\belowdisplayskip}{1.5pt}
\setlength{\abovedisplayskip}{1.5pt}
\begin{equation}
\label{eq:detect}
\text{Detection probability} = \frac{\sum_{i=1}^K 1\left( | \tau_i - \tilde{\tau}_{m(i)}| \leq b \right)}{K}.
\end{equation}
}
We evaluate the detection probability and conditional power on data generated from \eqref{eq:obs-model} with $T=10,000$, $\gamma = 0.98$, $z_t \overset{\text{i.i.d.}}{\sim} \text{Poisson}(0.01)$ for all $t=2,\ldots,T$, and $\sigma \in \{1,2,\hdots, 10\}$. In \eqref{eq:l0-opt}, $\lambda$ is chosen to yield $J=100$ estimated spikes, i.e., $J=100$ in \eqref{eq:tautilde}; this is the expected number of spikes in this simulation. We generate 500 datasets, and consider $h \in \{1,2,10,20\}$ in \eqref{eq:nu_def}. Results with $\alpha=0.05$ and $b=2$ are displayed in Figure \ref{fig:Type_I_power}. Panels (c) and (d) display the detection probability and conditional power, respectively. 
Both quantities increase as $1/\sigma$ increases. Interpreting the relationship between conditional power and $h$ requires more care: larger values of $h$ typically give rise to higher conditional power for the same value of $\sigma$. 
However, the null hypothesis in \eqref{eq:null-nu} changes as a function of $h$, and it may be the case that $H_0$ holds for a smaller value of $h$, but not for a larger value.

\subsection{Confidence interval coverage and width}
\label{sec:CI_result}
We now generate data from \eqref{eq:obs-model} with $T=10,000$, $\gamma = 0.98$, $z_t \overset{\text{i.i.d.}}{\sim} \text{Poisson}(0.01)$ for $t=2,\ldots,T$, and $\sigma \in \{1,2,\ldots, 6\}$. The tuning parameter $\lambda$ in  \eqref{eq:l0-opt} is chosen to yield $J=100$ estimated spikes, i.e., $J=100$ in \eqref{eq:tautilde}. 
For each timepoint $\tilde\tau_m$ in \eqref{eq:tautilde},  we construct 95\% selective confidence intervals $\qty[\theta_L \qty(\nu^\top y), \theta_U \qty(\nu^\top y)]$ for the parameter $\nu^\top c$,
 with $h\in \{1,2,10,20\}$. As a comparison, we also construct 95\% naive (Wald) confidence intervals for $\nu^\top c$, 
 {
  \setlength{\belowdisplayskip}{1pt}
\setlength{\abovedisplayskip}{1pt}
\begin{equation}
\qty[\nu^\top y -1.96 \sigma ||\nu||_2 ,  \nu^\top y + 1.96 \sigma ||\nu||_2 ],
\label{eq:naiveci}
\end{equation} 
}
which do not account for the fact that we decided to test $H_0: \nu^\top c=0$ after looking at the data. 


Suppose that we construct $M$ confidence intervals (see \eqref{eq:tautilde} for the definition of $M$), we define their coverage, average width, and average midpoint relative to the value of $\nu^\top y$, as follows:
{
  \setlength{\belowdisplayskip}{1em}
\setlength{\abovedisplayskip}{1em}
\begin{align}
\label{eq:coverage}
\text{Coverage} &=  \frac{1}{M}\sum_{m=1}^M 1\left( \nu^{\top}c \in \left[\theta_L\qty(\nu^\top y), \theta_U\qty(\nu^\top y) \right] \right) ,\\
 \label{eq:width}
 \text{Width} &=  \frac{1}{M}\sum_{m=1}^M \left( \theta_U\qty(\nu^\top y)-\theta_L \qty(\nu^\top y)\right),\\
\label{eq:skew}
\text{Midpoint}&=  \frac{1}{M}\sum_{m=1}^M \left( \frac{\theta_L\qty(\nu^\top y)+\theta_U\qty(\nu^\top y)}{2}-\nu^{\top}y \right).
\end{align}
}
There is a slight abuse of notation in \eqref{eq:coverage}--\eqref{eq:skew}, since  $\nu$ is a function of $\tilde\tau_m$ (see \eqref{eq:nu_def}). 

Panels (a) and (b) of Figure~\ref{fig:CI} display the coverage of the  selective and naive confidence intervals, respectively. The selective  intervals achieve the nominal 95\% coverage of the parameter $\nu^{\top}c$ across all values of $\sigma$ and  $h$. The naive intervals have poor coverage when $1/\sigma$ is small. As $1/\sigma$ increases, however, the coverage of the naive approach improves. This is because when $1/\sigma$ is very large (and hence $\sigma$ is very small), 
 the spikes estimated by solving the $\ell_0$ problem \eqref{eq:l0-opt} do not change much as a function of $\phi$, and thus the truncation set $\{\phi: \thj \in \mathcal{M}(y'(\phi))\}$ in \eqref{eq:S_set} is very large; this means that ignoring this conditioning set has little effect on the confidence interval computed. A similar observation was made for the lasso in \citet{Zhao2017-he}. 
 
Figure~\ref{fig:CI}(c) investigates the average width of the naive and selective confidence intervals as a function of $\sigma$, for $h=1$. Selective intervals are much wider, on average. But the difference in width diminishes as $1/\sigma$ increases. This is congruent with our observations in panel (b): selective intervals can be well-approximated by naive intervals when $1/\sigma$ is large.
 
To understand how selective intervals achieve the nominal coverage, we plot the average midpoint of the selective intervals, after subtracting out $\nu^\top y$,  in panel (d). If a confidence interval is symmetric around $\nu^\top y$ (as is the case for the naive  interval in \eqref{eq:naiveci}), then this value equals zero. A  positive value indicates that the  interval is shifted upwards relative to $\nu^\top y$, and a negative value indicates the opposite.
We see that for all values of $h$ and $\sigma$, the selective intervals have a negative value of the midpoint after subtracting out $\nu^\top y$. This indicates that the  selective approach provides an interval that is centered below the observed value of $\nu^\top y$. 

Throughout this section, we have assumed that $\sigma^2$ in \eqref{eq:obs-model} is known. However, if it is unknown, we propose to use $\hat\sigma^2 = \frac{1}{T-1}\sum_{t=1}^T\left(y_t - \hat{c}_t \right)^2$ as an estimator for $\sigma^2$ in evaluating the $p$-value in \eqref{eq:pval}, where $\hat{c}_t$ is the  solution to \eqref{eq:l0-opt}. In Appendix~\ref{appendix:estimated_sigma}, we demonstrate that this estimator has adequate selective Type I error control and substantial power in a simulation study.

%% file: real_data.tex

\subsection{Overview of data and analysis plan} 
\label{sec:real-overview}

Here we examine data aggregated as part of the \texttt{spikefinder} challenge \citep{Theis2016-wd}. The data consist of simultaneous electrophysiology and calcium recordings for a number of neurons. We consider the spike times recorded through electrophysiology to be the true, or ``ground truth", spike times, against which we assess the accuracy of the spikes estimated via calcium imaging \citep{Theis2016-wd,Berens2018-wq}. The calcium recordings have been resampled to 100 Hz, and linear trends removed, as described in \cite{Theis2016-wd}.

As in prior work \citep{Pachitariu2018-ew,Jewell2019-jw}, we set the value of $\gamma$ in \eqref{eq:l0-opt} based on known properties of the calcium indicators (0.986 for GCamp6f and 0.995 for GCamp6s).  In settings where the properties of the calcium indicators are unknown, we can leverage a proposal from \citet{Fleming2021-mk} for estimating $\gamma$.
Since the calcium has a nonzero baseline, we solve a slight modification of \eqref{eq:l0-opt}:
{
\setlength{\belowdisplayskip}{1pt}
\setlength{\abovedisplayskip}{1pt}
\begin{align}
\minimize{c_1,\ldots, c_T\geq 0, \beta_0}{\frac12 \sum_{t = 1}^{T} (y_{t} - c_{t} - \beta_0)^{2} + \lambda \sum_{t=2}^{T} 1_{(c_{t}\neq \gamma c_{t-1})}}
\label{eq:l0-opt-intercept}.
\end{align} 
}
We first computed the average firing rates for data from \citet{Chen2013-ha}, which are 0.53 and 0.42 spikes per second for GCamp6f and GCamp6s recordings, respectively. For each recording, we solved \eqref{eq:l0-opt-intercept} over a two-dimensional grid of $(\lambda,\beta_0)$ values on the first 25\% of the recording, and considered only the 20 pairs that yield an estimated average firing rate closest to the average firing rate of the corresponding calcium indicator. Among the 20 pairs, we then chose the   $(\lambda,\beta_0)$ pair that results in the smallest objective in \eqref{eq:l0-opt-intercept} on the first 25\% of the recording.

We quantify the accuracy of the estimated spikes resulting from \eqref{eq:l0-opt-intercept}
by comparing them to the ground truth spikes recorded using electrophysiology on the remaining 75\% of each recording, using two widely-used metrics: (i)
 The \emph{correlation between the true and estimated spikes}, after downsampling to 25 Hz, as described in \citet{Theis2016-wd}. Larger values of the correlation suggest better agreement between the true and estimated spikes.
(ii) The \emph{Victor-Purpura distance between the true and estimated spikes}, with cost parameter 10, as proposed in \citet{victor1996nature,victor1997metric}. Smaller values of the Victor-Purpura distance suggest better agreement between the true and estimated spikes.

We also quantify the accuracy of the subset of estimated spikes from \eqref{eq:l0-opt-intercept} for which the $p$-values in \eqref{eq:pval} are below $0.05$.  As before, we computed the $p$-value \eqref{eq:pval} only on the 
estimated spikes for which $\nu^\top y > 0$. For each recording, we used $\hat\sigma^2 = \frac{1}{T-1}\sum_{t=1}^T\left( y_t - \hat{c}_t \right)^2$ to estimate the variance parameter $\sigma^2$, where $\hat{c}_t$ is the  solution to \eqref{eq:l0-opt-intercept}. We used $h=20$ in \eqref{eq:nu_def}; this choice is motivated by the half decay times of the calcium indicators used in \citet{Chen2013-ha}, which are approximately 150 ms and 250 ms for GCamp6f and GCamp6s, respectively. Results for other values of $h$, as well as diagnostics to model \eqref{eq:obs-model}, are in Section~\ref{appendix:sensitivity_real_data} of the Appendix.

\subsection{Results for a single cell}
\label{section:result_one_cell}

In Figure~\ref{fig:real_data_analysis}, we display results for a single cell: recording 29 of dataset 7 from the \texttt{spikefinder} challenge. 
Each panel displays the following quantities, at varying levels of zoom: (i) the fluorescence trace  (grey dots); 
(ii) the estimated spikes from \eqref{eq:l0-opt-intercept} (orange ticks); (iii) the estimated spikes from \eqref{eq:l0-opt-intercept} for which the $p$-values from \eqref{eq:pval} with $h=20$ are below $0.05$ (blue ticks); and (iv) the true spikes (black ticks). 

We see that the estimated spikes with $p$-values less than $0.05$ match very closely with the true spikes. For example, \eqref{eq:l0-opt-intercept} estimates spikes near 79.3, 83.0, 89.1, and 92.9 seconds. None of these correspond to a true spike, and none  have a $p$-value less than $0.05$. Thus, the spikes with $p$-values above $0.05$ appear to be false positives. By contrast, those with $p$-values below 0.05 are mostly true positives. The quantitative measures defined in Section~\ref{sec:real-overview} further indicate that considering only spikes with $p$-values below $0.05$ increases accuracy:  the correlations between the true spikes and the estimated spikes including and excluding $p$-values below $0.05$ are $0.54$ and $0.62$, respectively, and the Victor-Purpura distances are $278$ and $244$, respectively.

\subsection{Results for recordings in \citet{Chen2013-ha}}
\label{section:result_all_cell}

We now examine datasets 7 and 8 of the \texttt{spikefinder} challenge. Their original source is \citet{Chen2013-ha}. The data consist of 58 recordings; each is approximately 230 seconds long. 

Figure~\ref{fig:perm_test_cor_vp_chen_cells} displays the accuracy --- relative to the ground truth spikes obtained via electrophysiology --- of the spikes estimated via \eqref{eq:l0-opt-intercept} (in orange), along with the subset of those  spikes for which the $p$-value is below $0.05$ (in blue). Accuracy is measured using Victor-Purpura distance and correlation.  
We find that the spikes from \eqref{eq:l0-opt-intercept} with $p$-values below $0.05$ are more accurate than the full set of spikes from \eqref{eq:l0-opt-intercept}.
These results are based on $h=20$ in \eqref{eq:pval}. Results for $h=5$ and $h=50$  are similar; see Figures~\ref{fig:sensitivity_hcor_vp_chen_cells} and \ref{fig:sensitivity_hcor_vp_chen_cells_large_h} in Section~\ref{appendix:sensitivity_real_data} of the Appendix. 

It is natural to wonder whether retaining only estimated spikes with $p$-values below $0.05$ improves the correlation and Victor-Purpura distance merely as a byproduct of reducing the number of estimated spikes, rather than due to the high quality
of the estimated spikes with $p$-values below $0.05$. We assess this using a resampling approach. Let $M$ denote the number of spikes for which $p$-values are computed, and let $\tilde{M}$ denote the number that are below $0.05$.  We sample $\tilde{M}$ out of $M$ estimated spike times for which $p$-values are computed without replacement, and compute the correlation and Victor-Purpura distance between the true spike times and the sampled subset.  We do this 1,000 times, and record the 2.5\% and 97.5\% quantiles of the accuracy measures obtained. 
These are shown as the endpoints of the black lines displayed in Figure~\ref{fig:perm_test_cor_vp_chen_cells}. We see that even after taking into account the effect of a smaller number of estimated spikes, excluding spikes with $p$-values greater than 0.05 still provides improved accuracy, measured using either correlation (56 out of 58 recordings) or Victor-Purpura distance (51 out of 58 recordings).

%% file: discussion.tex
Methods developed in this paper are implemented in the \texttt{R} package \texttt{SpikeInference}, available at \texttt{https://github.com/yiqunchen/SpikeInference}. We provide a tutorial for the package at \texttt{https://yiqunchen.github.io/SpikeInference/}. Code for reproducing the results in this paper can be found at \texttt{https://github.com/yiqunchen/SpikeInference-experiments}.

Our work leads to a few future directions of research.

\subsection{Alternative conditioning sets and contrast vectors for testing \eqref{eq:null-nu}}

Instead of conditioning on the $j$th estimated spike $\thj$ to obtain the $p$-value in \eqref{eq:pval}, we could instead condition on $\thj$ and its immediate neighbors, $\hat{\tau}_{j-1}$ and $\hat{\tau}_{j+1}$. 
This would allows us to define the contrast vector $\nu$ as 
{
\begin{equation*}
  \setlength{\belowdisplayskip}{1em}
\setlength{\abovedisplayskip}{1em}
\nu_{t} = 
\begin{cases}
-\frac{\gamma(\gamma^2-1)}{\gamma^2-\gamma^{2(\hat\tau_{j-1}-\thj)}} \cdot \gamma^{t-\thj}  , & \hat{\tau}_{j-1}+1 \leq  t \leq \thj \\
\frac{\gamma^2-1}{\gamma^{2(\hat\tau_{j+1}-\thj)}-1}\cdot  \gamma^{t - (\thj + 1)}  , & \thj+1 \leq t \leq \hat{\tau}_{j+1} \\
 0, & \text{otherwise},
\end{cases}
\end{equation*}
}
leading to a $p$-value given by $\mathbb{P}\left(\phi\geq \nu^\top y \;\middle\vert\; \{\hat{\tau}_{j-1},\thj,\hat{\tau}_{j+1}\} \subseteq \mathcal{M}(y'(\phi)), \phi>0  \right)$, where $\phi\sim\mathcal{N}(0,\sigma^2||\nu||_2^2)$. 
This approach eliminates the need to specify a window size $h$, and instead chooses the window size adaptively. Computing this new $p$-value requires only minor modifications of the results in Section~\ref{section:method}, using ideas from \citet{Jewell2019-vv}; we leave the details to future work.

As an alternative, we could keep the conditioning set in \eqref{eq:pval}, but define a contrast vector $\nu$ that uses different numbers of timepoints to the left and right of $\hat\tau_j$ (in contrast to \eqref{eq:nu_def}). 


 
\subsection{Selective inference for other spike detection methods}
In this paper, we considered selective inference on spikes estimated via the $\ell_0$ problem in  \eqref{eq:l0-opt}. However, another line of research \citep{vogelstein2010fast,friedrich2016fast,Friedrich2017-vn} involves estimating spikes via an $\ell_1$-penalized approach: {
  \setlength{\belowdisplayskip}{1em}
\setlength{\abovedisplayskip}{1em}
\begin{equation}
 \minimize{c_1,\hdots, c_T\geq 0; z_1,\hdots, z_T}{\frac12 \sum_{t = 1}^{T} (y_{t} - c_{t})^{2} + \lambda \sum_{t=2}^{T} |z_{t}|} \text{ subject to } z_t= c_t-\gamma c_{t-1}\geq 0. 
 \label{eq:l1}
\end{equation} 
Spikes are estimated at timepoints for which $\hat{c}_t \neq \gamma \hat{c}_{t-1}$. To conduct inference on these estimated spikes, we could leverage the framework in Section \ref{sec:selective_test}, along with recent developments in selective inference for the lasso and related problems \citep{Lee2016-te,Hyun2018-pe}. 
 
\subsection{Propagating uncertainty to downstream data analysis}
This article focused on quantifying the uncertainty associated with $\nu^\top c$, the change in calcium associated with an estimated spike. It is also of interest to propagate this uncertainty to downstream analyses, such as the neural decoding model \citep{Pillow2011-er,Ventura2008-ld}. This model is  similar to \eqref{eq:obs-model} with $z_t \overset{\text{i.i.d.}}{\sim} \text{Poisson}\qty(f\qty(\theta_t))$ for a function $f$; the goal is to estimate the coefficients  $\theta_t$. We could leverage the framework proposed in \citet{Wei2019-gm} to propagate uncertainty of estimating $\nu^\top c$ to $\theta_t$.

\section*{Supplementary Materials}
The reader is referred to the online Supplementary Materials for technical appendices, proofs of all Propositions, and additional results. 

 \section*{Acknowledgments}
 We thank Paul Fearnhead for helpful conversations. \emph{Conflict of Interest:} None declared.

\section*{Funding}

  This work was partially supported by National Institutes of Health grants [R01EB026908, R01DA047869] and a Simons Investigator Award in Mathematical Modeling of Living Systems to D.W.

%% file: figures.tex

\noindent
\begin{figure}[htbp!]
\begin{centering}
\hspace{15mm} (a) \hspace{40mm} (b) \hspace{40mm} (c)\\
\end{centering}
\begin{subfigure}{\linewidth}
  \centering
  \includegraphics[width=.3\linewidth]{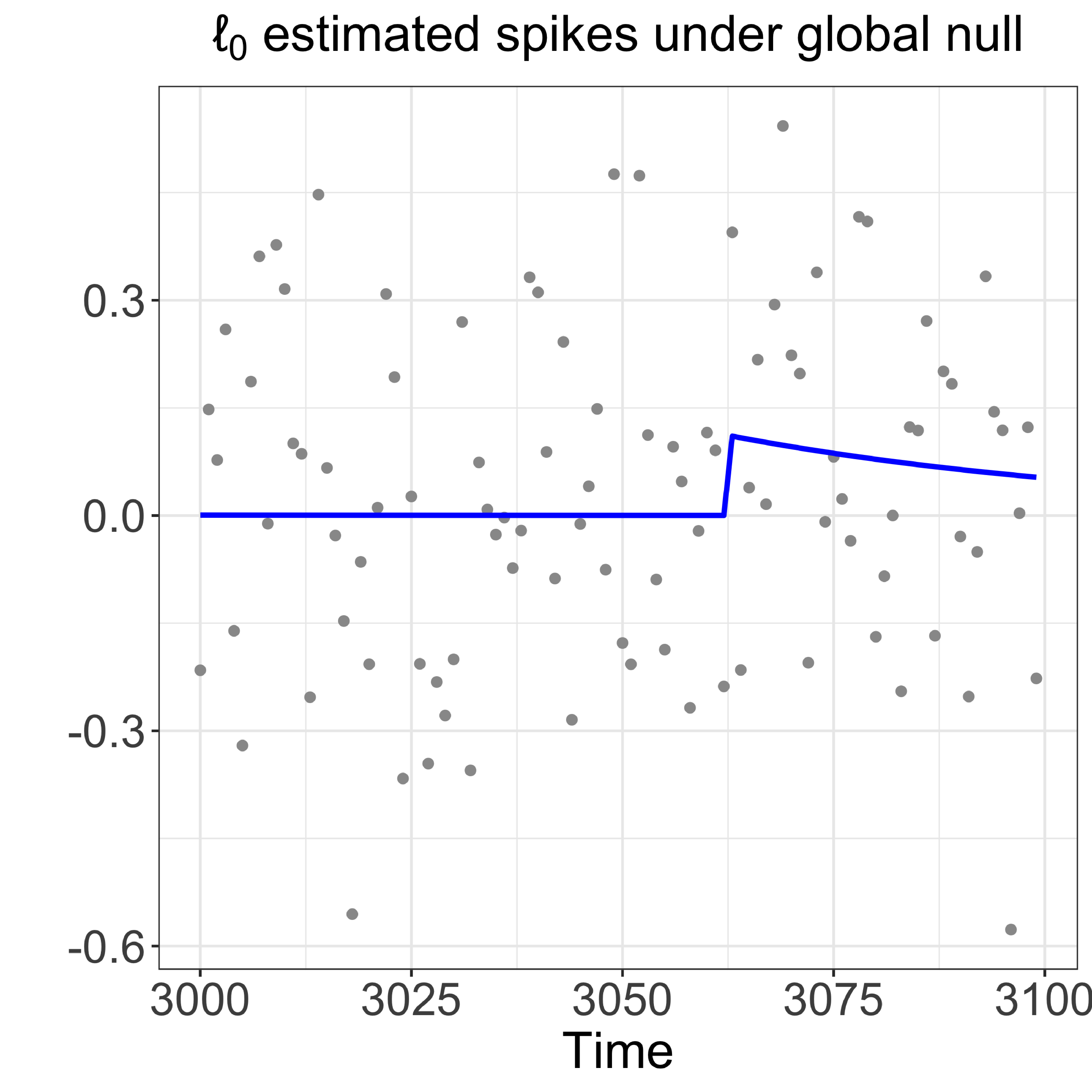}
  \includegraphics[width=.3\linewidth]{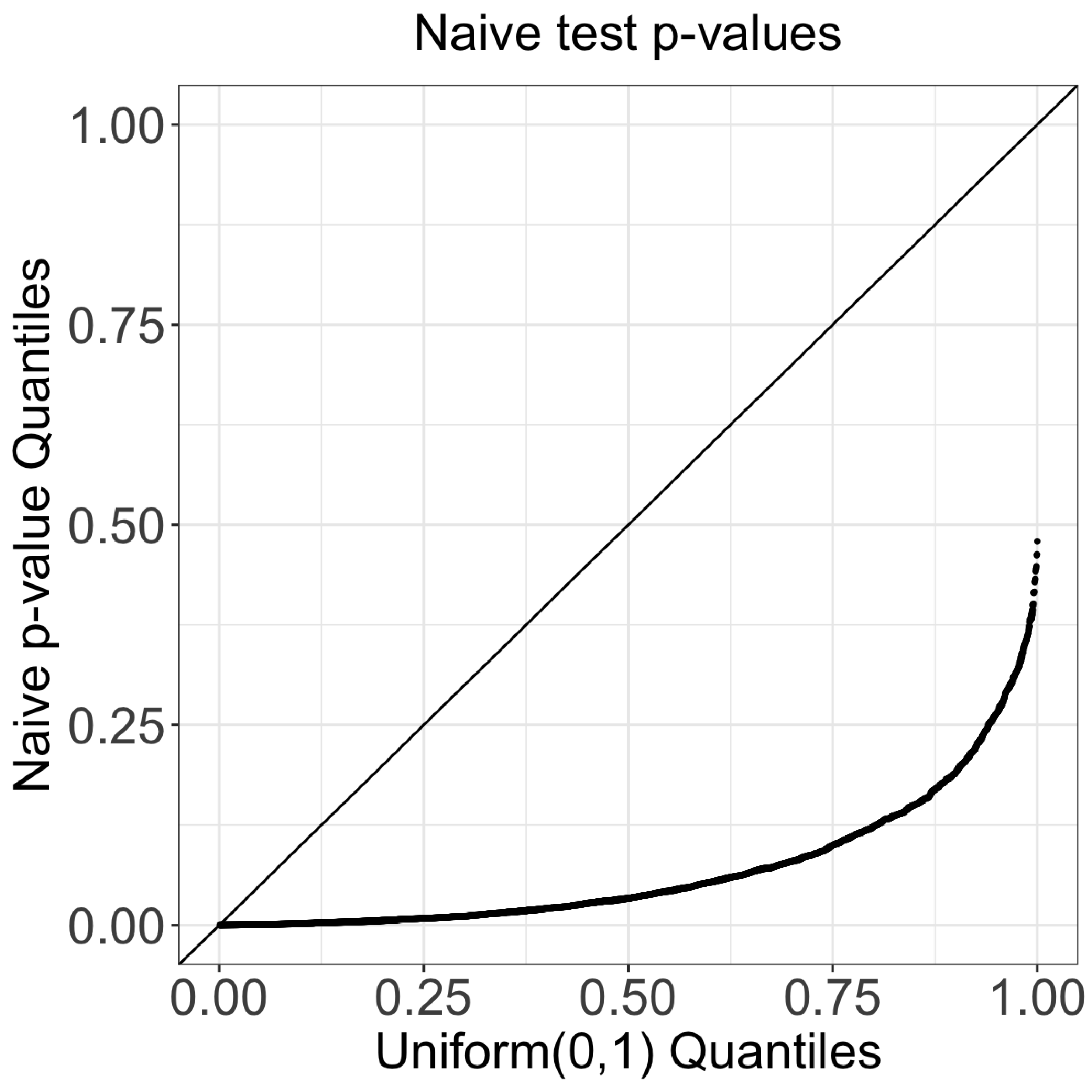}
  \includegraphics[width=.3\linewidth]{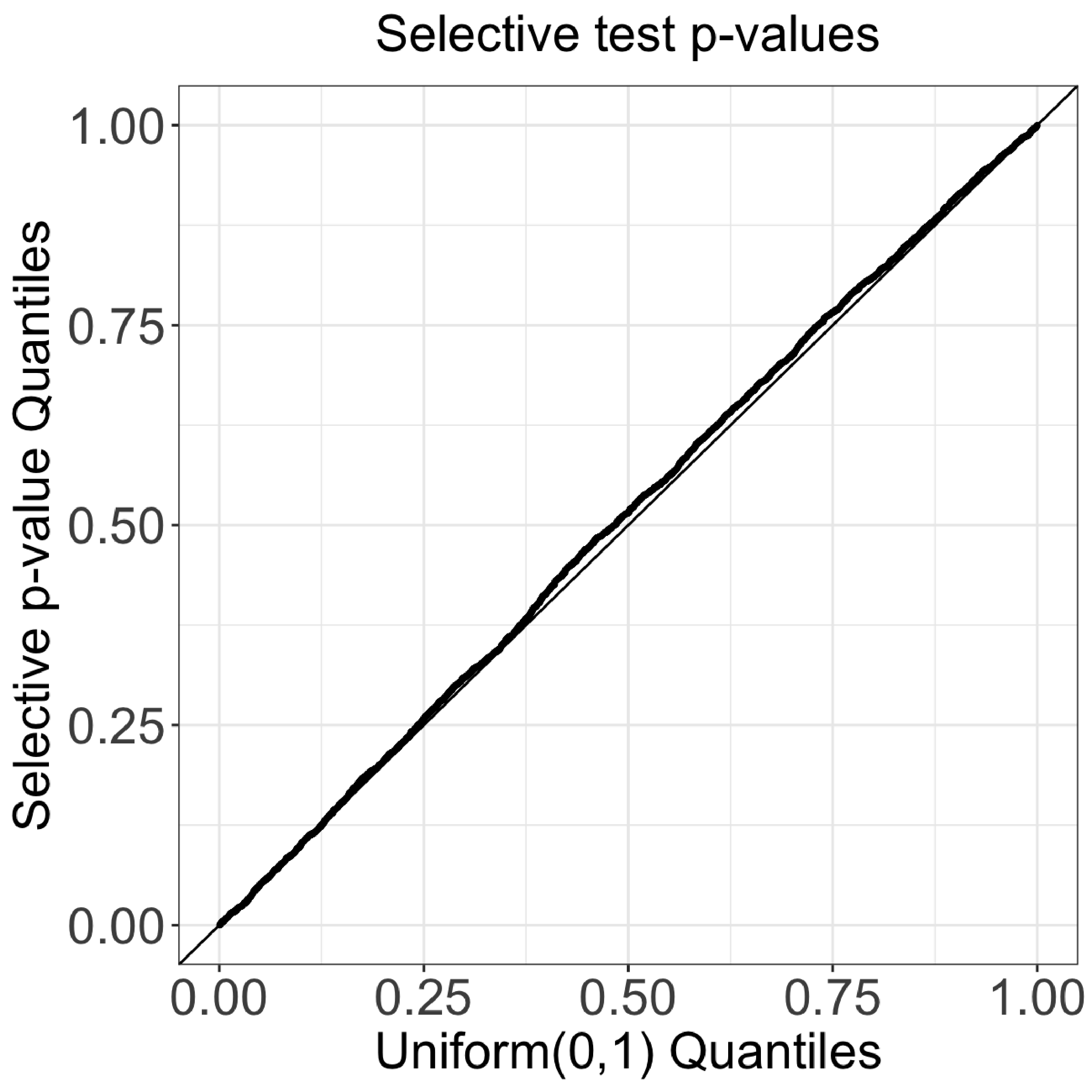}
\end{subfigure}
\caption{\textit{(a): } One simulation with $y_1,\dots, y_{10,000}$ (grey dots) generated according to model \eqref{eq:obs-model} with $\gamma = 0.98$, $\sigma = 0.2$, and $z_t = 0$ for all $t$. The $\ell_0$ problem in \eqref{eq:l0-opt} was solved with $\lambda=0.1$, resulting in 47 estimated spikes with fluorescence increases. Estimated calcium is displayed in blue. We display one estimated spike at time $\hat\tau = 3,060$ with $y_{3,000},\ldots, y_{3,100}$. \textit{(b): } Quantile-quantile plot for the Wald $p$-values (defined in \eqref{eq:wald_test}) based on 100 simulations (2,988 hypothesis tests).
\textit{(c): } Quantile-quantile plot for the selective $p$-values (defined in \eqref{eq:pval} with $h=1$) based on 100 simulations (2,988 hypothesis tests).} 
\label{fig:motivation}
\end{figure}

\noindent
\begin{figure}[htbp!]
\begin{centering}
\hspace{15mm} (a) \hspace{40mm} (b) \hspace{40mm} (c)\\
\end{centering}

\begin{subfigure}{\linewidth}
  \centering
  \includegraphics[width=0.3\linewidth]{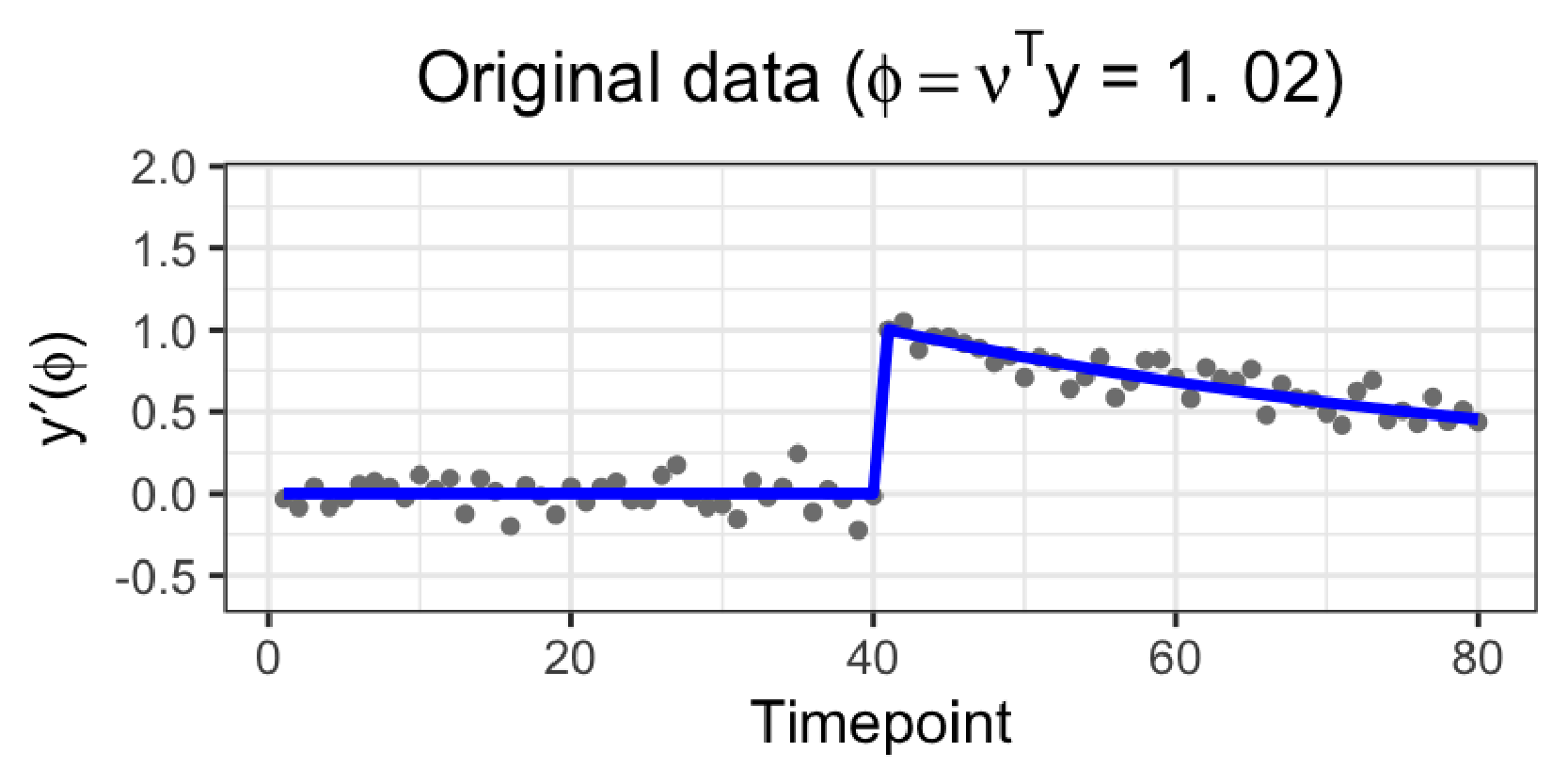}
  \includegraphics[width=0.3\linewidth]{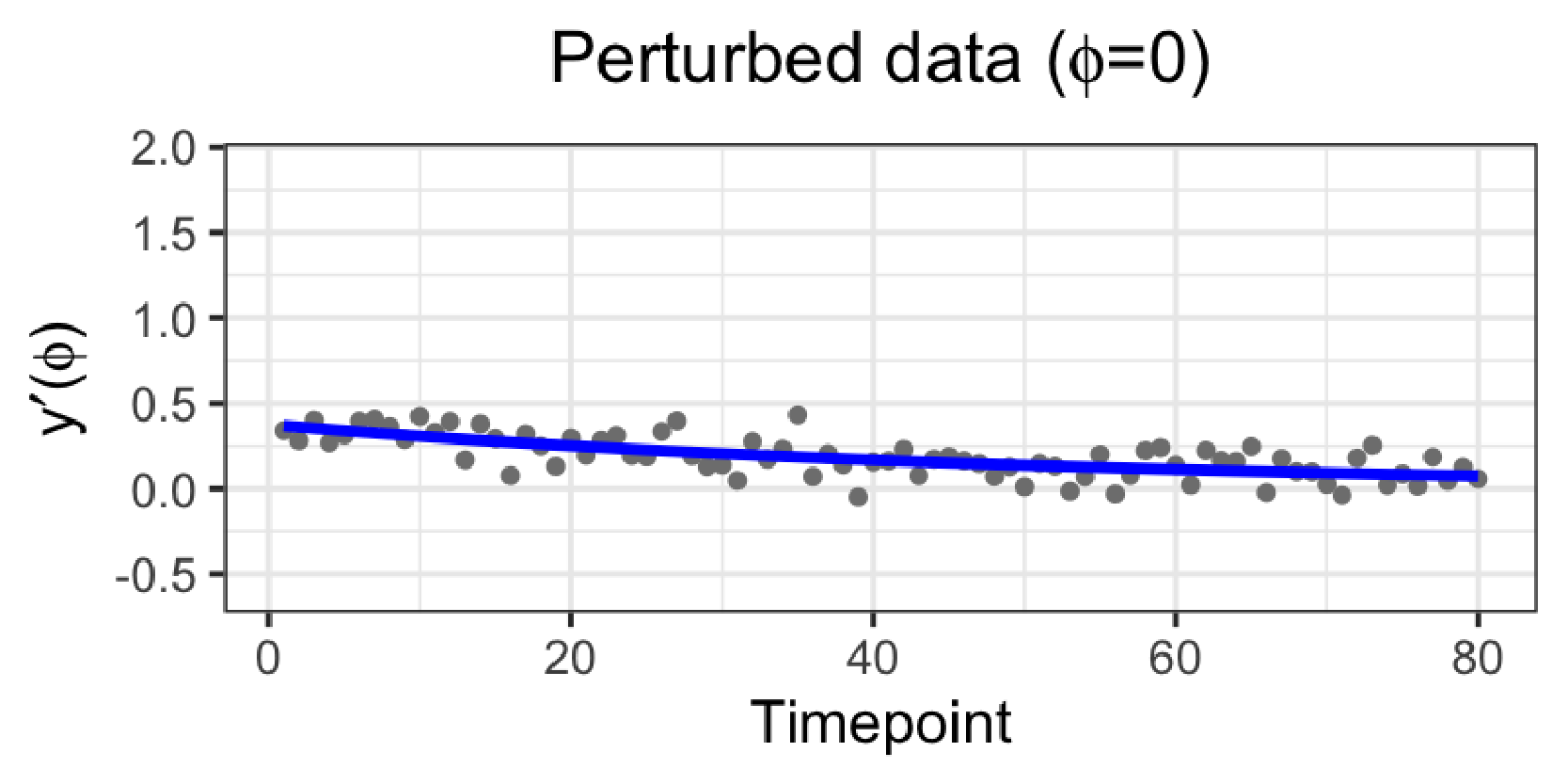}
  \includegraphics[width=0.3\linewidth]{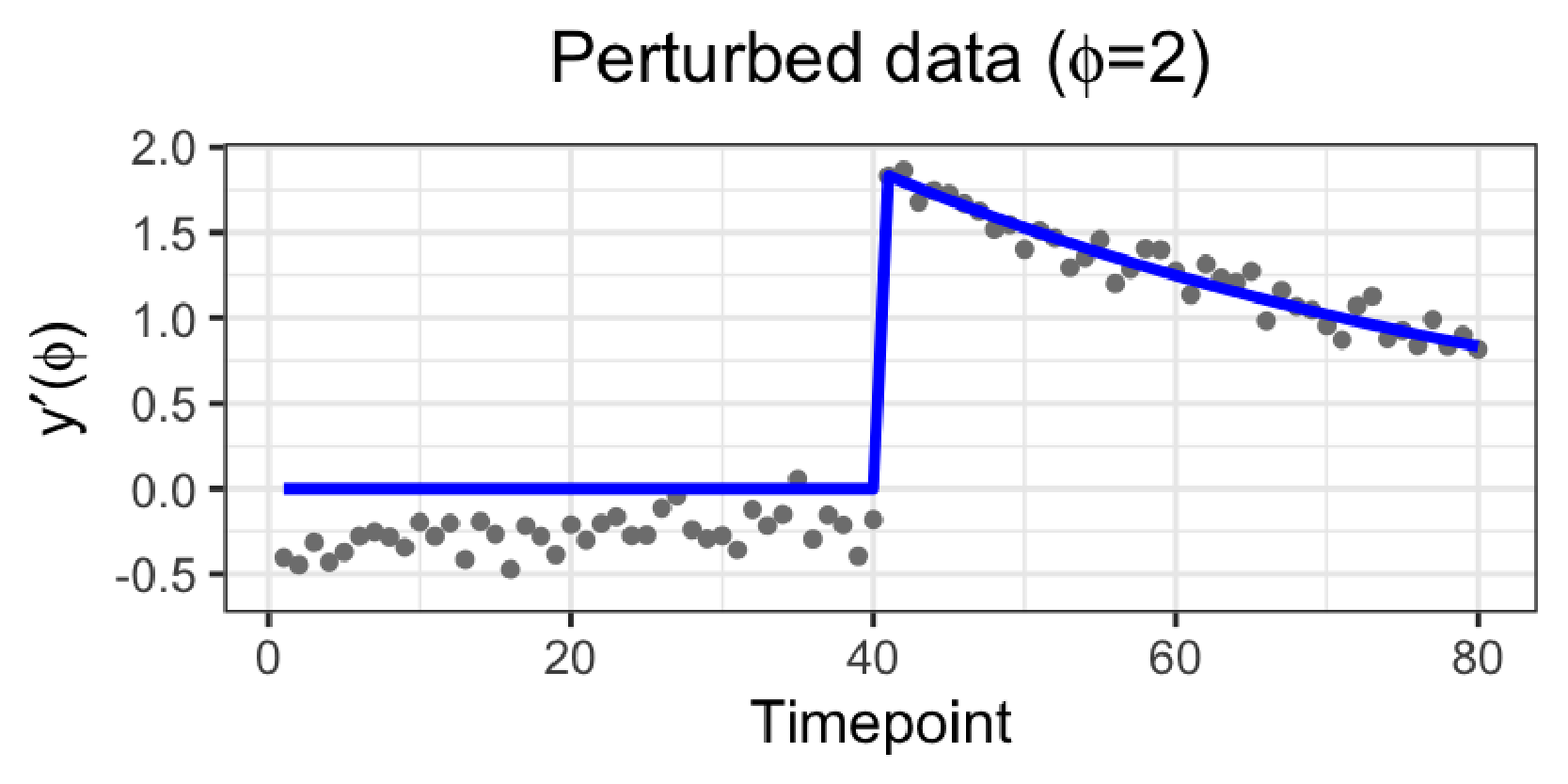}
\end{subfigure}

\begin{centering}
\hspace{15mm}(d)  \hspace{5mm}\\
\end{centering}
 \begin{subfigure}{\linewidth}
    \centering
   \noindent\includegraphics[width=0.9\linewidth]{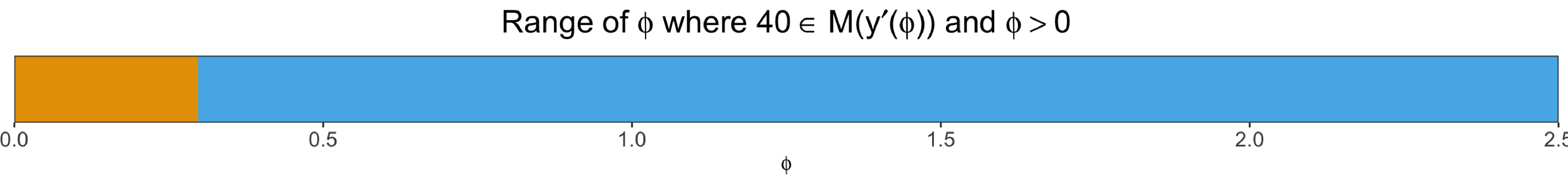}
    \end{subfigure}

\caption{Data generated according to \eqref{eq:obs-model}, with $T=80,\sigma = 0.1, \gamma=0.98$, and one  spike at $t=40$. Solving the $\ell_0$ problem \eqref{eq:l0-opt} with $\lambda = 0.75$ yields a single estimated spike at $t= 40$. \textit{(a): } We plot the original data, which corresponds to $y'(\phi)$ with $\phi = \nu^{\top}y = 1.02$, where $\nu$ is constructed according to \eqref{eq:nu_def} with $\thj = 40$ and $h=40$. The estimated calcium concentration is displayed in blue. \textit{(b): } The perturbed dataset $y'(\phi)$ with $\phi = 0$ is shown. Now there is no increase in calcium at $t=40$ on $y'(\phi)$, and no spike is estimated. \textit{(c): } The perturbed dataset $y'(\phi)$ with $\phi=2$ is shown. There is now a very pronounced increase in calcium at $t=40$, and a spike is estimated. \textit{(d): } The set of $\phi$ for which $40 \in \mathcal{M}(y'(\phi))$ and $\phi>0$ is displayed in blue; other values of $\phi$ are in orange. }
\label{fig:perturbation}
\end{figure}

\noindent
\begin{figure}[htbp!]
\begin{centering}
\hspace{5mm} (a) \hspace{60mm} (b) \\
\end{centering}
\begin{subfigure}{\linewidth}
  \noindent\includegraphics[width=0.9\linewidth]{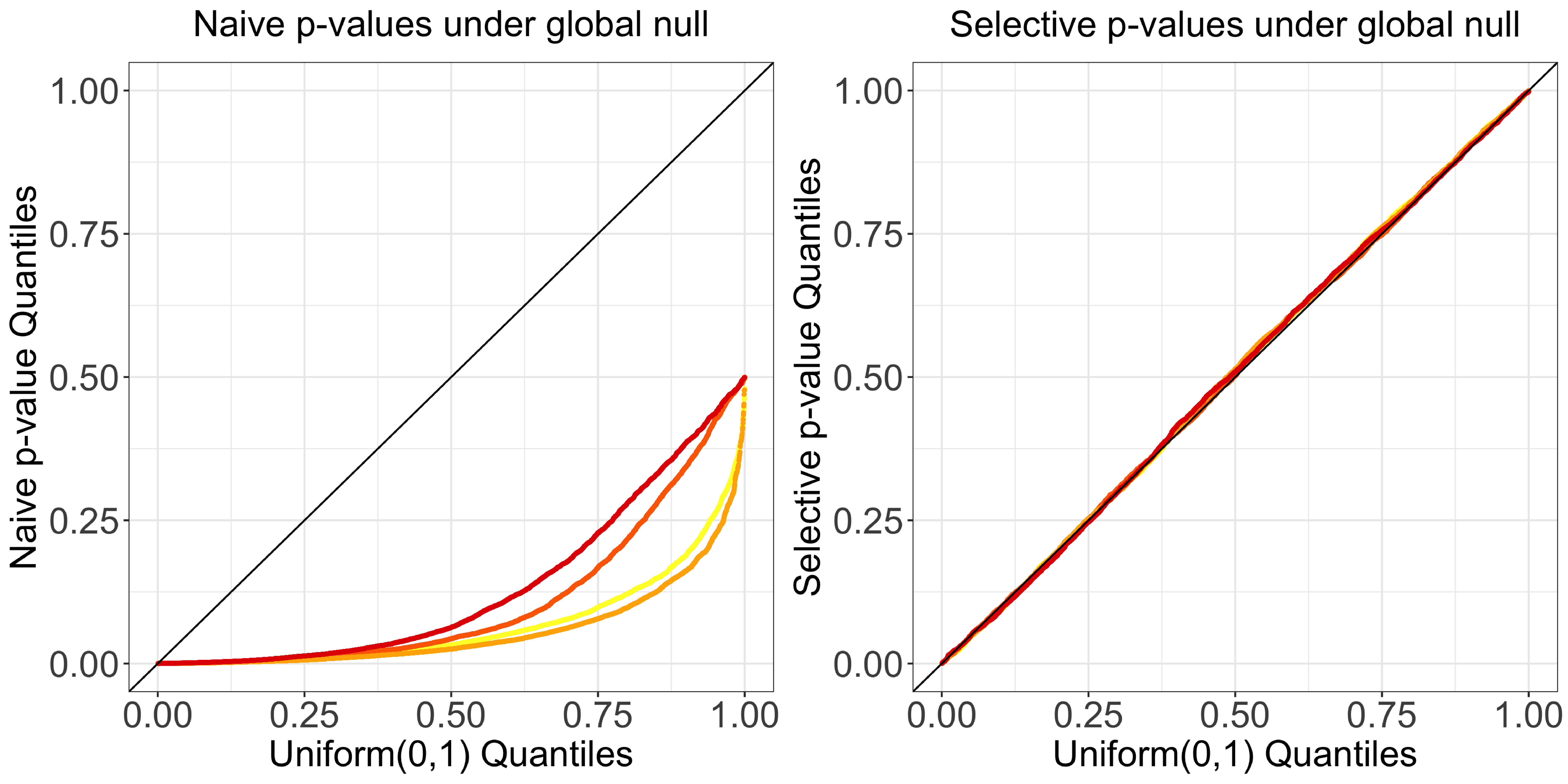}
\end{subfigure}

\begin{centering}
\hspace{5mm} (c) \hspace{60mm} (d) \\
\end{centering}
\begin{subfigure}{\linewidth}
  \noindent\includegraphics[width=0.9\linewidth]{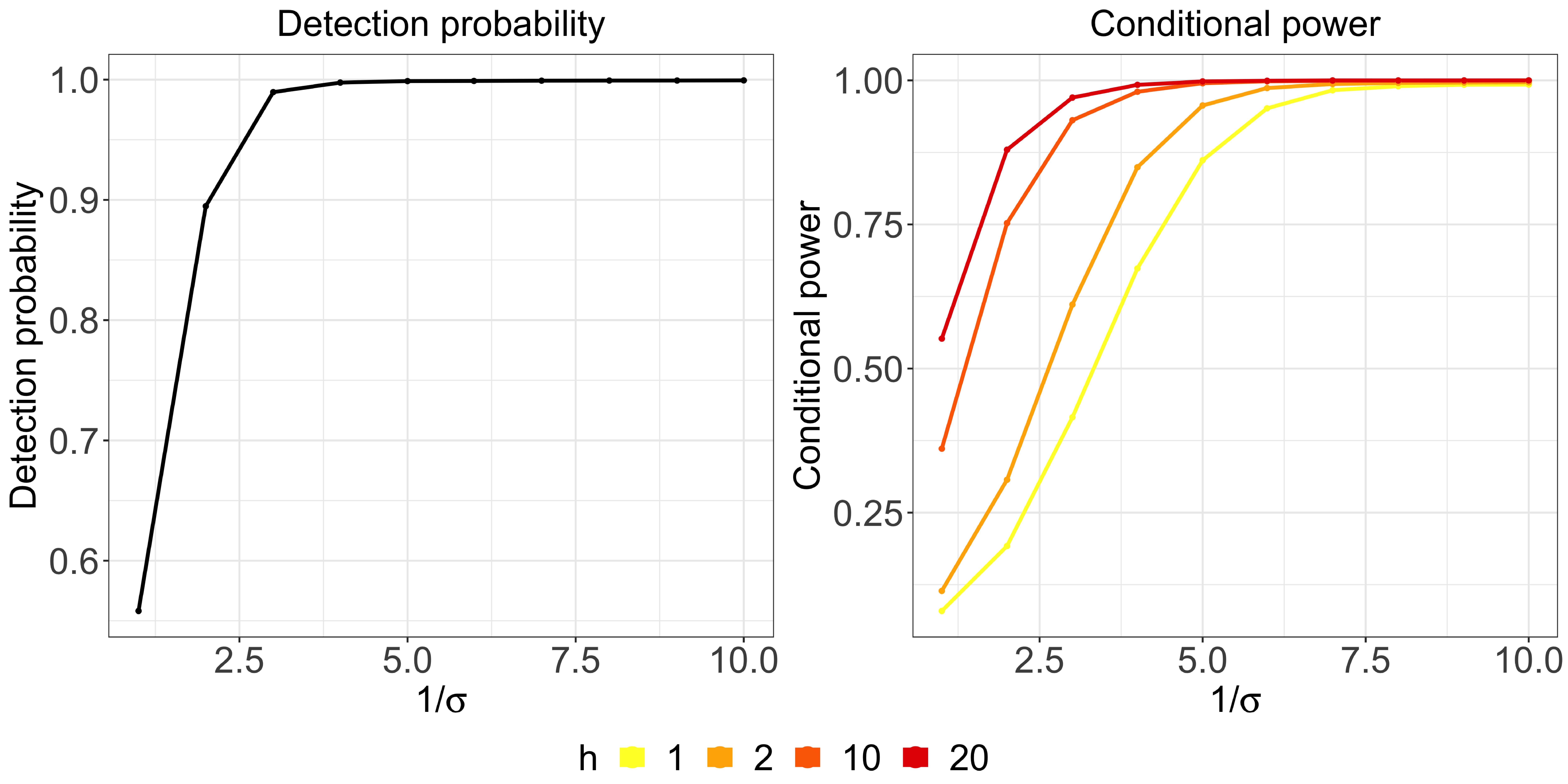}
\end{subfigure}
\caption{\textit{(a):} Quantile-quantile plot for the naive $p$-values defined in \eqref{eq:wald_p_val}, which have inflated selective Type I error. \textit{(b):} Quantile-quantile plot for $p$-values from our proposed selective test in \eqref{eq:pval}, which controls selective Type I error. \textit{(c):} Under the model \eqref{eq:obs-model}, detection probability \eqref{eq:detect}  is an increasing  function of $1/\sigma$. \textit{(d):} Conditional power \eqref{eq:cond} increases as a function of $1/\sigma$ for all $h$. For a given value of $\sigma$, a larger value of  $h$ corresponds to higher conditional power, with the caveat that the meaning of the null hypothesis in \eqref{eq:null-nu} changes as a function of $h$, and the null hypothesis that holds for a smaller $h$ might not hold for a larger value of $h$. The constant $h$ appears in the definition of the contrast vector $\nu$; see \eqref{eq:nu_def}.
}
\label{fig:Type_I_power}
\end{figure}

\begin{figure}[htbp!]
\begin{centering}
\hspace{5mm} (a) \hspace{60mm} (b) \\
\end{centering}
\begin{subfigure}{\linewidth}
  \noindent\includegraphics[width=0.9\linewidth]{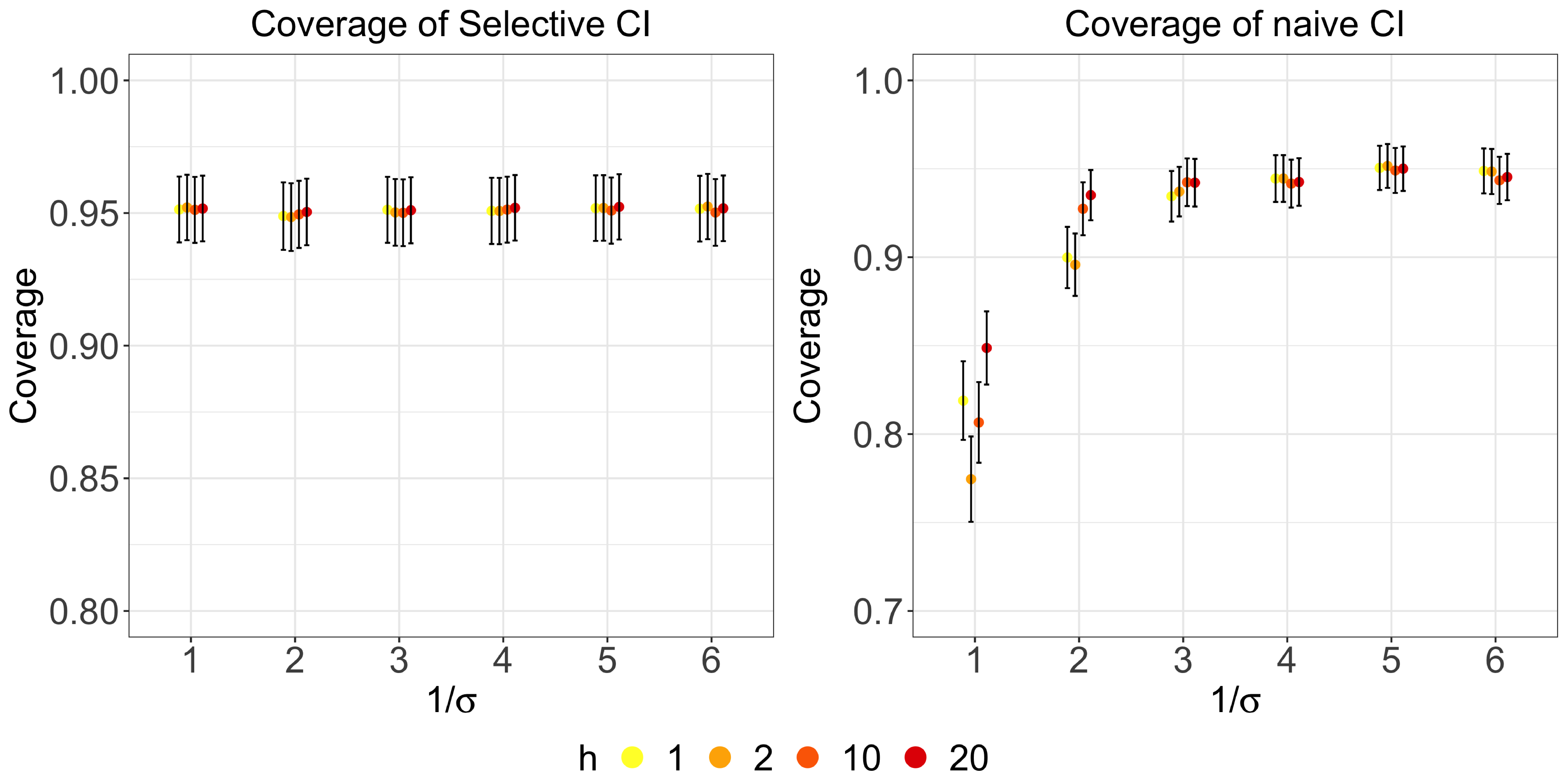}
\end{subfigure}

\begin{centering}
\hspace{5mm} (c) \hspace{60mm} (d) \\
\end{centering}

\begin{subfigure}{0.9\linewidth}
   \noindent\includegraphics[width=\linewidth]{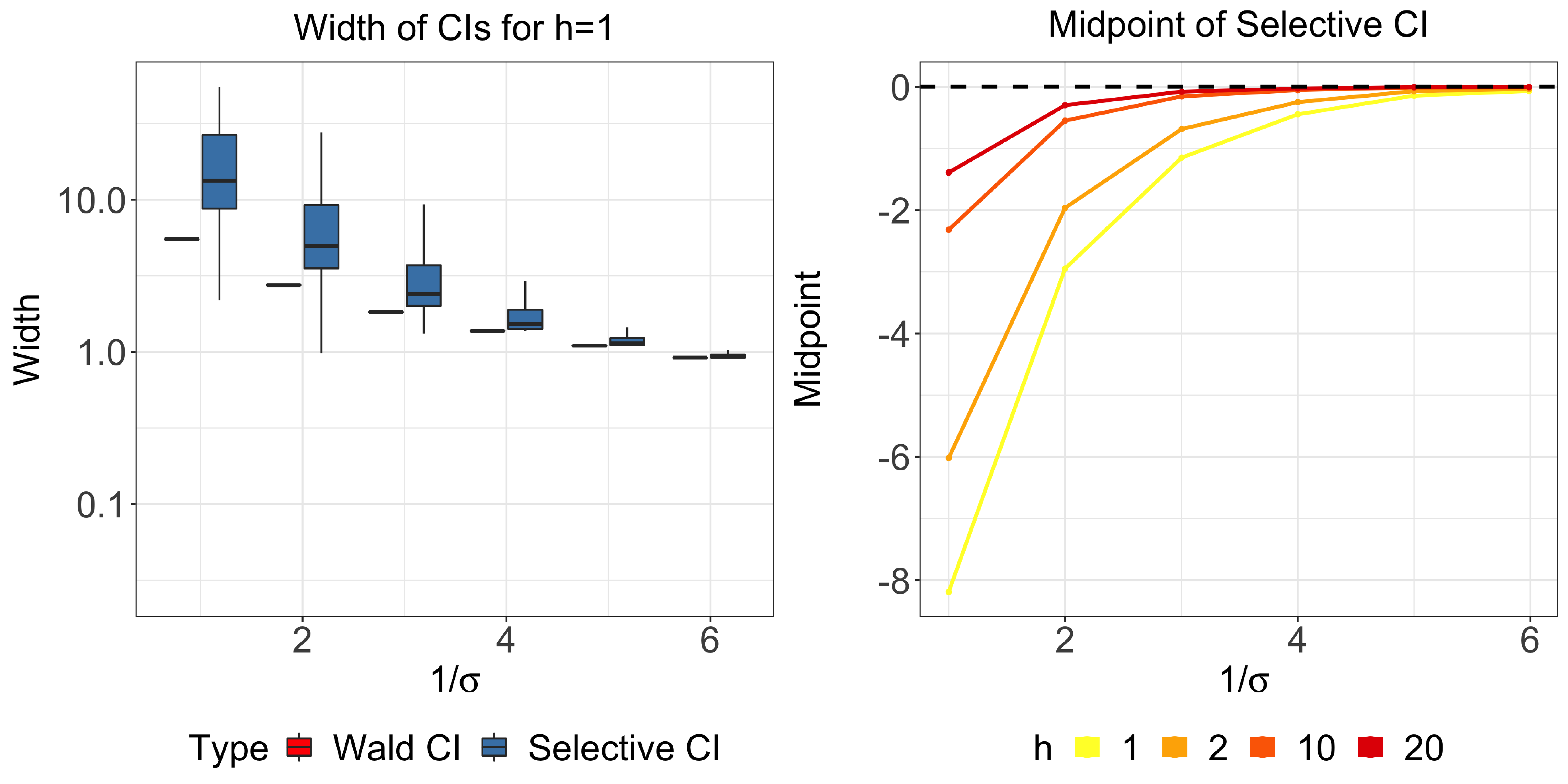}
\end{subfigure}

\caption{\textit{(a):} Selective confidence intervals achieve correct nominal coverage (95\% coverage at level $\alpha=0.05$) across all values of $h$ (defined in \eqref{eq:nu_def}) and $\sigma$ (defined in \eqref{eq:obs-model}). The mean (and standard deviation) over 500 simulated datasets are displayed. \textit{(b):} Naive confidence intervals  have poor coverage when $1/\sigma$ is small, for all values of $h$. \textit{(c):} For $h=1$, selective confidence intervals are on average wider than naive intervals, but the difference decreases as $1/\sigma$ increases. \textit{(d):} The midpoint of the selective confidence interval is, on average, smaller than $\nu^\top y$. 
}
\label{fig:CI}
\end{figure}

\begin{figure}[htbp!] 
  \noindent\includegraphics[width=\linewidth]{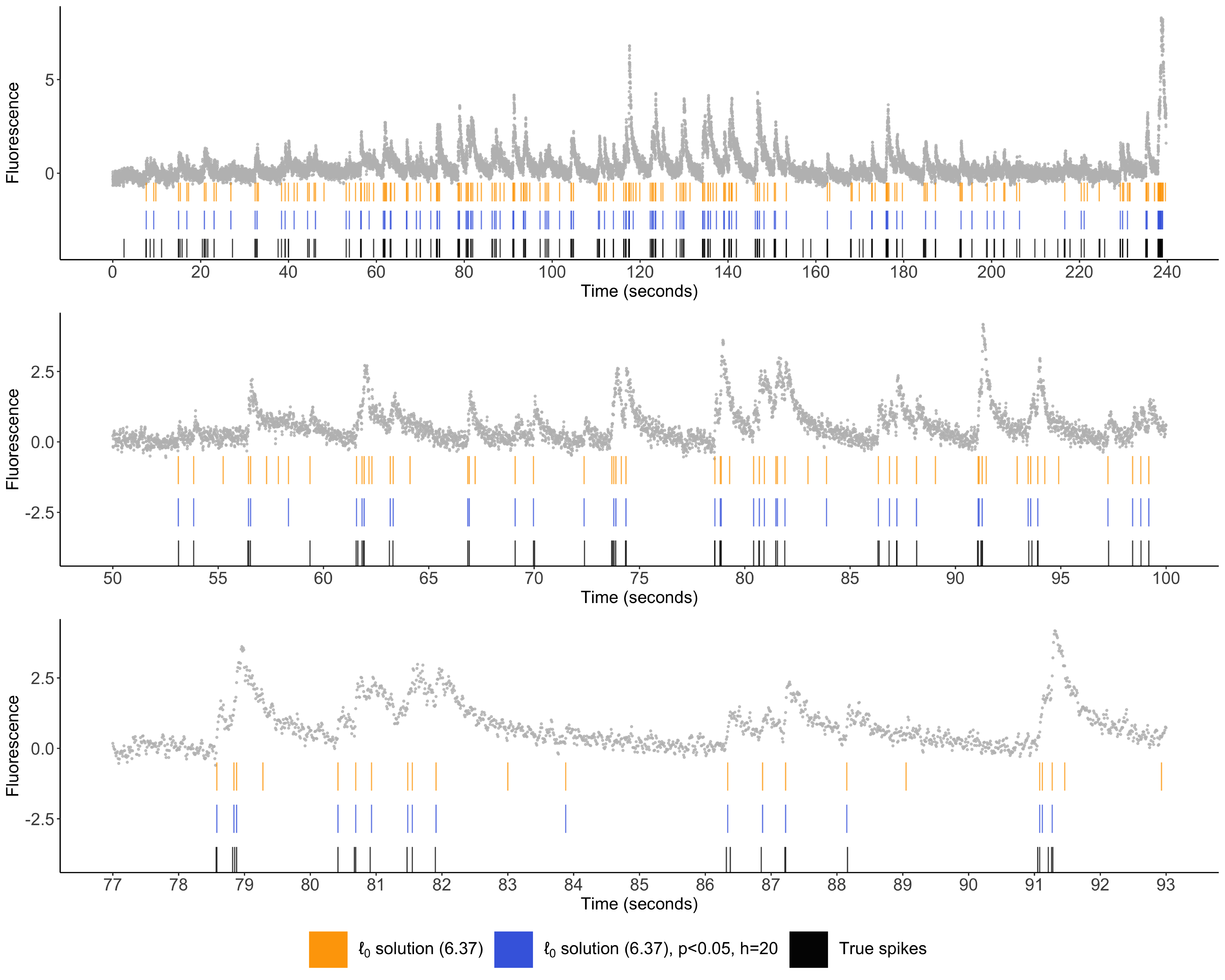}
  \caption{Illustrative example for recording 29 from \citet{Chen2013-ha}, which uses the GCaMP6f indicator, after preprocessing as described in \citet{Theis2016-wd}. The cell's fluorescence trace is displayed in grey. Estimated spikes from \eqref{eq:l0-opt-intercept} are displayed in orange; the spikes with $p$-values from \eqref{eq:pval} below $0.05$ (with $h=20$) are displayed in blue; and the true spike times are shown in black.}
\label{fig:real_data_analysis}
\end{figure}

\begin{figure}[htbp!]
\begin{centering}
\hspace{15mm}(a)  \hspace{5mm}\\
\end{centering}
\begin{subfigure}{\linewidth}
  \noindent\includegraphics[width=0.9\linewidth]{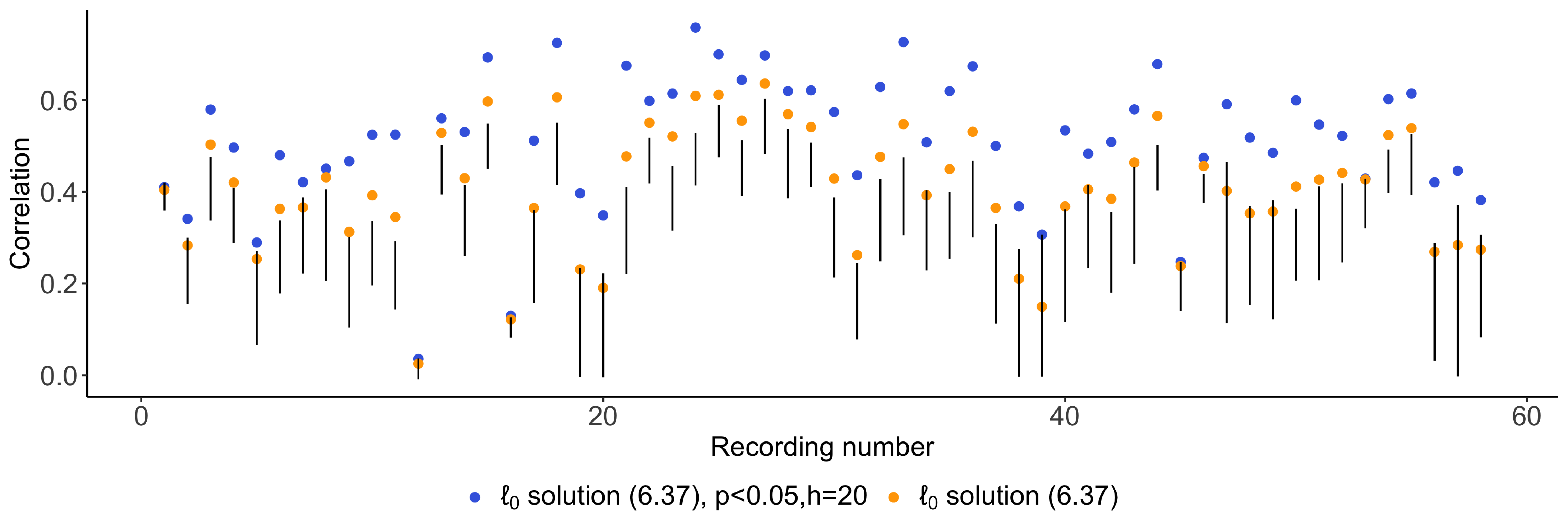}
\end{subfigure}

\begin{centering}
\hspace{15mm}(b)  \hspace{5mm}\\
\end{centering}
\begin{subfigure}{\linewidth}
  \noindent\includegraphics[width=0.9\linewidth]{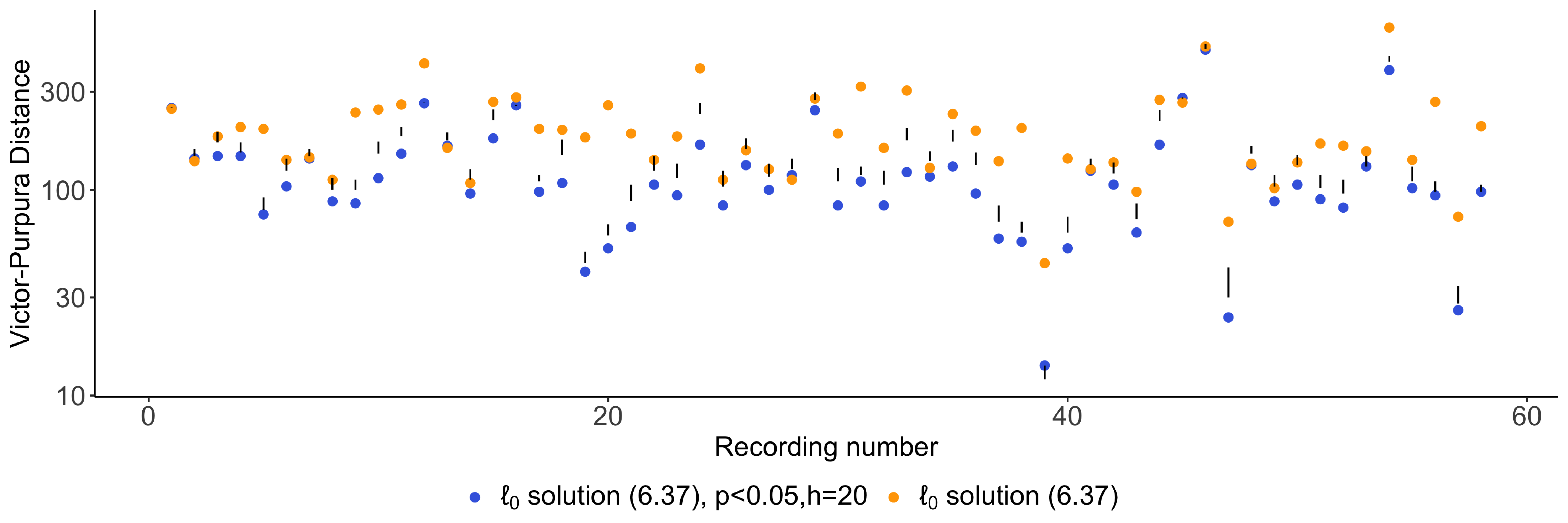}
\end{subfigure}
\caption{Result for recordings from the \cite{Chen2013-ha} dataset. \textit{(a):} The correlations between the true spike times and the spikes estimated from \eqref{eq:l0-opt-intercept} are plotted in orange. The correlations between the true spike times and the subset of the spikes from \eqref{eq:l0-opt-intercept} with $p$-value \eqref{eq:pval} below $0.05$ are plotted in blue. For each recording, the black line represents the 2.5\% and 97.5\% quantiles of the resampling distribution with 1,000 samples. \textit{(b):} As in (a), but Victor-Purpura distance is displayed instead of correlation.}
\label{fig:perm_test_cor_vp_chen_cells}
\end{figure}

%% file: appendix_a_1_2.tex

\subsection{Proof of Proposition~\ref{prop:pval}}
\label{appendix:prop_pval}
We first prove the statement \eqref{eq:single_param_general}.
The following equalities hold:
\begin{small}
\begin{align*}
&\mathbb{P}\left(\nu^{\top} Y \geq \nu^{\top} y \;\middle\vert\; \thj(y) \in \mathcal{M}(Y), \Pi_{\nu}^{\perp} Y= \Pi_{\nu}^{\perp} y,\nu^{\top} Y>0 \right) \\
 &\overset{a.}{=} \mathbb{P}\left(\nu^{\top} Y \geq \nu^{\top} y \;\middle\vert\; \thj(y) \in \mathcal{M}(\Pi_{\nu}^{\perp}y +\Pi_{\nu} Y), \Pi_{\nu}^{\perp} Y= \Pi_{\nu}^{\perp} y , \nu^{\top} Y > 0 \right) \\
 &\overset{b.}{=}  \mathbb{P} \left( \phi \geq  \nu^{\top}y  \;\middle\vert\; \thj(y) \in \mathcal{M}\left(y'(\phi)\right) ,\Pi_{\nu}^{\perp} Y= \Pi_{\nu}^{\perp} y , \nu^{\top} Y > 0 \right) \\
 &\overset{c.}{=}  \mathbb{P} \left( \phi \geq  \nu^{\top}y  \;\middle\vert\; \thj(y) \in \mathcal{M}\left(y'(\phi)\right),\phi>0 \right).
\end{align*}
\end{small}
Here, $a.$ follows from the fact that $Y = \Pi_{\nu}^{\perp}Y+\Pi_{\nu}Y$, and the fact that we have conditioned on the event $\Pi_{\nu}^{\perp}Y = \Pi_{\nu}^{\perp}y$. To prove $b.$, we first note that $I = \Pi_{\nu} + \Pi_{\nu}^{\perp}$ and $\Pi_{\nu} Y=\frac{\nu\nu^{\top}}{||\nu||_2^2}Y$, which implies $$\Pi_{\nu}^{\perp}y +\Pi_{\nu} Y = y- \Pi_{\nu}y+\Pi_{\nu} Y = y-\frac{\nu^{\top}y}{||\nu||_2^2}\nu+\frac{\nu^{\top}Y}{||\nu||_2^2}\nu  = y'(\phi), $$ where we define $\phi = \nu^{\top} Y \sim  \mathcal{N}(\nu^\top c, \sigma^2 ||\nu||_2^2)$. Finally, $c.$ follows from the fact that $Y \sim \mathcal{N}(c, \sigma^2 I)$ implies independence of $\phi = \nu^{\top}Y$ and $\Pi_\nu^\perp Y$.

Now to prove \eqref{eq:single_param}, we note that under $H_0$ in \eqref{eq:null-nu}, $\nu^\top Y \sim \mathcal{N}\qty(0, \sigma^2 ||\nu||_2^2)$. Therefore, applying the result above with $\nu^\top c =0$ completes the proof.

\subsection{General case for the contrast vector $\nu$}
\label{appendix:general_nu}

The definition \eqref{eq:nu_def} only applies when $\thj-h+1 \geq 1$ and $\thj+h\leq T$. In the case that $\thj-h+1< 1$ or $\thj+h> T$, we define the contrast vector $\nu$ as follows:
\begin{align}
\label{eq:nu_def_special}
\nu_{t} &= 
\begin{cases}
-\gamma \frac{\gamma^2-1}{\gamma^2-\gamma^{2(\hat\tau_L-\thj)}}\cdot \gamma^{t-\thj}  , & \hat\tau_L \leq  t \leq \thj \\
\frac{\gamma^2-1}{\gamma^{2(\thi{R}-\thj)}-1}\cdot  \gamma^{t - (\thj + 1)}  , & \thj+1 \leq t \leq \hat\tau_R \\
0, & \text{otherwise}  \\
\end{cases},
\end{align}
where
$\hat\tau_L = \max(1, \thj-h+1)$, and $\hat\tau_R = \min(T, \thj+h)$.

In Figure~\ref{fig:plot_nu_vec}, we plot the contrast vector in \eqref{eq:nu_def}, generated with $T=50$, $\gamma=0.98$, $\hat\tau_j=20$, and $h=5$.

%% file: appendix_a_4.tex

\subsection{Proof of Proposition~\ref{prop:characterization_S}}
\label{appendix:prop_S_char}

Recall that the $\ell_0$ problem \eqref{eq:l0-opt} is equivalent to the changepoint detection problem \eqref{eq:changepoint_form}, in the sense that \eqref{eq:l0-opt} results in an estimated changepoint at  $\hat\tau_j$ if and only if $\hat\tau_j$ is in the solution to \eqref{eq:changepoint_form}. 

We first prove that $C(\phi)$ defined in \eqref{eq:cost_spike_thj} equals the objective of \eqref{eq:changepoint_form} applied to data $y'(\phi)$, subject to the constraint that $\hat\tau_j$ is in the solution. 

{
\setlength{\belowdisplayskip}{1pt}
\setlength{\abovedisplayskip}{1pt}
\begin{small}
\begin{align*}
C(\phi) &= \min_{\alpha\geq 0}\qty{\cost\left(y_{1:\thj}'(\phi),\alpha;\gamma\right)} + \min_{\alpha\geq 0}\qty{\cost\qty(y_{T:(\thj+1)}'(\phi),\alpha;1/\gamma)}   +\lambda  \\
  &\overset{a.}{=} \min_{\alpha\geq 0}\qty{\cost\qty(y_{1:\thj}'(\phi),\alpha;\gamma)} + \min_{\alpha\geq 0}\qty{\cost\qty(y_{(\thj+1):T}'(\phi),\alpha;\gamma)}   +\lambda  \\
 &\overset{b.}{=}  \min_{\substack{0=\tau_0<\tau_1<\ldots<\tau_k<\tau_{k+1}=\thj, k}}  \qty{  \sum_{j=0}^k \min_{\alpha\geq 0} \qty(\frac{1}{2} \sum_{t=\tau_j+1}^{\tau_{j+1}} \qty(y_t'(\phi)-\alpha \gamma^{t-\tau_{j+1}})^2 ) +\lambda k } \\
&+ \min_{\substack{\thj=\tilde{\tau}_0<\tilde{\tau}_1<\ldots<\tilde{\tau}_l<\tilde{\tau}_{l+1}=T, l }}  \qty{\sum_{j=0}^l \min_{\alpha\geq 0} \qty(\frac{1}{2} \sum_{t=\tilde{\tau}_{j}+1}^{\tilde{\tau}_{j+1}} \qty(y_{t}'(\phi)-\alpha \gamma^{t-\tilde{\tau}_{j+1}})^2 ) +\lambda l  } +\lambda \\
&\overset{c.}{=}  \min_{\substack{0=\tau_0<\tau_1<\ldots<\tau_k<\tau_{k+1}=\thj, k,\\
\thj=\tilde{\tau}_0<\tilde{\tau}_1<\ldots<\tilde{\tau}_l<\tilde{\tau}_{l+1}=T, l}}   \vast\{  \sum_{j=0}^k \min_{\alpha\geq 0} \qty(\frac{1}{2} \sum_{t=\tau_j+1}^{\tau_{j+1}} \qty(y_t'(\phi)-\alpha \gamma^{t-\tau_{j+1}})^2 ) \\
&+\sum_{j=0}^l \min_{\alpha\geq 0} \qty(\frac{1}{2} \sum_{t=\tilde{\tau}_{j}+1}^{\tilde{\tau}_{j+1}} \qty(y_{t}'(\phi)-\alpha \gamma^{t-\tilde{\tau}_{j+1}})^2 ) + \lambda (k+l) \vast\}+\lambda \\
 &\overset{d.}{=}  \min_{\substack{0=\tau_0<\tau_1<\ldots<\tau_k<\tau_{k+1}=T,  \\
 k, \thj \in \{\tau_1,\ldots, \tau_k\}}}   \qty{  \sum_{j=0}^k \min_{\alpha\geq 0} \qty(\frac{1}{2} \sum_{t=\tau_j+1}^{\tau_{j+1}} \qty(y_t'(\phi)-\alpha \gamma^{t-\tau_{j+1}})^2 ) +\lambda k }. 
 \end{align*}
\end{small}
}

Here, $a.$ follows from Lemma \ref{lemma:cost_time_reversal} and $b.$ follows from Lemma \ref{lemma:alt_cost}. Part $c.$ follows from combining the two minimization problems, and finally part $d.$ follows from treating $(k+l)$ as a new variable in the optimization problem.

Next, we show that $C'(\phi)$ defined in \eqref{eq:cost_no_spike_thj} equals the objective of \eqref{eq:changepoint_form} applied to data $y'(\phi)$, subject to the constraint that $\hat\tau_j$ is \emph{not} in the solution. 

\begin{small}
\begin{align*}
 C'(\phi) &= \min_{\alpha \geq 0}\qty{ \cost\qty(y_{1:\thj}'(\phi),\alpha;\gamma) +  \cost\qty(y_{T:(\thj+1)}'(\phi),\gamma\alpha;1/\gamma)}  \\
&\overset{a.}{=}  \min_{\alpha\geq 0}\vast\{ \min_{\substack{0=\tau_0<\tau_1<\ldots<\tau_k<\tau_{k+1}=\thj,  \\
\alpha_0,\ldots,\alpha_{k-1}\geq 0, \alpha_k = \alpha, k}} \qty{ \frac{1}{2} \sum_{j=0}^k \sum_{t=\tau_j+1}^{\tau_{j+1}} \qty(y_t'(\phi)-\alpha_j \gamma^{t-\tau_{j+1}})^2  +\lambda k } \\
&+  \min_{\substack{\thj=\tau_0<\tau_1<\ldots<\tau_k<\tau_{k+1}=T,  \\
\alpha_0,\ldots,\alpha_k\geq 0, \alpha_k = \gamma\alpha, k}}  \qty{\frac{1}{2} \sum_{j=0}^k  \sum_{t=\tau_j+1}^{\tau_{j+1}} \qty(y_{T+\thj+1-t}'(\phi)-\alpha_j (1/\gamma)^{t-\tau_{j+1}})^2  +\lambda k } \vast\}\\ 
&\overset{b.}{=}  \min_{\alpha\geq 0}\vast\{ \min_{\substack{0=\tau_0<\tau_1<\ldots<\tau_k<\tau_{k+1}=\thj,  \\
\alpha_0,\ldots,\alpha_{k-1}\geq 0, \alpha_k = \alpha, k}} \qty{ \frac{1}{2} \sum_{j=0}^k \sum_{t=\tau_j+1}^{\tau_{j+1}} \qty(y_t'(\phi)-\alpha_j \gamma^{t-\tau_{j+1}})^2  +\lambda k }  \\
&+  \min_{\substack{T=\tilde{\tau}_0>\tilde{\tau}_1>\ldots>\tilde{\tau}_k>\tilde{\tau}_{k+1}=\thj,  \\
\alpha_0,\ldots,\alpha_k\geq 0, \alpha_k = \gamma\alpha , k}}  \qty{\frac{1}{2} \sum_{j=0}^k  \sum_{t=\tilde{\tau}_{j+1}}^{\tilde{\tau}_{j}+1} \qty(y_{t}'(\phi)- \alpha_j \gamma^{t-\tilde{\tau}_{j+1}})^2  +\lambda k } \vast\}   \\
&\overset{c.}{=} \min_{\substack{
0=\tau_0<\tau_1<\ldots<\tau_k<\tau_{k+1}=\thj,  \\
\alpha_0,\ldots,\alpha_{k-1}\geq 0, \alpha_k = \alpha, k, \\
\thj=\tilde{\tau}_{\tilde{k}+1}<\tilde{\tau}_{\tilde{k}}<\ldots<\tilde{\tau}_1<\tilde{\tau}_{0}=T\,  \\
\tilde{\alpha}_0,\ldots,\tilde{\alpha}_{\tilde{k}}\geq 0, \tilde{\alpha}_{\tilde{k}} = \gamma\alpha , \tilde{k}
}}    \left\{ \frac{1}{2} \sum_{j=0}^k \sum_{t=\tau_j+1}^{\tau_{j+1}} \qty(y_t'(\phi)-\alpha_j \gamma^{t-\tau_{j+1}})^2  +\lambda k +  \frac{1}{2} \sum_{j=0}^{\tilde{k}}\sum_{t=\tilde{\tau}_{j}+1}^{\tilde{\tau}_{j+1}} \qty(y_{t}'(\phi)- \tilde{\alpha}_j \gamma^{t-\tilde{\tau}_{j+1}})^2  +\lambda \tilde{k}\right\} \\
&\overset{d.}{=}
 \min_{\substack{0=\tau_0<\tau_1<\ldots<\tau_k<\tau_{k+1}=T,  \\
\alpha_0,\ldots,\alpha_{k}\geq 0, k ,\\
\thj \neq \tau_j, \forall j = 1,\ldots,k.
}}    \left\{ \frac{1}{2} \sum_{j=0}^k \sum_{t=\tau_j+1}^{\tau_{j+1}} \qty(y_t'(\phi)-\alpha_j \gamma^{t-\tau_{j+1}})^2  +\lambda k \right\} \\
&\overset{e.}{=} \min_{\substack{0=\tau_0<\tau_1<\ldots<\tau_k<\tau_{k+1}=T,  \\
 k, \thj \notin \{\tau_1,\ldots, \tau_k\}}}   \qty{  \sum_{j=0}^k \min_{\alpha\geq 0} \qty(\frac{1}{2} \sum_{t=\tau_j+1}^{\tau_{j+1}} \qty(y_t'(\phi)-\alpha \gamma^{t-\tau_{j+1}})^2 ) +\lambda k }.
\end{align*}
\end{small}

Part $a.$ follows from expanding $\cost(\cdot)$ using Lemma~\ref{lemma:cost_alpha_k}. We then change the optimization variable in the second term from $\tau_j$ to $\tilde{\tau}_j = T+\thj-\tau_j$, which does not change the optimization problem because the mapping between $\tilde{\tau}_j$ and $\tau_j$ is invertible; re-indexing the summation completes part $b.$ Next, $c.$ follows from combining the two optimization problems. In step $d.$, we observe that the two constraints $\alpha_k = \alpha$ (i.e., fitted value at timepoint $\thj$ is $\alpha$) and $\tilde{\alpha}_{\tilde{k}} = \gamma\alpha$ (i.e., fitted value at timepoint $\thj+1$ is $\gamma\alpha$) are equivalent to a single constraint that $\thj$ is not a changepoint. Finally, step $e.$ follows from pulling the optimization over $\alpha_j$ inside the summation.

To summarize, we have proven that 
\begin{align}
\label{eq:c_phi_final}
 C(\phi) &=  \min_{\substack{0=\tau_0<\tau_1<\ldots<\tau_k<\tau_{k+1}=T,  \\
 k, \thj \in \{\tau_1,\ldots, \tau_k\}}}   \qty{  \sum_{j=0}^k \min_{\alpha\geq 0} \qty(\frac{1}{2} \sum_{t=\tau_j+1}^{\tau_{j+1}} \qty(y_t'(\phi)-\alpha \gamma^{t-\tau_{j+1}})^2 ) +\lambda k } ,
\end{align}
and 
\begin{align}
\label{eq:c_phi_prime_final}
 C'(\phi) &=  \min_{\substack{0=\tau_0<\tau_1<\ldots<\tau_k<\tau_{k+1}=T,  \\
 k, \thj \notin \{\tau_1,\ldots, \tau_k\}}}   \qty{  \sum_{j=0}^k \min_{\alpha\geq 0} \qty(\frac{1}{2} \sum_{t=\tau_j+1}^{\tau_{j+1}} \qty(y_t'(\phi)-\alpha \gamma^{t-\tau_{j+1}})^2 ) +\lambda k }.
\end{align}

By inspection of \eqref{eq:c_phi_final} and \eqref{eq:c_phi_prime_final}, we conclude that $\left\{\phi:  C(\phi)\leq  C'(\phi)\right\} = \left\{\phi:  \thj \in \mathcal{M}\qty(y'\qty(\phi))\right\}$, which completes the proof. 

We present the technical lemmas used in the proof below.

\begin{lemma}
\label{lemma:cost_alpha_k}
For $\cost\qty(y_{1:s},\alpha;\gamma)$ defined in \eqref{eq:def_cost}, we have
\begin{equation}
\cost\qty(y_{1:s},\alpha;\gamma) =  \min_{\substack{0=\tau_0<\tau_1<\ldots<\tau_k<\tau_{k+1}=s,  \\
\alpha_0,\ldots,\alpha_k\geq 0, \alpha_k = \alpha, k}} \qty{\frac{1}{2}\sum_{j=0}^k\sum_{t=\tau_j+1}^{\tau_{j+1}} \qty(y_t-\alpha_j \gamma^{t-\tau_{j+1}})^2 + \lambda k}.
\end{equation}
\end{lemma}
\begin{proof}

\begin{small}
\begin{align*}
\cost(y_{1:s},\alpha;\gamma) &\overset{a.}{=} \min_{0\leq\tau <s} \qty{ F(\tau) +  \frac{1}{2}   \left( \sum_{t=\tau+1}^s (y_t - \alpha \gamma^{t-s})^2   \right) + \lambda  } \\
&\overset{b.}{=} \min_{0\leq\tau <s} \Bigg\{ \min_{0=\tau_0<\tau_1<\ldots<\tau_k<\tau_{k+1}=\tau, k} \frac{1}{2}  \left( \sum_{j=0}^k \min_{\alpha\geq 0} \left\{ \sum_{t=\tau_j+1}^{\tau_{j+1}} \qty(y_t-\alpha \gamma^{t-\tau_{j+1}})^2 \right\} +\lambda k \right)  \\
&+  \frac{1}{2} \left( \sum_{t=\tau+1}^s \qty(y_t - \alpha \gamma^{t-s})^2   \right) + \lambda  \Bigg\}  \\
&\overset{c.}{=} \min_{\substack{0=\tau_0<\tau_1<\ldots<\tau_k<\tau_{k+1}=\tau<s,  \\
\alpha_0,\ldots,\alpha_k\geq 0, k, \tau}} \qty{ \frac{1}{2} \sum_{j=0}^k \sum_{t=\tau_j+1}^{\tau_{j+1}} \qty(y_t-\alpha_j \gamma^{t-\tau_{j+1}})^2 + \lambda k +  \frac{1}{2}\sum_{t=\tau_{k+1}+1}^s \qty(y_t-\alpha \gamma^{t-s})^2  + \lambda } \\
&\overset{d.}{=}  \min_{\substack{0=\tau_0<\tau_1<\ldots<\tau_k<\tau_{k+1}=s,  \\
\alpha_0,\ldots,\alpha_k\geq 0, \alpha_k = \alpha, k}}   \qty{\frac{1}{2}\sum_{j=0}^k\sum_{t=\tau_j+1}^{\tau_{j+1}} \qty(y_t-\alpha_j \gamma^{t-\tau_{j+1}})^2 + \lambda k}.
\end{align*}
\end{small}
Here, $a.$ follows from the definition in \eqref{eq:def_cost} and $b.$ follows from the definition of $F(\tau)$, the optimal cost of segmenting the first $\tau$ data points. Part $c.$ follows from pulling the $\min_{\alpha\geq 0}$ operation out of the summation, which is performed separately for each data segment $y_{(\tau_j+1):\tau_{j+1}}$. Finally, part $d.$ follows by inspection.
\end{proof}

\begin{lemma}
\label{lemma:alt_cost}
For $\cost\qty(y_{1:s},\alpha;\gamma)$ defined in \eqref{eq:def_cost}, we have
$$ \min_{\alpha\geq 0}\qty{\cost\qty(y_{1:s},\alpha;\gamma)} = \min_{0=\tau_0<\tau_1<\ldots<\tau_k<\tau_{k+1}=s, k}\qty{\sum_{j=0}^k \min_{\alpha\geq 0} \left\{\frac{1}{2} \sum_{t=\tau_j+1}^{\tau_{j+1}} \qty(y_t-\alpha \gamma^{t-\tau_{j+1}})^2 \right\} +\lambda k}.$$
 \end{lemma}
\begin{proof}

\begin{align*}
\min_{\alpha\geq 0}\qty{\cost\qty(y_{1:s},\alpha;\gamma)} &\overset{a.}{=} 
 \min_{\alpha\geq 0}\qty{\min_{\substack{0=\tau_0<\tau_1<\ldots<\tau_k<\tau_{k+1}=s,  \\
\alpha_0,\ldots,\alpha_k\geq 0, \alpha_k = \alpha, k}}   \qty{\frac{1}{2}\sum_{j=0}^k\sum_{t=\tau_j+1}^{\tau_{j+1}} \qty(y_t-\alpha_j \gamma^{t-\tau_{j+1}})^2 + \lambda k}} \\
 &=  \min_{\substack{0=\tau_0<\tau_1<\ldots<\tau_k<\tau_{k+1}=s,  \\
\alpha_0,\ldots,\alpha_k\geq 0, k}}   \qty{\frac{1}{2}\sum_{j=0}^k\sum_{t=\tau_j+1}^{\tau_{j+1}} \qty(y_t-\alpha_j \gamma^{t-\tau_{j+1}})^2 + \lambda k}  \\
 &\overset{b.}{=}   \min_{\substack{0=\tau_0<\tau_1<\ldots<\tau_k<\tau_{k+1}=s, k}}  \qty{\frac{1}{2}\sum_{j=0}^k \min_{\alpha\geq 0}\qty{ \frac{1}{2} \sum_{t=\tau_j+1}^{\tau_{j+1}} \qty(y_t-\alpha \gamma^{t-\tau_{j+1}})^2 }+ \lambda k}.
\end{align*}

Here, $a.$ follows from Lemma~\ref{lemma:cost_alpha_k}. $b.$ follows from noting that $\alpha_j$ can be minimized independently for each data segment $y_{(\thj+1):\hat\tau_{j+1}}$.
\end{proof}

\begin{lemma}
\label{lemma:cost_time_reversal}
For $\cost\qty(y_{1:s},\alpha;\gamma)$ defined in \eqref{eq:def_cost}, we have
$$ \min_{\alpha\geq 0}\qty{\cost\qty(y_{1:s},\alpha;\gamma)} = \min_{\alpha\geq 0}\qty{\cost\qty(y_{s:1},\alpha;1/\gamma)}.$$
 \end{lemma}

 \begin{proof}
\begin{align*}
\min_{\alpha\geq 0}\qty{\cost\qty(y_{1:s},\alpha;\gamma)} &\overset{a.}{=}\min_{\substack{0=\tau_0<\tau_1<\ldots<\tau_k<\tau_{k+1}=s,  \\
\alpha_0,\ldots,\alpha_k\geq 0, k}} \qty{\frac{1}{2}\sum_{j=0}^k\sum_{t=\tau_j+1}^{\tau_{j+1}} \qty(y_t-\alpha_j \gamma^{t-\tau_{j+1}})^2 + \lambda k} \\
&\overset{b.}{=}\min_{\substack{s=\tilde{\tau}_0>\tilde{\tau}_1>\ldots>\tilde{\tau}_k>\tilde{\tau}_{k+1}=0,  \\
\alpha_0,\ldots,\alpha_k\geq 0, k}} \qty{\frac{1}{2}\sum_{j=0}^k\sum_{t=s-\tilde{\tau}_j+1}^{s-\tilde{\tau}_{j+1}} \qty(y_t-\alpha_j \gamma^{t-(s-\tilde{\tau}_{j+1})})^2 + \lambda k} \\
&\overset{c.}{=} \min_{\substack{s=\tilde{\tau}_0>\tilde{\tau}_1>\ldots>\tilde{\tau}_k>\tilde{\tau}_{k+1}=0,  \\
\alpha_0,\ldots,\alpha_k\geq 0, k}} \qty{\frac{1}{2}\sum_{j=0}^k\sum_{t=\tilde{\tau}_{j+1}}^{\tilde{\tau}_{j}} \qty(y_{s-t}-\alpha_j (1/\gamma)^{t-\tilde{\tau}_{j+1}})^2 + \lambda k} \\
&\overset{d.}{=} \min_{\alpha\geq 0}\qty{\cost\qty(y_{s:1},\alpha;1/\gamma)}.
\end{align*}

Part $a.$ follows from Lemma~\ref{lemma:cost_alpha_k}. In step $b.$, we change the optimization variable from $\tau_j$ to $\tilde{\tau}_j = s-\tau_j$, which does not change the optimization problem because the mapping between $\tilde{\tau}_j$ and $\tau_j$ is invertible. Step $c.$ follows from re-indexing the summation,  and finally $d.$ follows from Lemma~\ref{lemma:cost_alpha_k} again.
 \end{proof}

%% file: appendix_a_5.tex

\subsection{Proof of Proposition~\ref{prop:bivariate_cost_recursion}}
\label{appendix:prop_bivariate}

To begin, we will prove \eqref{eq:forward_fpop}  using an induction argument.  
The following claim serves as the ``base case" for the recursion.
\begin{lemma}
\label{lemma:base}
\begin{equation}
\cost\left(y_{1:(\thj-h+1)}'(\phi),\alpha;\gamma\right) = \min_{f \in \mathcal{C}_{\thj-h+1}} f(\alpha,\phi),
\label{eq:lemma:base}
\end{equation}
where 
\begin{small}
\begin{align}
\mathcal{C}_{\thj-h+1} = \Bigg\{\cost\left(y_{1:(\thj-h)}'(\phi),\alpha/\gamma;\gamma\right)+\frac{1}{2}\left(y_{\thj-h+1}'(\phi)-\alpha\right)^2, \nonumber \\
\min_{\alpha'\geq 0}\qty{\cost\left(y_{1:(\thj-h)}'(\phi),\alpha';\gamma\right)} +\lambda +\frac{1}{2}\left(y_{\thj-h+1}'(\phi)-\alpha\right)^2\Bigg\}.
\label{eq:base_case_c_set}
\end{align}
\end{small}
\end{lemma}

\begin{proof}
To prove Lemma~\ref{lemma:base}, we will first compute 
$\cost\left(y_{1:(\thj-h+1)}'(\phi),\alpha;\gamma\right)$ using the definition in \eqref{eq:def_cost}; we will then show that this equals $\min_{f \in \mathcal{C}_{\thj-h+1}} f(\alpha,\phi)$, with $\mathcal{C}_{\thj-h+1}$ in \eqref{eq:base_case_c_set}.

Per the definition of $\nu$ in \eqref{eq:nu_def}, $y'_{1:(\thj-h)}(\phi)= y_{1:(\thj-h)}$; therefore, $\mathcal{C}_{\thj-h} =  \cost\left(y_{1:(\thj-h)}'(\phi),\alpha;\gamma\right)  =  \cost\left(y_{1:(\thj-h)},\alpha;\gamma\right)$. From Proposition~\ref{prop:cost_recursion},  this means that $\mathcal{C}_{\thj-h}$ is a piecewise quadratic function of $\alpha$ only.

Now we consider the function $\cost\left(y_{1:(\thj-h+1)}'(\phi),\alpha;\gamma\right)$. There are two possibilities:
\begin{enumerate}
\item
\emph{There is no changepoint at the $(\thj-h)$th time step.} In this case, $\cost\left(y_{1:(\thj-h+1)}'\left(\phi\right),\alpha;\gamma\right)$ equals 
\begin{align*}
\cost\left(y_{1:(\thj-h)}'(\phi),\alpha/\gamma;\gamma\right) + \frac{1}{2}\left(y_{\thj-h+1}'\qty(\phi)-\alpha\right)^2,
\end{align*} where $\alpha/\gamma$ accounts for the exponential calcium decay.

\item
\emph{There is a changepoint at the $(\thj-h)$th time step.} In this case, $\cost\left(y_{1:(\thj-h+1)}'(\phi),\alpha;\gamma\right)$ equals 
\begin{align*}
\min_{\alpha'\geq 0}\qty{\cost\left(y_{1:(\thj-h)}'(\phi),\alpha';\gamma\right)} + \lambda +  \frac{1}{2}\left(y_{\thj-h+1}'(\phi)-\alpha\right)^2,
\end{align*} where the changepoint incurs a penalty of $\lambda$, and there can be an arbitrary change in the calcium from timepoint $\thj-h$ to $\thj-h+1$.
\end{enumerate}

Therefore, 
\begin{align}
\cost\left(y_{1:(\thj-h+1)}'(\phi),\alpha;\gamma\right) =& \min \Bigg\{ \cost\left(y_{1:(\thj-h)}'(\phi),\alpha/\gamma;\gamma\right) + \frac{1}{2}\left(y_{\thj-h+1}'(\phi)-\alpha\right)^2, \nonumber \\
& \min_{\alpha'\geq 0}\qty{\cost\left(y_{1:(\thj-h)}'(\phi),\alpha';\gamma\right)} + \lambda +  \frac{1}{2}\left(y_{\thj-h+1}'(\phi)-\alpha\right)^2
 \Bigg\} \nonumber \\
 =& \min_{f\in \mathcal{C}_{\thj-h+1}} f(\alpha,\phi), 
 \label{eq:first_step_induction}
\end{align}
where the last equality follows from 
 \eqref{eq:base_case_c_set}. This completes the proof. 
\end{proof}

We will now prove the inductive step for the recursion, which relies on the following claim.
\begin{lemma}
Suppose that for some $s\in \left\{ \thj-h+1,\ldots,\thj-1 \right\}$, 
\begin{align}
\label{eq:inductive_hypo}
\cost\left(y_{1:s}'(\phi),\alpha;\gamma\right) = \min_{f\in\mathcal{C}_{s}} f(\alpha,\phi).
\end{align}
Then, 
\begin{align}
\cost\left(y_{1:(s+1)}'(\phi),\alpha;\gamma\right) = \min_{f\in\mathcal{C}_{s+1}} f(\alpha,\phi),
\end{align}
where $\mathcal{C}_{s+1}$ is defined recursively according to \eqref{eq:union_update_bivariate}.
\label{lemma:inductive}
\end{lemma}

\begin{proof}
To begin, we apply Proposition~\ref{prop:cost_recursion} with $y'(\phi)$ instead of $y$ and get 
\begin{equation}
\label{eq:step_1}
\cost\left(y_{1:(s+1)}'(\phi),\alpha;\gamma\right) = \min \left\{ \cost\left(y_{1:s}'(\phi),\alpha/\gamma;\gamma\right), \min_{\alpha'\geq 0}\cost\left(y_{1:s}'(\phi),\alpha';\gamma\right)+\lambda \right\} + \frac{1}{2}\left(y_{s+1}'(\phi)-\alpha\right)^2 .
\end{equation}
Applying the inductive hypothesis in \eqref{eq:inductive_hypo} with $\alpha/\gamma$ instead of $\alpha$, we have that
\begin{align}
\cost\left(y_{1:s}'(\phi),\alpha/\gamma;\gamma\right) = \min_{f\in\mathcal{C}_s} f(\alpha/\gamma,\phi),
\end{align} and
\begin{align}
\min_{\alpha'\geq 0}\qty{\cost\left(y_{1:s}'(\phi),\alpha';\gamma\right)} = \min_{\alpha'\geq 0} \qty{\min_{f\in\mathcal{C}_s} f(\alpha',\phi)}.
\end{align}

Therefore,
\begin{small}
\begin{align}
\cost\left(y_{1:(s+1)}'(\phi),\alpha;\gamma\right) &\overset{a.}{=} \min \left\{  \min_{f\in\mathcal{C}_s} f(\alpha/\gamma,\phi), \min_{\alpha'\geq 0} \qty{\min_{f\in\mathcal{C}_s} f(\alpha',\phi)} + \lambda \right\} + \frac{1}{2}\left(y_{s+1}'(\phi)-\alpha\right)^2  \\
&\overset{b.}{=}  \min \qty{  \min_{f\in\mathcal{C}_s} f(\alpha/\gamma,\phi)  + \frac{1}{2}\left(y_{s+1}'(\phi)-\alpha\right)^2 ,  \min_{f\in\mathcal{C}_s} \qty{ \min_{\alpha'\geq 0}\qty{f(\alpha',\phi)}} + \lambda  + \frac{1}{2}\left(y_{s+1}'(\phi)-\alpha\right)^2 }, 
\label{eq:induction_step_1}
\end{align}
\end{small}
where $a.$ follows from \eqref{eq:inductive_hypo} and \eqref{eq:step_1}, and $b.$ follows from exchanging the order of minimization and distributing the $\frac{1}{2}\left(y_{s+1}'(\phi)-\alpha\right)^2$ term inside.

Furthermore, 
\begin{align}
\min_{f\in\mathcal{C}_{s+1}} f(\alpha,\phi) &\overset{a.}{=}  \min_{f\in \qty{\left( \bigcup_{f \in \mathcal{C}_{s}} \left\{ f(\alpha/\gamma,\phi)+\frac{1}{2}\left(y'_{s+1}(\phi)-\alpha\right)^2  \right\} \right) \bigcup \left\{ g_{s+1}(\phi)+\frac{1}{2}\left(y'_{s+1}(\phi)-\alpha\right)^2  \right\} } } f(\alpha,\phi) \\
&\overset{b.}{=} \min \qty{ \min_{f \in \mathcal{C}_{s}} \left\{ f(\alpha/\gamma,\phi)+\frac{1}{2}\left(y'_{s+1}(\phi)-\alpha\right)^2  \right\},  g_{s+1}(\phi)+\frac{1}{2}\left(y'_{s+1}(\phi)-\alpha\right)^2  } \\
&\overset{c.}{=} \min \qty{ \min_{f \in \mathcal{C}_{s}} \left\{ f(\alpha/\gamma,\phi) \right\} +\frac{1}{2}\left(y'_{s+1}(\phi)-\alpha\right)^2 ,   \min_{f\in\mathcal{C}_{s}}\qty{\min_{\alpha\geq 0} \qty{f(\alpha,\phi)}} +\lambda +\frac{1}{2}\left(y'_{s+1}(\phi)-\alpha\right)^2  },  \label{eq:induction_step_2}
\end{align}
where $a.$ follows from the definition of $\mathcal{C}_{s+1}$ in \eqref{eq:union_update_bivariate}; $b.$ follows from noting that $\min_{f\in A \bigcup B} f = \min \qty{ \min_{f\in A} f , \min_{f\in B} f }$; and $c.$ follows from the definition of $g_{s+1}(\phi)$ in \eqref{eq:g_update}. 

Now by inspection, \eqref{eq:induction_step_1} is equal to \eqref{eq:induction_step_2}; this completes the proof. 
\end{proof}

The inductive proof of \eqref{eq:forward_fpop} follows directly from combining Lemmas~\ref{lemma:base} and \ref{lemma:inductive}.

We will now show that  for $\thj-h+1 \leq s \leq \thj$, $\mathcal{C}_{s}$ is a collection of piecewise quadratic functions of $\alpha$ and $\phi$. We will show this by induction. 
We first make the following observations, which follow from simple algebra:
\begin{itemize}
	\item
    \emph{Observation 1: }For $\thj-h+1 \leq s \leq \thj$, $\frac{1}{2}(y'_s(\phi)-\alpha)^2$ is a quadratic function of $\phi$ and $\alpha$, where $y'(\phi)$ is defined in \eqref{eq:phi}.
	\item
	\emph{Observation 2: }If both $f_1(\alpha,\phi)$ and $f_2(\alpha,\phi)$ are  piecewise quadratic functions of $\alpha$ and $\phi$, then $f_1+f_2$ is  also a piecewise quadratic function of $\alpha$ and $\phi$.
	\item
	\emph{Observation 3: }If $f(\alpha,\phi)$ is a piecewise quadratic function of $\phi$ and $\alpha$, then $\min_{\alpha \geq 0} f(\alpha,\phi)$ is a piecewise quadratic function of only $\phi$.
	\item
	\emph{Observation 4: }If $\mathcal{C}_s$ is a finite set of piecewise quadratic functions of $\phi$ and $\alpha$, then $\min_{f\in\mathcal{C}_s} f(\alpha,\phi)$ is  a piecewise quadratic function of $\phi$ and $\alpha$.
\end{itemize}
In our induction, Lemma~\ref{lemma:pwq_base_case} serves as our ``base case". The induction step is presented in Lemma~\ref{lemma:pwq_induction_step}.
\begin{lemma}
\label{lemma:pwq_base_case}
$\mathcal{C}_{\thj-h+1}$ is a collection of piecewise quadratic functions of $\alpha$ and $\phi$. 
\end{lemma}
\begin{proof}
Applying the recursion  in \eqref{eq:union_update_bivariate}, we see that \begin{small}
\begin{align*}
\mathcal{C}_{\thj-h+1} = \Bigg\{\cost\qty(y_{1:(\thj-h)}'(\phi),\alpha;\gamma)+\frac{1}{2}\qty(y_{\thj-h+1}'(\phi)-\alpha)^2,  \\
\min_{\alpha\geq 0}\qty{\cost\qty(y_{1:(\thj-h)}'(\phi),\alpha;\gamma)} +\lambda +\frac{1}{2}\qty(y_{\thj-h+1}'(\phi)-\alpha)^2\Bigg\}.
\end{align*}
\end{small}
By Proposition~\ref{prop:cost_recursion},  $\cost(y_{1:(\thj-h)}'(\phi),\alpha;\gamma)=\cost(y_{1:(\thj-h)},\alpha;\gamma)$ is a piecewise quadratic function of $\alpha$. Furthermore, $\frac{1}{2}(y_{\thj-h+1}'(\phi)-\alpha)^2$ is a quadratic function of $\phi$ and $\alpha$, according to Observation 1. Therefore, the first term in $\mathcal{C}_{\thj-h+1}$ is a piecewise quadratic function of $\phi$ and $\alpha$ according to Observation 2. As for the second term, we note that $\min_{\alpha\geq 0}\qty{\cost(y_{1:(\thj-h)}'(\phi),\alpha;\gamma)}$ is a piecewise quadratic function of $\phi$ according to Observation 3, so its sum with $\lambda +\frac{1}{2}(y_{\thj-h+1}'(\phi)-\alpha)^2$ is piecewise quadratic in  $\phi$ and $\alpha$.
\end{proof} 

\begin{lemma}
\label{lemma:pwq_induction_step}
Suppose that for some $s \in \{ \thj-h+1,\ldots,\thj-1\} $, $\mathcal{C}_{s}$ is a collection of piecewise quadratic functions of $\alpha$ and $\phi$. Then, 
\begin{align}
\mathcal{C}_{s+1} = \left( \bigcup_{f \in \mathcal{C}_{s}} \left\{ f(\alpha/\gamma,\phi)+\frac{1}{2}\qty(y'_{s+1}(\phi)-\alpha)^2  \right\} \right) \bigcup \left\{ g_{s+1}(\phi)+\frac{1}{2}\qty(y'_{s+1}(\phi)-\alpha)^2  \right\}
\label{eq:c_s_1_induction}
\end{align}
 is also a collection of piecewise quadratic functions of $\alpha$ and $\phi$, where $g_{s+1}$ is defined in \eqref{eq:g_update}.
\end{lemma}
\begin{proof}
According to the induction hypothesis, each $f \in \mathcal{C}_{s}$ is a piecewise quadratic function of $\alpha$ and $\phi$. Therefore, $f(\alpha/\gamma,\phi)+\frac{1}{2}(y'_{s+1}(\phi)-\alpha)^2 $ is a piecewise quadratic function of $\alpha$ and $\phi$ for all $f \in \mathcal{C}_{s}$, according to Observation 2. Furthermore, from Observations 3 and 4, we can see that 
\begin{align*}
g_{s+1}(\phi) = \min_{f\in\mathcal{C}_{s}}\min_{\alpha\geq 0} f(\alpha,\phi) +\lambda
\end{align*} is a piecewise quadratic function of $\phi$.
\end{proof}

Combining Lemmas~\ref{lemma:pwq_base_case} and \ref{lemma:pwq_induction_step} completes the argument that  for $s\in \{ \thj-h,\ldots,\thj\} $, $\mathcal{C}_{s}$ is a collection of piecewise quadratic functions. 

To complete the proof of Proposition~\ref{prop:bivariate_cost_recursion}, it remains to show that for $s\in \{ \thj-h,\ldots,\thj\} $, $|\mathcal{C}_{s}|=s-\thj+h+1$. According to \eqref{eq:recursion_init}, $\mathcal{C}_{\thj-h}$ consists of a single function. At each iteration of the recursion in \eqref{eq:union_update_bivariate}, only one additional function is added; therefore, $\mathcal{C}_{s}$ consists of $1+s-(\thj-h)= s-\thj+h+1$ functions.

%% file: appendix_a_6_7.tex

\subsection{Extension of Proposition \ref{prop:bivariate_cost_recursion} to $y_{T:(\thj+1)}'(\phi)$}
\label{appendix:recursion_reverse}
The following proposition is a straightforward extension of Proposition \ref{prop:bivariate_cost_recursion} to the sequence $y_{T:(\thj+1)}'(\phi)$ with decay parameter $1/\gamma$ to account for the time reversal. 

\begin{Proposition}
\label{prop:bivariate_cost_recursion_reverse}
For $\thj+1 \leq s \leq \thj+h$, 
\begin{align}
{\cost}\qty(y_{T:s}'(\phi),\alpha;1/\gamma) = \min_{f\in\tilde{\mathcal{C}}_s} f(\alpha,\phi),
\label{eq:backward_fpop}
\end{align}
where $\tilde{\mathcal{C}}_s$ is a collection of $\thj+h+2-s$ piecewise quadratic functions of $\alpha$ and $\phi$, $f(\alpha,\phi)$, constructed with the initialization  
\begin{align}
\tilde{\mathcal{C}}_{\thj+h+1} = \left\{ \cost\qty(y_{T:(\thj+h+1)}'(\phi),\alpha;1/\gamma) \right\},
\label{eq:recursion_init_reverse}
\end{align} 
and the recursion
\begin{align}
\label{eq:union_update_bivariate_reverse}
\tilde{\mathcal{C}}_{s} = \left( \bigcup_{f \in \tilde{\mathcal{C}}_{s+1}} \left\{ f(\alpha\gamma,\phi)+\frac{1}{2}\qty(y'_s(\phi)-\alpha)^2  \right\} \right) \bigcup \left\{ g_s(\phi)+\frac{1}{2}\qty(y'_s(\phi)-\alpha)^2  \right\} \, ,
\end{align} 
where \begin{align}
g_s(\phi) = \min_{f\in\tilde{\mathcal{C}}_{s+1}}\min_{\alpha\geq 0} f(\alpha,\phi) +\lambda
\end{align}
and $y'(\phi)$ is defined in \eqref{eq:phi}. 
\end{Proposition}

\subsection{General case for Propositions \ref{prop:bivariate_cost_recursion} and \ref{prop:bivariate_cost_recursion_reverse}}
\label{appendix:general_tau_L_tau_R}

Propositions \ref{prop:bivariate_cost_recursion} and \ref{prop:bivariate_cost_recursion_reverse} assumed that $\thj-h \geq 1$ and $\thj+h+1 \leq T$ (where $T$ is the length of the observed data), respectively. We now provide details for the cases where $\thj-h < 1$ and $\thj+h+1 > T$.

\begin{itemize}
	\item
	\emph{Case 1: $\thj-h< 1$.} Define $\hat\tau_L = \max\{1, \thj-h \}$ and initialize with \begin{align}
\mathcal{C}_{\hat\tau_L} = \left\{ {\cost}\qty(y_{1:\hat\tau_L}'(\phi),\alpha;\gamma) \right\}
\end{align} in Proposition \ref{prop:bivariate_cost_recursion} instead of \eqref{eq:recursion_init}, with the convention $y_{1:1}'(\phi) = y_1'(\phi)$.
	\item
	\emph{Case 2: $\thj+h+1 > T$.} Define $\hat\tau_R = \min\{T, \thj+h+1 \}$ and initialize with \begin{align}
\mathcal{C}_{\hat\tau_R} = \left\{ {\cost}\qty(y_{T:\hat\tau_R}'(\phi),\alpha;1/\gamma) \right\}
\end{align} in Proposition \ref{prop:bivariate_cost_recursion_reverse} instead of \eqref{eq:recursion_init_reverse}, with the convention $y_{T:T}'(\phi) = y_T'(\phi)$.

\end{itemize}

%% file: appendix_a_8_9.tex

\subsection{Algorithm for computing $\mathcal{S}$ in \eqref{eq:S_in_phi}}
\label{appendix:algorithm_s}
\begin{algorithm}[!htp]
  \SetKwInOut{Input}{Input}
  \SetKwInOut{Output}{Output}
  \Input{Data $y_{1:T}$, spike location $\thj$, exponential decay parameter $\gamma$}
  \Output{Set $\mathcal{S}$}
  \begin{enumerate}
  	\itemsep0em 

    \item
     Compute the collection of functions $\mathcal{C}_{\thj}$ using Proposition \ref{prop:bivariate_cost_recursion}.

    \item
     Compute the collection of functions  $\tilde{\mathcal{C}}_{\thj+1}$ using Proposition \ref{prop:bivariate_cost_recursion_reverse}.
  
	 \item
     Compute $C(\phi)$ using \eqref{eq:Cphi}.
     \item
 
     Compute $C'(\phi)$ using \eqref{eq:Cphi_prime}.

     \item
     Compute $\mathcal{S}  = \{\phi: C(\phi) \leq C'(\phi)\}$.
 \end{enumerate}
 \caption{Computing $\mathcal{S}$ in \eqref{eq:S_in_phi} for a spike $\hat\tau_j$ resulting from \eqref{eq:l0-opt}}
 \label{alg:S_computation}
\end{algorithm}

\subsection{Proof of Proposition \ref{prop:S_set_timing}}
\label{appendix:computation}

Throughout the proof, we assume that the number of pieces in the piecewise quadratic functions under consideration is a constant that does not depend on $h$ and $T$. Moreover, we will leverage the toolkit from \citet{Maidstone2017-vc,Rigaill2015-pm,Jewell2019-vv}, which allows for efficient manipulation of both univariate and bivariate piecewise quadratic functions. Provided with an efficient implementation of the toolkit, we make the following two observations for our timing complexity analysis:  
\begin{itemize}
	\item
	\emph{Observation 1: } $\min_{f\in\mathcal{C}} f(\phi)$ can be computed in $O(|\mathcal{C}|)$ operations, provided that $f(\phi)$ is a piecewise quadratic function of $\phi$; \label{eq:fp_1}
	\item
    \emph{Observation 2: } $\forall f_1, f_2 \in \mathcal{C}, \, f_1(\alpha,\phi)+f_2(\alpha,\phi)$ can be computed in $O(1)$ operations, provided that $f_1(\alpha,\phi)$ and $f_2(\alpha,\phi)$ are piecewise quadratic functions of $\alpha$ and $\phi$ with $O(1)$ pieces. \label{eq:fp_2}
\end{itemize}
Finally, we recall that if $f(\alpha,\phi)$ is a piecewise quadratic function of $\alpha$ and $\phi$, then $\min_{\alpha\geq 0} \{f(\alpha,\phi)\}$ is a piecewise quadratic function of $\phi$ only and can be computed analytically.

Now we will characterize the computational complexity of Algorithm~\ref{alg:S_computation}:
\begin{enumerate}
	\item
	Step 1: We first consider the time  to compute $\mathcal{C}_s$ for some $s\in \{\thj-h+1,\ldots, \thj\}$, assuming that we have computed $\mathcal{C}_{s-1}$. 
	\begin{enumerate}
		\item
			We first compute $\bigcup_{f\in \mathcal{C}_{s-1}}\{f(\alpha/\gamma,\phi)+ \frac{1}{2}(y_s'(\phi)-\alpha)^2\}$, which takes $O(|\mathcal{C}_{s-1}|) = O(s-\thj+h)$ operations. 
		\item
			We then compute $g_s(\phi)$ using \eqref{eq:g_update}: the inner minimization over $\alpha\geq 0$ takes $O(1)$ operations for each $f \in \mathcal{C}_{s-1}$ since it admits an analytical solution; the outer minimization over $\mathcal{C}_{s-1}$ takes $O(|\mathcal{C}_{s-1}|)= O(s-\thj+h)$ operations according to Observation 1. 
	\end{enumerate}

	In summary, computing $\mathcal{C}_s$ from $\mathcal{C}_{s-1}$ takes $O(s-\thj+h)$ operations for any $s\in \{\thj-h+1,\ldots, \thj\}$. The first step of Algorithm~\ref{alg:S_computation} requires computing $\mathcal{C}_s$ for all $s\in \{\thj-h+1,\ldots, \thj\}$, a total of $O\qty(\sum_{t=\thj-h+1}^{\thj} \qty(t-\thj+h)) = O(h^2)$ operations.

	\item

	Step 2: Applying the same logic used in analyzing Step 1 to the second step of Algorithm~\ref{alg:S_computation}, we conclude that computing $\tilde{\mathcal{C}}_{\thj+1}$ takes $O(h^2)$ operations using Proposition \ref{prop:bivariate_cost_recursion_reverse}.

	\item

    According to \eqref{eq:Cphi}, computing $C(\phi)$ requires $\min_{f \in \mathcal{C}_{\thj}} \left\{ \min_{\alpha\geq 0}f(\alpha,\phi) \right\}$  and $\min_{f \in \tilde{\mathcal{C}}_{\thj+1}}  \left\{\min_{\alpha' \geq 0}  f(\alpha' ,\phi) \right\}$. Both terms can be computed in $O(|\mathcal{C}_{\thj}|) = O(h)$ operations using Observation 1; moreover, the summation will take $O(1)$ operations according to Observation 2. Hence Step 3 takes $O(h)$ operations in total.

	\item
    According to \eqref{eq:Cphi_prime}, 
    $$C'(\phi) = \min_{ f\in  \mathcal{C}_{\thj}, \tilde{f}\in\tilde{\mathcal{C}}_{\thj+1}} \qty{ \min_{\alpha\geq 0} \qty{  f(\alpha,\phi)+\tilde{f}(\gamma\alpha,\phi) } }.$$

    \begin{enumerate}
    	\item 
   Computing the set $\qty{  f(\alpha,\phi)+\tilde{f}(\gamma\alpha,\phi)\;\middle\vert\; f\in  \mathcal{C}_{\thj}, \tilde{f}\in\tilde{\mathcal{C}}_{\thj+1} }$ takes $O(|\mathcal{C}_{\thj}|\cdot|\tilde{\mathcal{C}}_{\thj+1}|) = O(h^2)$ operations, since each addition takes $O(1)$ operations (Observation 2) and there are $|\mathcal{C}_{\thj}|\cdot|\tilde{\mathcal{C}}_{\thj+1}|$ such sums. 

    	\item

    	Minimizing over $\alpha\geq0$ for each $f(\alpha,\phi)+\tilde{f}(\gamma\alpha,\phi)$ takes $O(1)$ operations, so the cost of minimization over the entire collection is $O(h^2)$. 

    	\item

    	Computing $C'(\phi)$ as the minimum of $O(h^2)$ piecewise quadratic functions of $\phi$ requires $O(h^2)$ operations by Observation 1.

    \end{enumerate}

	To summarize, we need $O(h^2)$ operations to compute $C'(\phi)$. 

	\item

	To carry out Step 5, we first compute $\min\{C(\phi),C'(\phi)\}$, the minimum of two piecewise quadratic functions of $\phi$ only, which takes $O(1)$ operations by Observation 1. In $O(1)$ operations, we can obtain $\mathcal{S}$ in \eqref{eq:S_set} by computing the set of $\phi$ such that $\min\{C(\phi),C'(\phi)\} = C(\phi)$.

\end{enumerate}

To summarize, computing $\mathcal{S}$ defined in \eqref{eq:S_set} using Algorithm~\ref{alg:S_computation} takes $O(h^2)$ operations.

\subsection{Empirical timing results for Proposition \ref{prop:S_set_timing}}

\label{appendix:empirical_computation}

In this section, we investigate the claim from Proposition \ref{prop:S_set_timing} that computing the set $\mathcal{S}$ defined in \eqref{eq:S_set} requires $O(h^2)$ operations, where $h$ is the window size that appears in \eqref{eq:nu_def}. 

Figure \ref{fig:supp_time} displays the running time, computed on a MacBook Pro with a 1.4 GHz Intel Core i5 processor, as a function of the window size, $h$, over 50 replicate datasets simulated according to \eqref{eq:obs-model} with $T=10,000$, $\gamma = 0.98$, and $z_t \overset{\text{i.i.d.}}{\sim} \text{Poisson}(0.01)$; the tuning parameter $\lambda$ for the $\ell_0$ problem in \eqref{eq:l0-opt} is set to $0.3$, which yields between 50 and 100 spikes. With $h=20$, the average running time is 2.1 seconds for each dataset. In addition, a quadratic fit is plotted for reference. We see that the running time is indeed approximately quadratic in the window size $h$.

%% file: appendix_a_10.tex

\subsection{An illustrative example for Propositions~\ref{prop:bivariate_cost_recursion} and \ref{prop:bivariate_cost_recursion_reverse} }
\label{appendix:example}

In this section, we walk through a very simple example of characterizing the set $\mathcal{S}=\left\{\phi: \thj \in \mathcal{M}(y'(\phi))\right\}$ in \eqref{eq:S_set} using Proposition~\ref{prop:characterization_S}.

Suppose $y_{1:4}=(8,4,6,3)$, and we want to compute $\mathcal{S}$ for $\thj=2$ with $h=1$ (i.e., $\thj-h=1, \thj+h=3$), $\gamma=\frac{1}{2}$, and $\lambda = 1$. We first compute $\nu$ according to \eqref{eq:nu_def} and $y'(\phi)$ according to \eqref{eq:phi}:
\begin{align*} 
\nu = \begin{pmatrix}
0 \\
-0.5 \\
1 \\
0
\end{pmatrix},  \quad 
y'(\phi) = \begin{pmatrix}
8 \\
5.6-0.4\phi \\
2.8+0.8\phi \\
3
\end{pmatrix}.
\end{align*}

According to \eqref{eq:S_in_phi}, to compute $\mathcal{S}$, it suffices to compute $C(\phi)$ in \eqref{eq:cost_spike_thj} and $C'(\phi)$ in \eqref{eq:cost_no_spike_thj}. 

We first compute $C(\phi)$ using Proposition~\ref{prop:bivariate_cost_recursion}. We start with $\mathcal{C}_{\thj-h} = \mathcal{C}_1$.

\begin{enumerate}
	\item
	$\mathcal{C}_1$ has only one function 
	\begin{align*}
	\mathcal{C}_1 = \cost\qty(y_1'(\phi),\alpha;\gamma) = \frac{1}{2}\qty(8-\alpha)^2.
	\end{align*}

	\item
	To compute $\mathcal{C}_2$, we apply \eqref{eq:union_update_bivariate}:
	\begin{align*}
	\mathcal{C}_2 = \Big\{ \frac{1}{2}\qty(8-\alpha/0.5)^2 + \frac{1}{2}\qty(5.6-0.4\phi-\alpha)^2, \\
	\frac{1}{2}\qty(5.6-0.4\phi-\alpha)^2 + g_2(\phi) \Big\},
	\end{align*} 
	where 
	\begin{align*}
	g_2(\phi)= \min_{\alpha\geq 0} \cost\qty(y_1,\alpha;\gamma) + \lambda = 0 + \lambda = 1.
	\end{align*} 
\end{enumerate}

This completes the calculation $$\cost\qty(y'_{1:\thj}(\phi),\alpha;\gamma) = \cost\qty(y'_{1:2}(\phi),\alpha;\gamma) = \min_{f\in \mathcal{C}_2 } f(\alpha,\phi).$$ 

For the reverse direction, we will apply Proposition~\ref{prop:bivariate_cost_recursion_reverse} to compute sets $\mathcal{C}_4$ and $\mathcal{C}_3$.

\begin{enumerate}
\item
$\mathcal{C}_4$ consists of a single function:
\begin{align*}
\mathcal{C}_4 = \cost\qty(y_4'(\phi),\alpha;1/\gamma) = \frac{1}{2}\qty(3-\alpha)^2.
\end{align*}

\item
Applying \eqref{eq:union_update_bivariate_reverse}, we get 
\begin{align*} 
\mathcal{C}_3 = \min\Big\{  \frac{1}{2}(3-\alpha/2)^2 +  \frac{1}{2}(2.8+0.8\phi-\alpha)^2,  \\
\min_{\alpha'\geq 0}\qty{\frac{1}{2}(3-\alpha'/2)^2} + \lambda + \frac{1}{2}(2.8+0.8\phi-\alpha)^2
  \Big\},
\end{align*}
which yields $$\cost\qty(y'_{T:\thj+1}(\phi),\alpha;1/\gamma) = \cost\qty(y'_{4:3}(\phi),\alpha;1/\gamma) = \min_{f\in \mathcal{C}_3 } f(\alpha,\phi).$$
\end{enumerate}

According to \eqref{eq:cost_spike_thj},
\begin{align*}
C(\phi) = \min_{\alpha\geq 0} \big\{ \cost\qty(y'_{1:2}(\phi),\alpha;\gamma) \big\} + \min_{\alpha\geq 0} \big\{ \cost\qty(y'_{4:3}(\phi),\alpha; 1/\gamma)\big\} + \lambda,
\end{align*}
where
\begin{small}
\begin{align*}
&\min_{\alpha\geq 0} \qty{ \cost\qty(y'_{1:2}(\phi),\alpha;\gamma) } = \min\qty{ \min_{\alpha\geq 0} \qty{\frac{1}{2}\qty(8-\alpha/(0.5))^2 + \frac{1}{2}\qty(5.6-0.4\phi-\alpha)^2}, 
 \min_{\alpha\geq 0}\qty{ \frac{1}{2}\qty(5.6-0.4\phi-\alpha)^2 +1 } } \\
 &= \min\vast\{
 \tilde{f}_1(\phi) = \begin{cases} 
      0.064\phi^2-0.512\phi+1.024 & \phi\leq54 \\
      0.08\phi^2-2.24\phi+47.68 & \phi > 54
   \end{cases},
   \tilde{f}_2(\phi) =  \begin{cases} 
      1 & \phi\leq14 \\
      0.08\phi^2-2.24\phi+16.68 & \phi > 14
   \end{cases}  \vast\} \\
 &=  \begin{cases} 
    1  & \phi\leq 0.047\\
 0.064\phi^2-0.512\phi+1.024 & 0.047<\phi\leq 7.953 \\
    1  & 7.953<\phi\leq 14\\
      0.08\phi^2-2.24\phi+16.68 & \phi > 14
    \end{cases},
\end{align*}
\end{small}

and 
\begin{small}
\begin{align*}
&\min_{\alpha\geq 0} \qty{ \cost\qty(y'_{4:3}(\phi),\alpha;1/\gamma) } = \min\qty{ \min_{\alpha\geq 0}\qty{\frac{1}{2}(3-\alpha/2)^2 + \frac{1}{2}(2.8+0.8\phi-\alpha)^2} ,  
 \min_{\alpha\geq 0}\qty{ \frac{1}{2}(2.8+0.8\phi-\alpha)^2 +1 } } \\
 &= \min\vast\{
 \tilde{f}_1(\phi) = \begin{cases} 
      0.32\phi^2+2.24\phi+8.42 & \phi < -5.375 \\
      0.064 \phi^2-0.512\phi+1.024 & \phi \geq -5.375
   \end{cases},
    \tilde{f}_2(\phi) =  \begin{cases} 
     0.32\phi^2+2.24\phi+ 4.92 &  \phi < -3.5 \\
     1 & \phi \geq  -3.5
   \end{cases}  \vast\} \\
 &=  \begin{cases} 
      0.32\phi^2+2.24\phi+ 4.92 &  \phi \leq -3.5 \\
      1 &  -3.5 < \phi \leq 0.047 \\
      0.064\phi^2-0.512\phi+1.024  &  0.047 < \phi \leq  7.953\\
    1  & \phi>7.953
    \end{cases}.
\end{align*}
\end{small} 

Therefore, 
\begin{align*}
C(\phi) =
 \begin{cases} 
      0.32\phi^2+2.24\phi+ 5.92 &  \phi \leq -3.5 \\
      2 &  -3.5 < \phi \leq 0.047 \\
      0.064\phi^2-0.512\phi+2.024  &  0.047 < \phi \leq  7.953\\
      2  &  7.953 <\phi \leq 14 \\
    0.08\phi^2-2.24\phi+17.68  & \phi>14
    \end{cases}.
\end{align*}

Moreover, according to \eqref{eq:cost_no_spike_thj},
\begin{align*}
C'(\phi) &= \min_{\alpha\geq 0} \left\{ \cost\qty(y'_{1:2}(\phi),u/0.5;\gamma) + \cost\qty(y'_{4:3}(\phi),\alpha;1/\gamma) \right\} \\
&= \min \Bigg\{  
 \min_{\alpha\geq 0}\qty{ \frac{1}{2}(8-\alpha/0.25)^2 + \frac{1}{2}(5.6-0.4\phi-\alpha/0.5)^2+ \frac{1}{2}(3-\alpha/2)^2 +  \frac{1}{2}(2.8+0.8\phi-\alpha)^2},\\
 & \min_{\alpha\geq 0}\qty{ \frac{1}{2}(8-\alpha/0.25)^2 + \frac{1}{2}(5.6-0.4\phi-\alpha/0.5)^2+  1 + \frac{1}{2}(2.8+0.8\phi-\alpha)^2}, \\
&\min_{\alpha\geq 0} \qty{\frac{1}{2}(5.6-0.4\phi-\alpha/0.25)^2 + 1 + \frac{1}{2}(3-\alpha)^2 +  \frac{1}{2}(2.8+0.8\phi-\alpha)^2}, \\
&\min_{\alpha\geq 0} \qty{\frac{1}{2}(5.6-0.4\phi-\alpha/0.25)^2 + 1 +  1 + \frac{1}{2}(2.8+0.8\phi-\alpha)^2}  \Bigg\} \\
 &= 0.4\phi^2 + 2 .
\end{align*}

Finally, to determine $\mathcal{S}$, we take the minimum of these two functions:
\begin{align*}
\min\left\{ C(\phi), C'(\phi)\right\} = \begin{cases} 
      0.32\phi^2+2.24\phi+ 5.92 &  \phi \leq -3.5 \; \; \text{Minimizer: } C(\phi) \\
   3  & -3.5 <\phi \leq -1.581 \; \; \text{Minimizer: } C(\phi) \\
   0.4\phi^2+2 &  -1.581 \leq \phi < 0.837 \; \; \text{Minimizer: } C'(\phi)\\
    0.064\phi^2-0.512\phi+3.024 &  0.837 < \phi \leq 7.953 \; \; \text{Minimizer: } C(\phi)\\
    3  & \phi \geq 7.953   <\phi \leq 14 \; \; \text{Minimizer: } C(\phi)\\
    3 0.08\phi^2-2.24\phi+17.68 & \phi > 14   \; \; \text{Minimizer: } C(\phi)
\end{cases}.
\end{align*}

According to \eqref{eq:S_in_phi}, $\mathcal{S} =(-\infty,-1.581) \cup [0.837,\infty)$ for this example.

%% file: appendix_a_11.tex

\subsection{Proof of Proposition~\ref{prop:ci}}
\label{appendix:ci}

We first present an auxiliary result.

\begin{lemma}[Lemma A.2. in \citet{Kivaranovic2020-ug}]
\label{lemma:trunc_normal_cdf}
Let $F_{\mu,\sigma^2}^{\mathcal{S}}$ denote the cumulative distribution function for a normal distribution with mean $\mu$ and variance $\sigma^2$, truncated to ths set $\mathcal{S}\subseteq \mathbb{R}$. For each $t\in \mathcal{S}$, $F^{\mathcal{S}}_{\mu, \sigma^2}(t)$ is continuous and monotonically decreasing in $\mu$ .
\end{lemma}

We now present the proof of Proposition~\ref{prop:ci}. 

According to Lemma~\ref{lemma:trunc_normal_cdf}, $F^{\mathcal{S}\cap(0,\infty)}_{\mu, \sigma^2\Vert\nu\Vert_2^2}(t)$ is a monotonically decreasing function of $\mu$ for each $t\in \mathcal{S}\cap(0,\infty)$. Since $\frac{\alpha}{2}<1-\frac{\alpha}{2}$, it follows that $\theta_L(t)$ and $\theta_U(t)$ defined in \eqref{eq:LCB_UCB} are unique, and that $\theta_L(t)< \theta_U(t)$.

In addition, monotonicity implies that $\forall t \in \mathcal{S}\cap(0,\infty)$,  (i) $\nu^{\top}c > \theta_L(t)$ if and only if $F^{\mathcal{S}\cap(0,\infty)}_{\nu^{\top}c, \sigma^2\Vert\nu\Vert_2^2}(t) < 1-\alpha/2$; and (ii) $\nu^{\top}c < \theta_U(t)$ if and only if $F^{\mathcal{S}\cap(0,\infty)}_{\nu^{\top}c, \sigma^2\Vert\nu\Vert_2^2}(t) > \alpha/2$. 

These two observations imply that
\begin{align}
\label{eq:ci_cdf_equiv}
\left\{\nu^{\top}c: \nu^{\top}c \in \left[\theta_L(t),\theta_U(t)\right] \right\} = \left\{\nu^{\top}c: \frac{\alpha}{2} \leq F^{\mathcal{S}\cap(0,\infty)}_{\nu^{\top}c, \sigma^2\Vert\nu\Vert_2^2}(t) \leq 1- \frac{\alpha}{2}   \right\}, \forall t \in \mathcal{S}\cap(0,\infty).
\end{align}

Recall that $Y\sim \mathcal{N}(c,\sigma^2 I)$. This implies that
\begin{align*}
&\mathbb{P} \qty( \nu^{\top}c \in \left[\theta_L(\nu^{\top} Y),\theta_U(\nu^{\top} Y)\right]  \;\middle\vert\; \thj\in\mathcal{M}(Y), \Pi_{\nu}^{\perp} Y = \Pi_{\nu}^{\perp} y, 
\nu^\top Y>0 ) \\
&\overset{a.}{=} \mathbb{P} \qty( \frac{\alpha}{2} \leq F^{\mathcal{S}\cap(0,\infty)}_{\nu^{\top}c, \sigma^2\Vert\nu\Vert_2^2}(\nu^{\top}Y) \leq 1- \frac{\alpha}{2}   \;\middle\vert\; \thj\in\mathcal{M}(Y), \Pi_{\nu}^{\perp} Y = \Pi_{\nu}^{\perp} y, 
\nu^\top Y>0 ) \\
&\overset{b.}{=} \mathbb{P}\qty(F^{\mathcal{S}\cap(0,\infty)}_{\nu^{\top}c, \sigma^2\Vert\nu\Vert_2^2}\qty(Z)\in \qty[\frac{\alpha}{2},1-\frac{\alpha}{2}]) \\
&\overset{c.}{=} 1-\alpha.
\end{align*}

To prove $a.$, we note that \eqref{eq:ci_cdf_equiv} holds for all  $t \in \mathcal{S}\cap(0,\infty)$; therefore it holds for $\nu^\top Y$ conditioning on $\qty{\thj\in\mathcal{M}(Y), \Pi_{\nu}^{\perp} Y = \Pi_{\nu}^{\perp} y, \nu^\top Y>0 }$ as well. Step $b.$ follows from Proposition~\ref{prop:pval} and letting $Z$ denote a normal random variable with mean $\nu^\top c$ and variance $\sigma^2||\nu||_2^2$, truncated to the set $\mathcal{S}\cap(0,\infty)$. The last step follows from the probability integral transform, which states that for a continuous random variable $X$, $F_X(X)$ is distributed as a Uniform(0,1) distribution.

%% file: appendix_a_12.tex

\subsection{Additional information for data analysis in Section~\ref{section:real_data}}
\label{appendix:sensitivity_real_data}

Data for the \texttt{spikefinder} challenge are available for download at \\
\texttt{https://s3.amazonaws.com/neuro.datasets/challenges/spikefinder/spikefinder.train.zip}.

In what follows, we reproduce Figure~\ref{fig:perm_test_cor_vp_chen_cells} with different choices of $h$ (defined in \eqref{eq:nu_def}). 
In Figures~\ref{fig:sensitivity_hcor_vp_chen_cells} and \ref{fig:sensitivity_hcor_vp_chen_cells_large_h}, we compare the accuracy --- as measured by the Victor-Purpura distance and correlation --- of the spikes estimated via \eqref{eq:l0-opt-intercept} (in orange), as well as the subset of spikes estimated via  \eqref{eq:l0-opt-intercept} for which the $p$-value is below $0.05$ (in blue). The black lines
indicate the 2.5\% and 97.5\% quantiles of the accuracy measures obtained over 1,000 resampled datasets, where each resampled dataset contains a subset of the estimated spikes from \eqref{eq:l0-opt-intercept}; details are as in 
Section~\ref{section:result_all_cell}. 
The results using $h=5$ and $h=50$ are quite similar to those with $h=20$ (see Figure~\ref{fig:perm_test_cor_vp_chen_cells}):  the subset of spikes estimated via  \eqref{eq:l0-opt-intercept} for which the $p$-value is below $0.05$  is the most accurate in almost every recording. 

In addition, we performed simple diagnostics of the normality assumption of the error terms in \eqref{eq:obs-model}. In Figure~\ref{fig:real_data_residual}, we plot the residuals ($y_t-\hat{c}_t$) for recordings from the \cite{Chen2013-ha} dataset; for most recordings, residuals appear approximately normal.

\subsection{Estimation of the error variance $\sigma^2$ in \eqref{eq:obs-model}}
\label{appendix:estimated_sigma}

Throughout the paper, we have assumed that $\sigma^2$ in \eqref{eq:obs-model} is known. However, if it is unknown, we propose to use $\hat\sigma^2 = \frac{1}{T-1}\sum_{t=1}^T\left(y_t - \hat{c}_t \right)^2$ as an estimator for $\sigma^2$ in evaluating the $p$-value in \eqref{eq:pval}. In Figure~\ref{fig:sensitivity_estimated_sigma}, we present the results of a simulation study using the estimator $\hat\sigma^2$ and demonstrate that it leads to (i) adequate selective Type I error control under the global null (see Figure~\ref{fig:sensitivity_estimated_sigma}(a)), (ii) substantial power under the alternative (see Figure~\ref{fig:sensitivity_estimated_sigma}(b)), and (iii) correct selective coverage of the parameter $\nu^\top c$ (see Figure~\ref{fig:sensitivity_estimated_sigma}(c)).

%% file: appendix_figure.tex

\begin{figure}[htbp!] 
  \centering
  \includegraphics[width=0.5\linewidth]{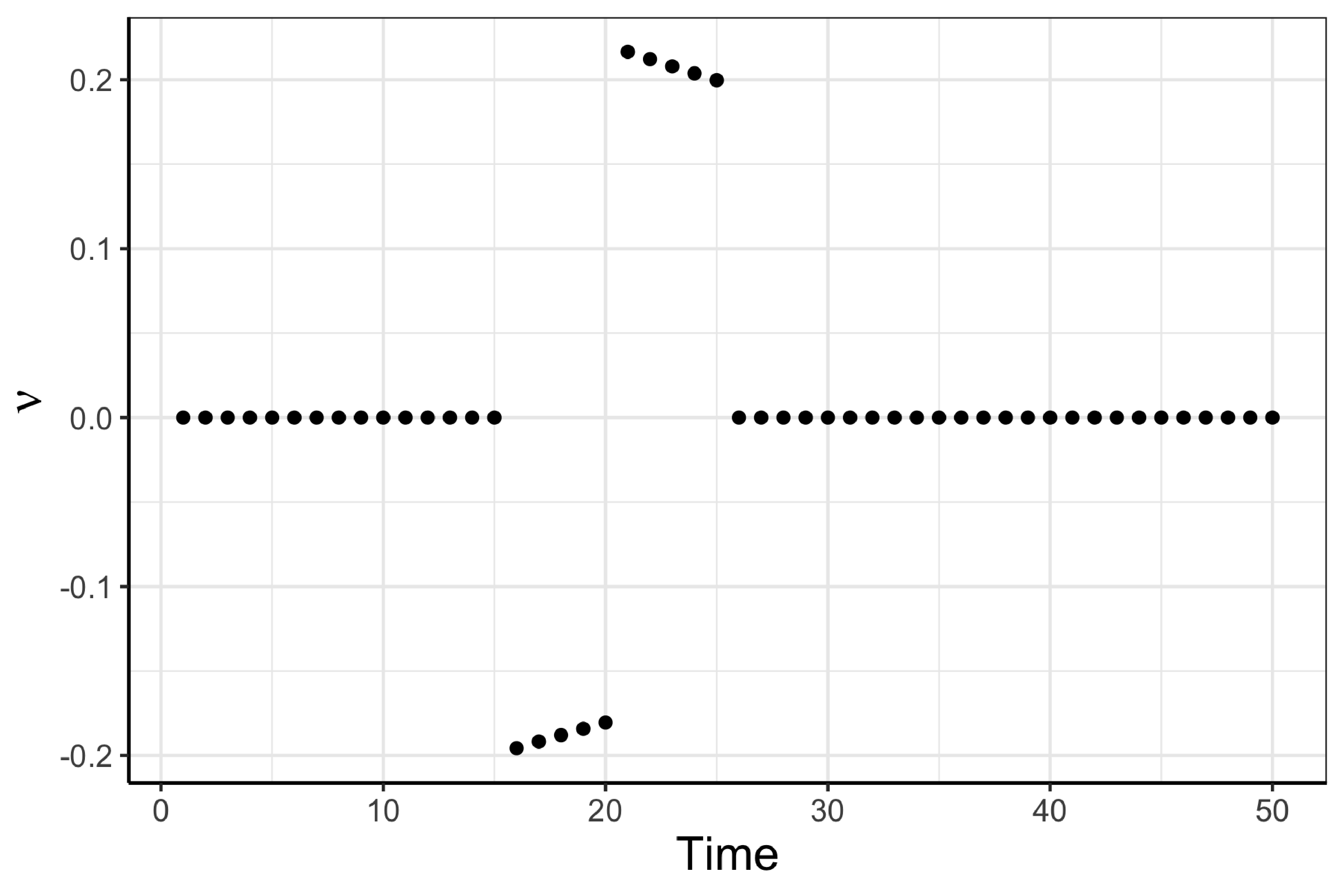}
  \caption{Plot of the contrast $\nu$ generated according to \eqref{eq:nu_def}, with $T=50$, $\gamma=0.98$, $\hat\tau_j=20$, and $h=5$. }
\label{fig:plot_nu_vec}
\end{figure}

\begin{figure}[htbp!]
  \centering
  \includegraphics[width=0.6\linewidth]{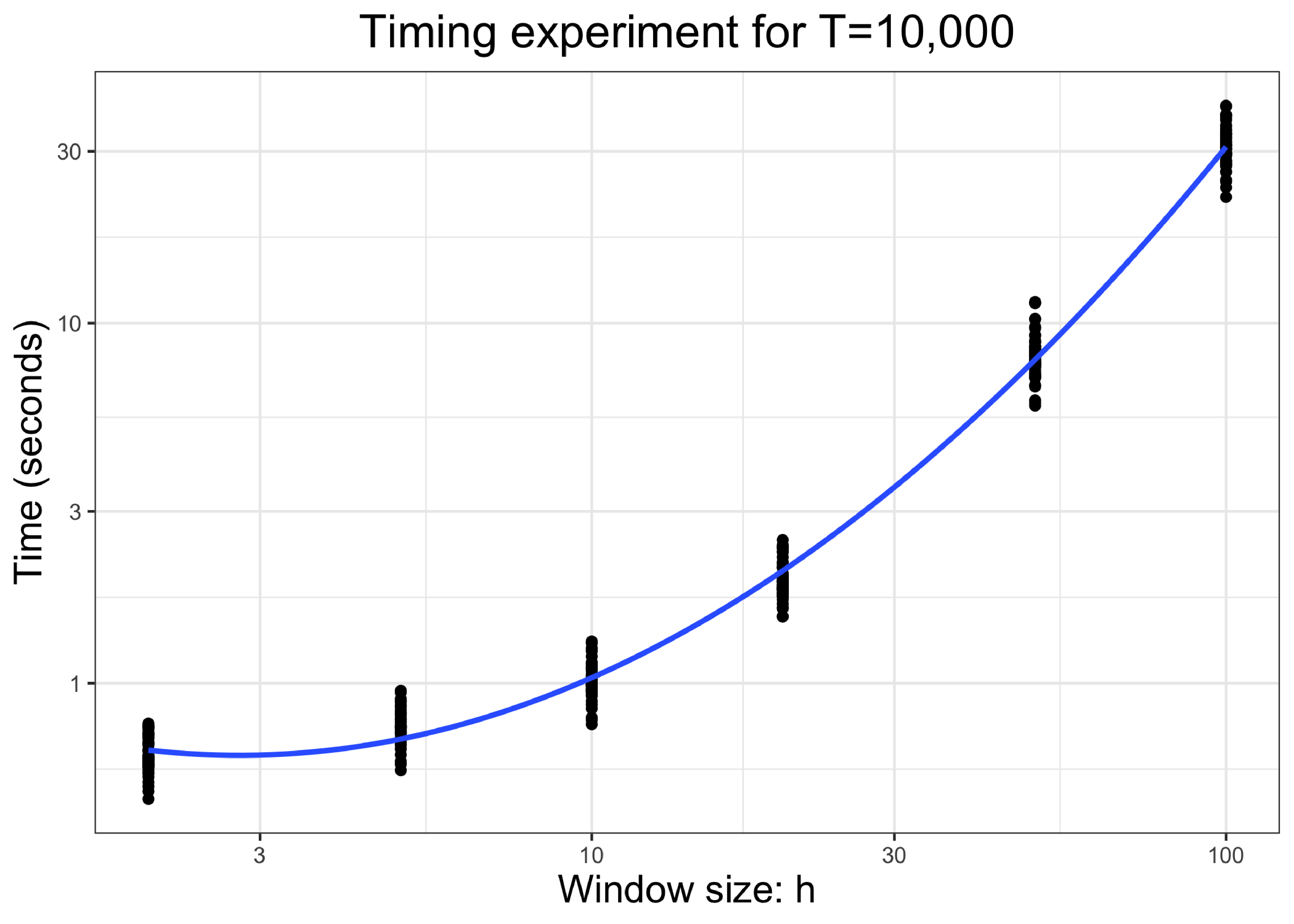}
\caption{Running time of Algorithm~\ref{alg:S_computation} over 50 replicate datasets, as a function of the window size, $h$. Each point represents a separate dataset. Each dataset is simulated according to \eqref{eq:obs-model}, and the $\ell_0$ problem is solved with $\lambda=0.3$. A quadratic equation ($\text{Time} = 0.003h^2-0.002 h + 0.695$) is plotted for reference. }
\label{fig:supp_time}
\end{figure}

\begin{figure}[htbp!]
\begin{centering}
\hspace{15mm}(a)  \hspace{5mm}\\
\end{centering}
\begin{subfigure}{\textwidth}
  \centering
  \includegraphics[width=\linewidth]{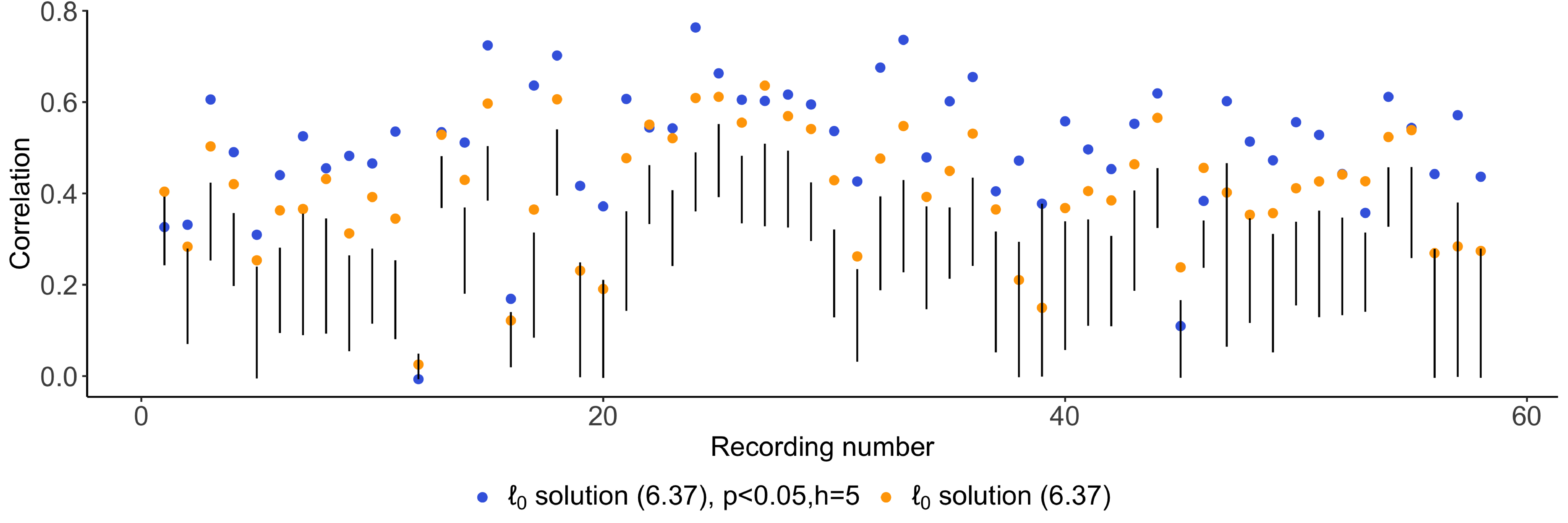}
\end{subfigure}

\begin{centering}
\hspace{15mm}(b)  \hspace{5mm}\\
\end{centering}
\begin{subfigure}{\textwidth}
  \centering
  \includegraphics[width=\linewidth]{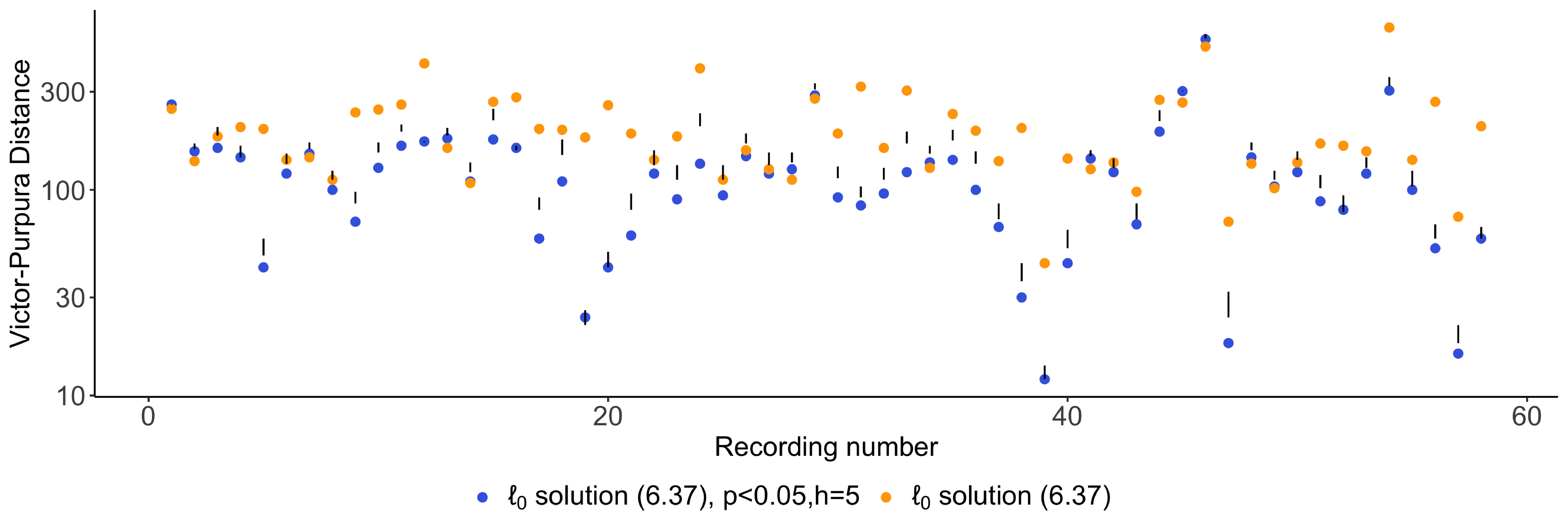}\end{subfigure}

\caption{Results for the \citet{Chen2013-ha} dataset. Details are as in Figure~\ref{fig:perm_test_cor_vp_chen_cells} but with $h=5$.
 }
\label{fig:sensitivity_hcor_vp_chen_cells}
\end{figure}

\begin{figure}[htbp!]
\begin{centering}
\hspace{15mm}(a)  \hspace{5mm}\\
\end{centering}
\begin{subfigure}{\textwidth}
  \centering
  \includegraphics[width=\linewidth]{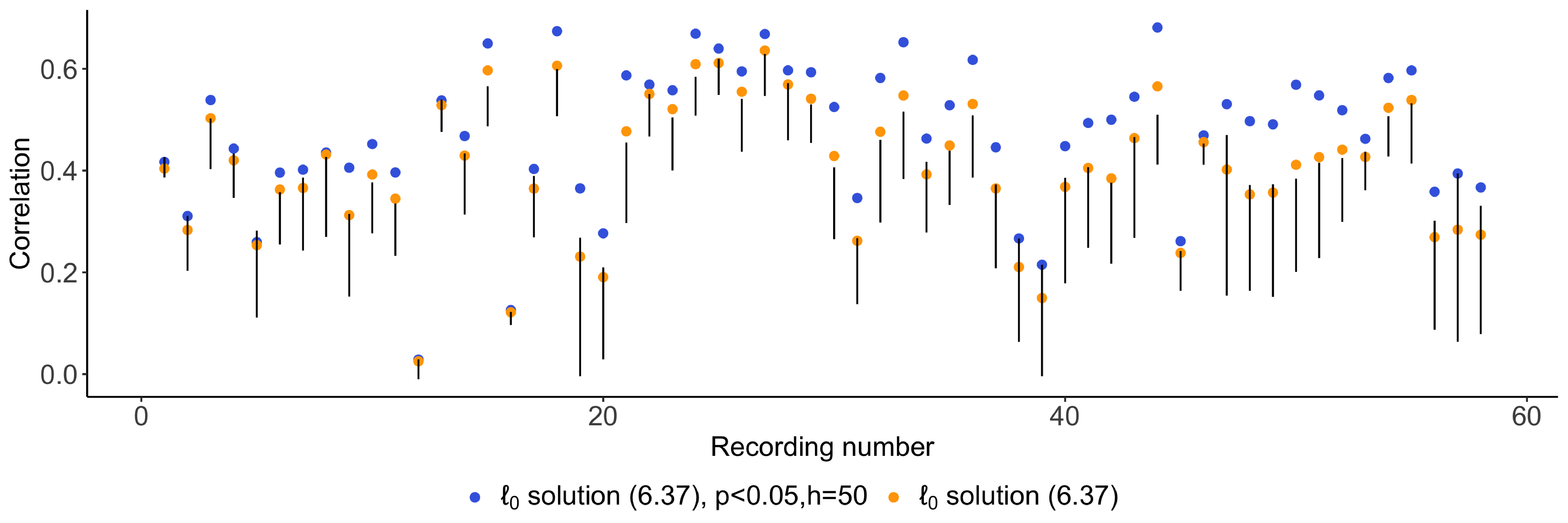}\end{subfigure}

\begin{centering}
\hspace{15mm}(b)  \hspace{5mm}\\
\end{centering}
\begin{subfigure}{\textwidth}
  \centering
  \includegraphics[width=\linewidth]{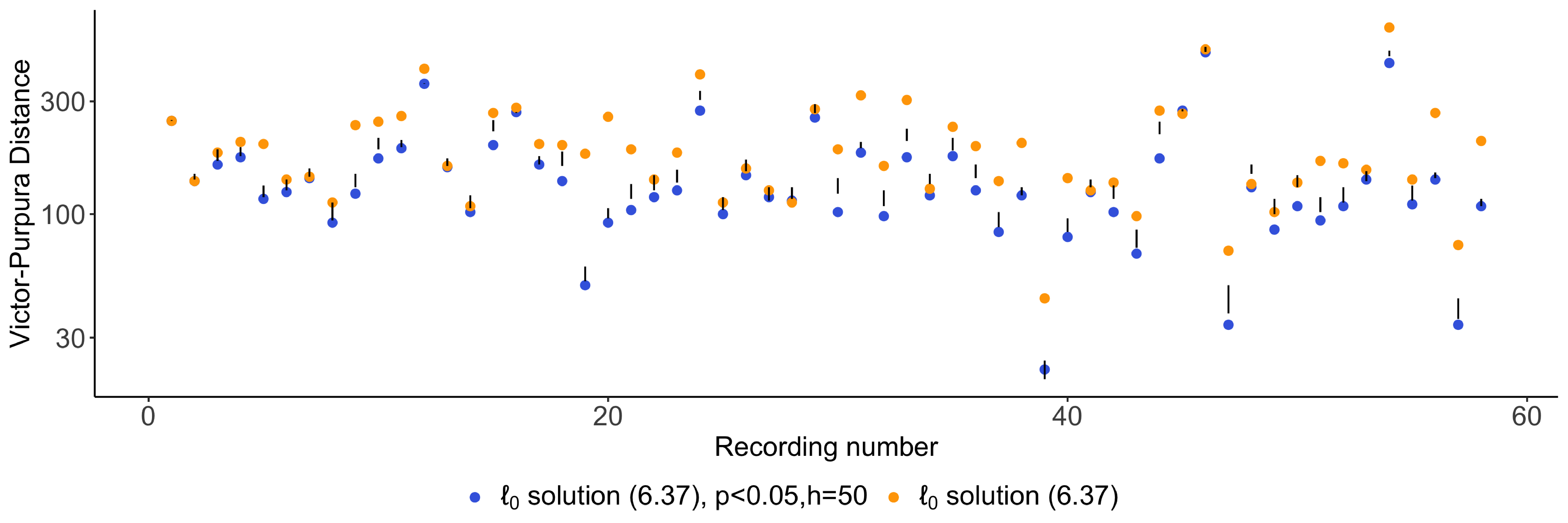}\end{subfigure}

\caption{Results for the \citet{Chen2013-ha} dataset. Details are as in Figure~\ref{fig:perm_test_cor_vp_chen_cells} but with $h=50$. }
\label{fig:sensitivity_hcor_vp_chen_cells_large_h}
\end{figure}

\begin{figure}[htbp!] 
  \centering
  \includegraphics[width=\linewidth]{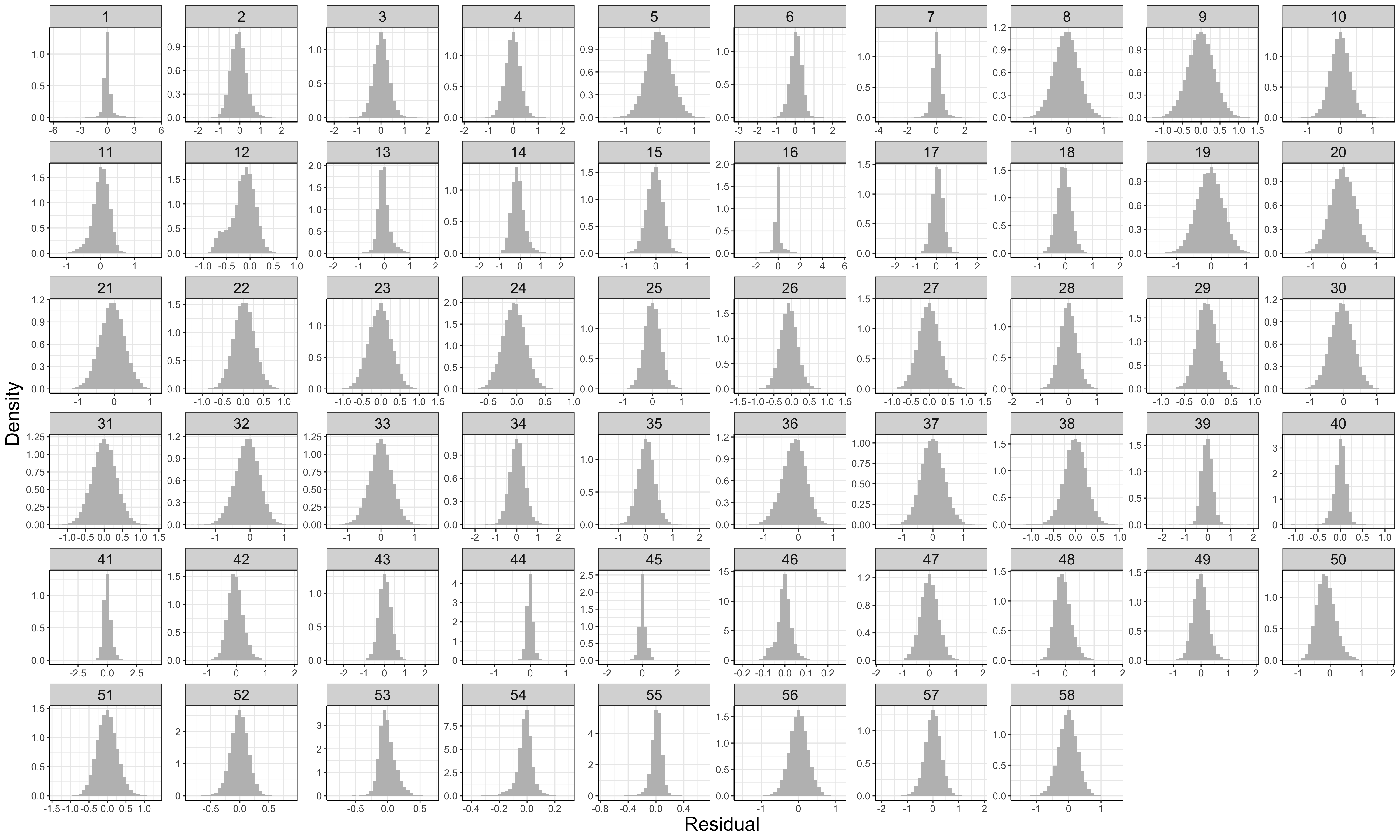}
  \caption{Residuals, $y_t-\hat{c}_t$, for recordings from the \cite{Chen2013-ha} dataset, where $\hat{c}_t$ is the solution to \eqref{eq:l0-opt-intercept}.}
\label{fig:real_data_residual}
\end{figure}

\begin{figure}[htbp!]
\begin{centering}
\hspace{15mm} (a) \hspace{40mm} (b) \hspace{40mm} (c)\\
\end{centering}
\begin{subfigure}{\textwidth}
  \centering
  \includegraphics[width=.3\linewidth]{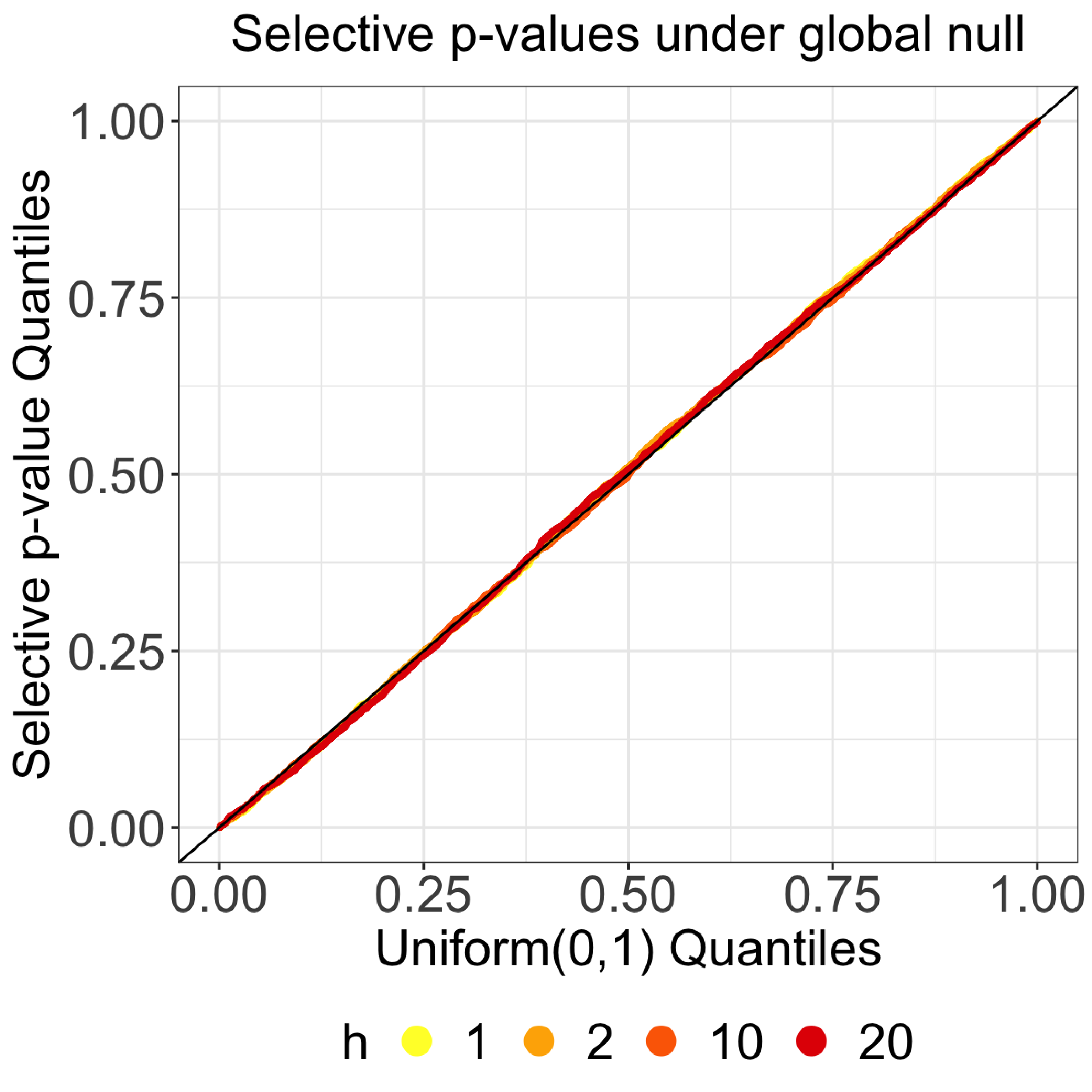}
  \includegraphics[width=.3\linewidth]{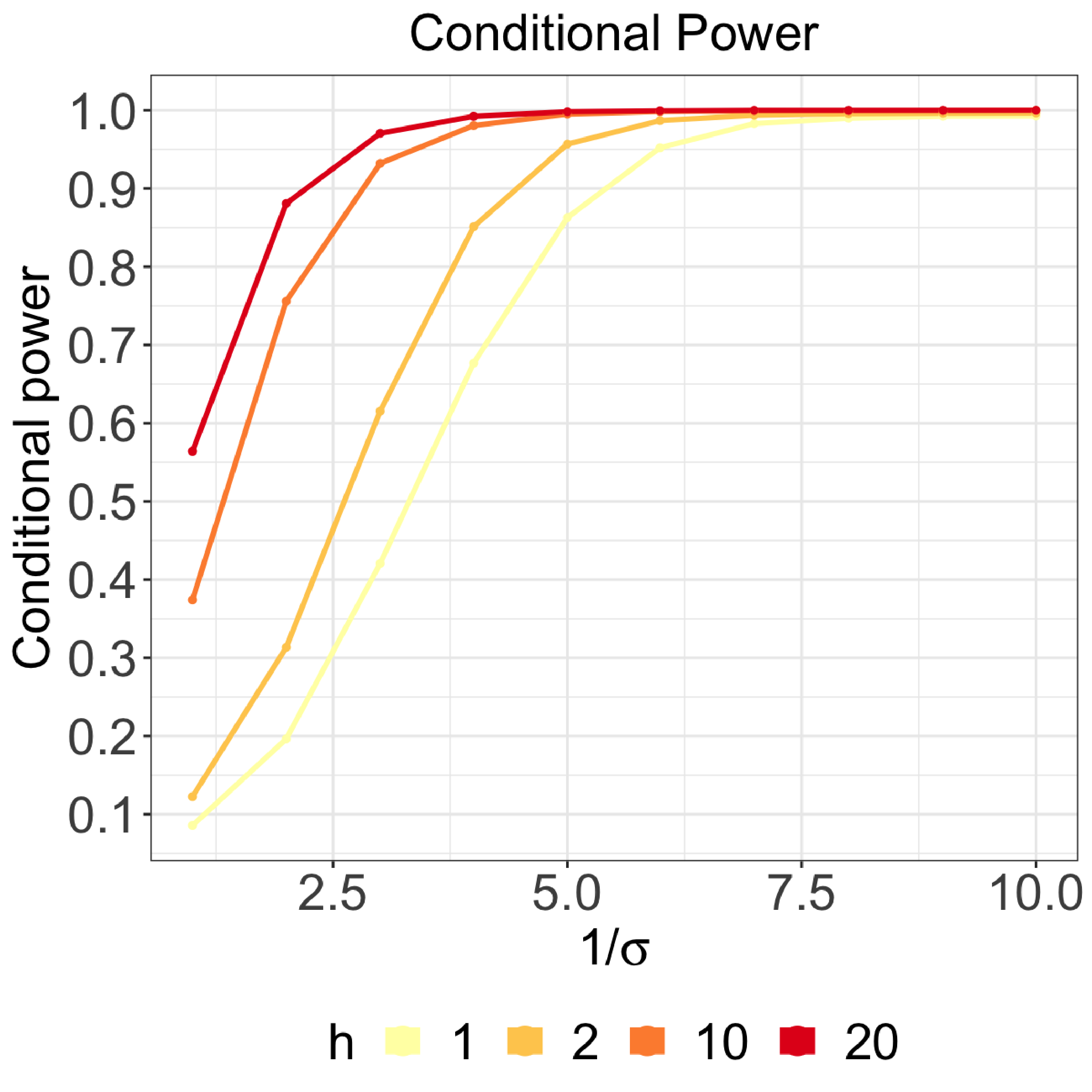}
  \includegraphics[width=.3\linewidth]{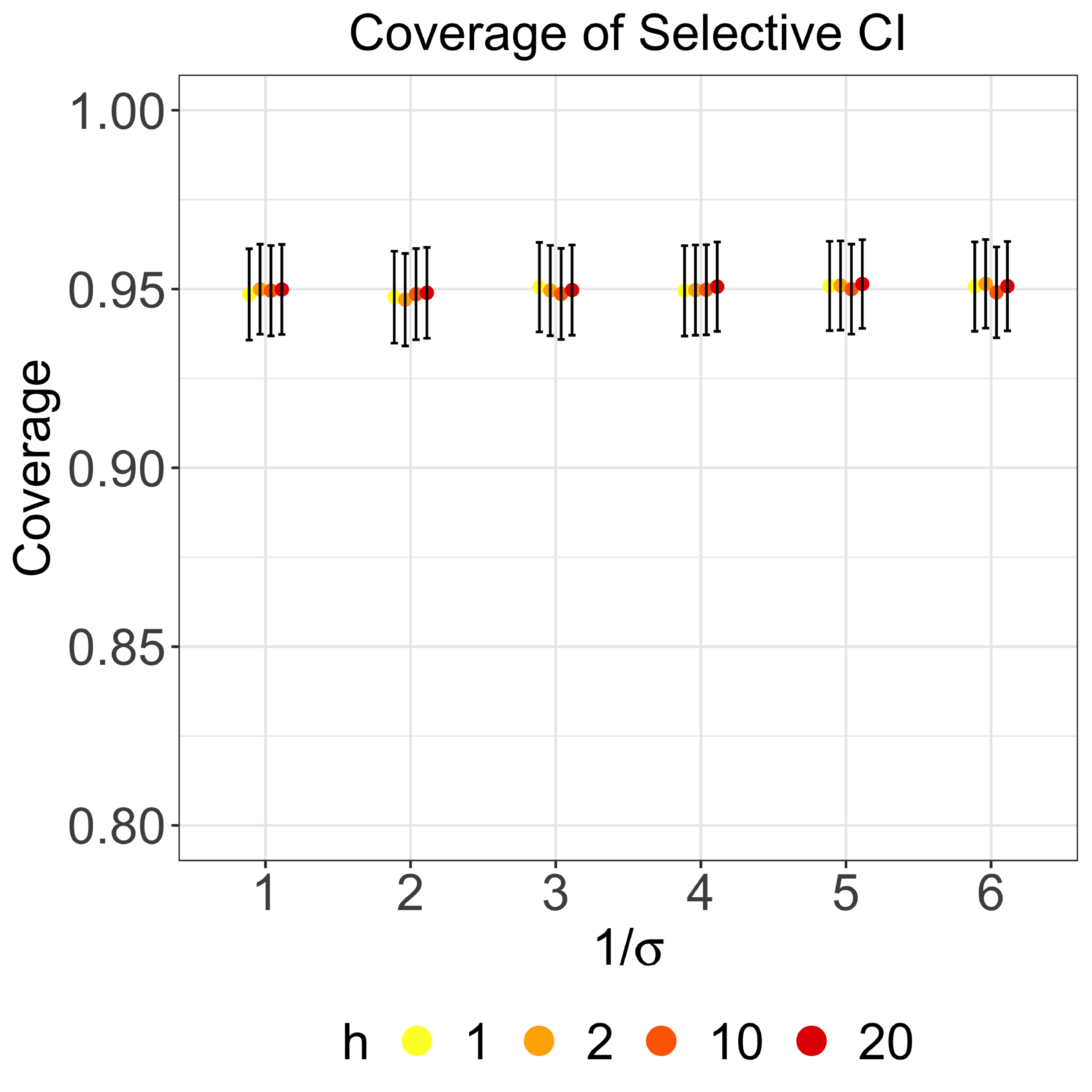}
\end{subfigure}
\caption{\textit{(a): } Quantile-quantile plot for selective $p$-values computed using estimated variance $\hat{\sigma}^2$ based on 100 simulations (2,988 hypothesis tests) under the global null. \textit{(b): }Conditional power for selective $p$-values with estimated variance $\hat{\sigma}^2$. \textit{(c): }Selective confidence intervals computed using estimated variance $\hat{\sigma}^2$ achieve correct nominal coverage (95\% coverage at level $\alpha=0.05$) across all values of $h$ and $\sigma$.}
\label{fig:sensitivity_estimated_sigma}
\end{figure}